\documentclass[runningheads]{llncs}
\usepackage[T1]{fontenc}
\usepackage{graphicx}
\usepackage[colorlinks=true, allcolors=blue]{hyperref}
\usepackage{color}

\usepackage[dvipsnames]{xcolor}
\usepackage{amsmath}
\usepackage{amssymb}
\usepackage{stmaryrd}
\usepackage{mathtools}
\usepackage{mdframed}
\usepackage{booktabs}
\usepackage{tikz}
\usepackage{csquotes}
\usetikzlibrary{patterns,shapes,positioning,automata}
\usepackage[altpo]{backnaur}
\usepackage[shortlabels]{enumitem}
\usepackage{multicol}
\usepackage{adjustbox}
\usepackage{wrapfig}
\usepackage{nicefrac}
\usepackage{siunitx}
\usepackage{colortbl}
\usepackage{listings}
\usepackage{cleveref}
\usepackage{thmtools}
\usepackage{thm-restate}
\usepackage{bussproofs}
\usepackage{xspace}
\usepackage{pifont}
\usepackage[textsize=scriptsize,bordercolor=white]{todonotes}
\usepackage{setspace}
\newcommand{\blue}[1]{\textcolor{Blue}{#1}}
\newcommand{\red}[1]{\textcolor{Bittersweet}{#1}}
\newcommand{\gray}[1]{\textcolor{gray}{#1}}

\newcommand{\weight}[1]{|#1|}
\newcommand{\subdists}{{\Deltasub \Sigma}}

\newcommand{\dists}{{\Delta \Sigma}}
\newcommand{\onlybdists}{\langle \bool \rangle}
\newcommand{\Deltasub}{\Delta_{\leq 1}}

\newcommand{\Deltaspec}[1]{\Delta_{#1}}

\newcommand{\LFPset}{X}
\newcommand{\Mfdr}{{\normalfont \initSet_{\textsc{fdr}}}}
\newcommand{\varsfdr}{\vars_{\textsc{fdr}}}
\newcommand{\Mfldr}{{\normalfont \initSet_{\textsc{fldr}}}}

\newcommand{\initDist}{\mu_0} 
\newcommand{\initSet}{M} 
\newcommand{\hSetOne}{M} 
\newcommand{\hSetTwo}{N}
\newcommand{\hSetThree}{O}
\newcommand{\hSetFour}{P}
\newcommand{\triple}[3]{\{\hspace{-2.1pt}|#1|\hspace{-2.1pt}\}~#2~\{\hspace{-2.1pt}|#3|\hspace{-2.1pt}\}} 
\newcommand{\tripleGen}{\triple{\hSetOne}{\prog}{\hSetTwo}} 
\newcommand{\inv}{I}
\newcommand{\marg}[2]{\mathbb{M}_{#1}(#2)}

\newcommand{\unif}[3]{\mathsf{unif}_{#1}(#2,\dots,#3)}

\newcommand{\prosim}[4]{\textnormal{Pr}_{#1}[#2 #3 #4]}
\newcommand{\pr}[2]{\textnormal{Pr}_{#1}[#2]}
\newcommand{\minim}[2]{\min\{#1, #2\}}
\newcommand{\dirac}[1]{\delta_{#1}} 

\newcommand{\prog}{C}
\newcommand{\loopbody}{B}

\newcommand{\progOne}{\prog_1}
\newcommand{\progTwo}{\prog_2}
\newcommand{\arith}{E}
\newcommand{\bool}{\varphi}
\newcommand{\probability}{p}
\newcommand{\textvar}[1]{\normalfont #1}
\newcommand{\progvar}[1]{#1}
\newcommand{\numb}[1]{\mathbf{#1}}
\newcommand{\supp}{\textnormal{supp}}
\newcommand{\pgcl}{\textnormal{pGCL}}
\newcommand{\conf}{\textnormal{Conf}}

\newcommand{\vars}{\mathit{Var}}
\newcommand{\false}{\mathsf{ff}}
\newcommand{\true}{\mathsf{tt}}
\newcommand{\bools}{\mathbb{B}}
\newcommand{\Ifdr}{{\normalfont I_{\textsc{fdr}}}}
\newcommand{\Ifldr}{{\normalfont I_{\textsc{fldr}}}}
\newcommand{\progfldr}{{\normalfont\prog_\textsc{fldr}}}
\newcommand{\progfdr}{{\normalfont\prog_\textsc{fdr}}}
\newcommand{\prob}[2]{\textnormal{prob}_{#1}(#2)}
\newcommand{\probt}[2]{\textnormal{probt}_{#1}(#2)}

\newcommand{\lfp}{\mathsf{lfp\,}}
\newcommand{\nats}{\mathbb{N}}

\newcommand{\unmod}[1]{\textnormal{unmod}(#1)}

\newcommand{\reach}[1]{\mathsf{reach}\llbracket #1\rrbracket}
\newcommand{\post}[1]{\mathsf{sp}\llbracket #1\rrbracket}
\newcommand{\den}[1]{\mathsf{D}\llbracket #1\rrbracket}

\newcommand{\reachn}{\mathsf{reach}}
\newcommand{\postn}{\mathsf{sp}}
\newcommand{\denn}{\mathsf{D}}

\newcommand{\prfont}[1]{\textnormal{\texttt{#1}}}
\newcommand{\prspace}{\,}
\newcommand{\prparenop}{\prfont{(}}
\newcommand{\prparencl}{\prfont{)}}
\newcommand{\prsqbrackop}{\prfont{[}}
\newcommand{\prsqbrackcl}{\prfont{]}}
\newcommand{\prbraceop}{\prfont{\{}}
\newcommand{\prbracecl}{\prfont{\}}}
\newcommand{\ski}{\prfont{skip}}
\newcommand{\ass}[2]{#1\prspace\prfont{:=}\prspace#2}
\newcommand{\pchoice}[3]{\prbraceop#1\prbracecl \prspace \prsqbrackop#2\prsqbrackcl \prspace \prbraceop#3\prbracecl}
\newcommand{\nchoice}[2]{\{#1\} \prsqbrackop\prsqbrackcl \{#2\}}
\newcommand{\comp}[2]{#1\prspace\prfont{;}\prspace#2}
\newcommand{\SEMICLN}{\prspace\prfont{;}\prspace}
\newcommand{\ifelse}[3]{\prfont{if}\prspace\prparenop#1\prparencl\prspace\prbraceop#2\prbracecl\prspace\prfont{else}\prspace\prbraceop#3\prbracecl}
\newcommand{\ifelseskip}[2]{\prfont{if}\prspace\prparenop#1\prparencl\prspace\prbraceop#2\prbracecl}
\newcommand{\ELSE}{\prbracecl\prspace\prfont{else}\prspace\prbraceop}
\newcommand{\IF}[1]{\prfont{if}\prspace\prparenop#1\prparencl\prspace\prbraceop}
\newcommand{\while}[2]{\prfont{while}\prspace\prparenop#1\prparencl\prspace\prbraceop#2\prbracecl}
\newcommand{\WHILE}[1]{\prfont{while}\prspace\prparenop#1\prparencl\prspace\prbraceop}

\renewcommand{\H}[2]{H[#1,#2]}
\newcommand{\bound}[2]{\phi_{#1}(#2)}
\newcommand{\isThere}[2]{\Theta_{#1}(#2)}

\newcommand{\clos}{c\ell}
\newcommand{\closOf}[1]{{\clos\left(#1\right)}}

\newcommand{\mdp}{\mathcal{M}}
\newcommand{\mdpTr}{\mdp_{\prog}}

\newcommand{\states}{S}
\newcommand{\act}{{Act}}
\newcommand{\prmdp}{P}

\newcommand{\subdistsStates}{\Deltasub \states}

\newcommand{\smallstep}[2]{\overset{#1.#2}{\longrightarrow}}

\newcommand{\markc}{\mathcal{M}}
\newcommand{\markctuple}{(\states,\prmdp)}
\newcommand{\markcTr}{\markc_{\prog}}

\newcommand{\hBrace}[1]{\blue{\{\hspace{-2.1pt}|#1|\hspace{-2.1pt}\}}}
\newcommand{\htrip}[3]{\hBrace{#1}~#2~\hBrace{#3}}
\newcommand{\intropre}{pre}
\newcommand{\intropost}{post}

\newcommand{\splitStateInline}[2]{\tikz[baseline=-2.75pt] \node [circle split,draw,rotate=90,inner sep=0.5pt] {\rotatebox{-90}{\tiny$#1$} \nodepart{lower} \rotatebox{-90}{\tiny$#2$}};}

\newcommand{\progAnno}[1]{\blue{\prfont{/\!\!/}~#1}}
\newcommand{\progAnnoNoSlashes}[1]{\blue{\textcolor{white}{\prfont{/\!\!/}}~#1}}

\newcommand{\eeq}{~{}={}~}
\newcommand{\iiff}{~{}\iff{}~}
\newcommand{\lland}{~{}\land{}~}

\newcommand{\toolzar}{\textsf{Zar}\xspace}
\newcommand{\toolpsi}{\textsf{PSI}\xspace}

\newcommand{\wpRelWork}{\textsf{wp}}

\newif\ifappendixversion
 \appendixversiontrue     
\begin{document}
\title{Verifying Sampling Algorithms via\\Distributional Invariants}
%
%
\author{
    Daniel Zilken\inst{1}\orcidID{0009-0007-7221-2554}\and%
    Kevin Batz\inst{2}\orcidID{0000-0001-8705-2564}\and%
    Joost-Pieter Katoen\inst{1}\orcidID{0000-0002-6143-1926}\and%
    Tobias Winkler\inst{1}\orcidID{0000-0003-1084-6408}%
}
\authorrunning{K. Batz et al.}
%
\institute{%
    RWTH Aachen University, Aachen, Germany\and%
    University College London, London, United Kingdom
}
\maketitle
\begin{abstract}
    This paper presents a Hoare-like verification framework for discrete probabilistic programs that we apply to two non-trivial sampling algorithms: Lumbroso's \textit{Fast Dice Roller} and Saad et al.'s \textit{Fast Loaded Dice Roller}.
    These algorithms have previously resisted formal verification due to their probabilistic nature, intricate loop structure, and parametric input.
    Our approach complements existing proof rules based on inductive distributional invariants, enabling us to verify both total and partial correctness of the two algorithms.
\end{abstract}
\section{Introduction}%
\label{sec:intro}%
\emph{Sampling random numbers} from prescribed probability distributions is fundamental to a growing number of applications in computer science and related fields.
Examples include simulations of physical processes \cite{doi:10.1073/pnas.1912789117}, randomized algorithms and data structures~\cite{DBLP:journals/cacm/Pugh90}, and machine learning~\cite{sgd}.

This paper addresses \emph{functional correctness} verification of sampling algorithms encoded as \emph{probabilistic programs}.
The goal is to ensure that such a program generates a specified---possibly parameterized---target distribution.
More precisely, the verification problem we address is as follows:
\begin{quote}
    \itshape
        Given an imperative program $\prog$ with coin flips and a distributional Hoare triple
        $\htrip{\intropre}{\prog}{\intropost}$,
        prove that $\prog$, when executed on any input (distribution) satisfying \blue{$\intropre$}, yields a (sub-)distribution satisfying \blue{$\intropost$}.
\end{quote}
In our framework, \blue{$\intropre$} and \blue{$\intropost$} are \emph{sets of (sub-)distributions} or predicates describing such sets; a (sub-)distribution $\mu$ \emph{satisfies} \blue{$\intropre$} or \blue{$\intropost$} if $\mu$ is an element of the respective set.
Hoare triples enable reasoning about samplers for \emph{parametric distributions}.
For example, $\htrip{n \geq 0}{\prog}{y \sim \unif{}{0}{n}}$
asserts that, given input $n \geq 0$, program $\prog$ outputs an integer $y$ drawn uniformly at random from $\{0, \dots, n\}$ (assuming $\prog$ does not modify $n$).

Our problem statement poses three main challenges:
\begin{enumerate}[(C1),leftmargin=9mm]
    \item\label{challenge1} The output distribution must be established \emph{exactly}---approximations do not necessarily suffice.
    \item\label{challenge2} When coin flips occur inside \emph{loops}, the target distribution may only be achieved \emph{in the limit}.
    \item\label{challenge3} Sampling algorithms may take \emph{parameters} as inputs, requiring verification over \emph{infinite families of distributions}.
\end{enumerate}
We tackle these challenges by adopting the \emph{distributional paradigm}~\cite{DBLP:conf/lics/AkshayGV18,DBLP:conf/cav/AkshayCMZ23,DBLP:conf/ijcai/0001CMZ24,DBLP:conf/birthday/AghamovBKNOPV25}, which treats distributions over program states as first-class semantic objects.
This point of view enables reasoning about \emph{distributional invariant properties} such as \enquote{In every step, variable $x$ is non-negative with probability $\geq 0.9$}.
Such properties are inexpressible in the more conventional state-centric approach to probabilistic verification (e.g.,~\cite[Chapter 10]{DBLP:books/daglib/0020348}).
Distributional invariants turn out to be a natural fit for verifying sampling algorithms, as we illustrate next.

\paragraph{Sampler Verification with Distributional Invariants: Von Neumann's Fair Coin.}%
Consider the problem of sampling a uniform random bit from a source of i.i.d.\ random bits with unknown bias $p \in (0,1)$ towards, say, $0$.
The following solution is due to Von Neumann~\cite{vonNeumann1951}:
Extract two bits from the source; if they differ, return the first (or the second), otherwise discard the bits and start over.
This is readily modeled as the program in \Cref{fig:vonNeumann}.
Albeit simple, the algorithm already presents challenges \ref{challenge1} and \ref{challenge2}.
Under the distributional paradigm, the program's dynamics is the following infinite trajectory of \emph{reachable distributions}, illustrated here for $p = 0.4$ and the initial distribution satisfying $Pr [\splitStateInline{0}{0}] = 1$:%
\footnote{%
    Some numbers have been rounded.
    The rightmost vector is the limiting distribution.
}%

\begin{figure}[t]
    \begin{minipage}{0.40\textwidth}
        \setlength{\jot}{2pt} 
        \begin{align*}
            &\progAnno{\hBrace{x = y}} \\
            &\WHILE{x = y} \\
            &\qquad \pchoice{\prspace\ass{x}{0}\prspace}{p}{\prspace\ass{x}{1}\prspace} \SEMICLN \\
            &\qquad \pchoice{\prspace\ass{y}{0}\prspace}{p}{\prspace\ass{y}{1}\prspace} \\
            &\prbracecl \\
            &\progAnno{\hBrace{x \sim \mathsf{unif}(0,1)}}
        \end{align*}
    \end{minipage}
    \begin{minipage}{0.59\textwidth}
        \centering
        \begin{tikzpicture}[scale=1.25,initial where=above,initial text=]
            \node [initial,circle split,draw,rotate=90,inner sep=2pt] at (0,0.6) (00) {\rotatebox{-90}{$0$} \nodepart{lower} \rotatebox{-90}{$0$}};
            \node [initial,circle split,draw,rotate=90,inner sep=2pt] at (0,-0.6) (11) {\rotatebox{-90}{$1$} \nodepart{lower} \rotatebox{-90}{$1$}};
            \node [accepting,circle split,draw,rotate=90,inner sep=2pt] at (-2,0) (01) {\rotatebox{-90}{$0$} \nodepart{lower} \rotatebox{-90}{$1$}};
            \node [accepting,circle split,draw,rotate=90,inner sep=2pt] at (2,0) (10) {\rotatebox{-90}{$1$} \nodepart{lower} \rotatebox{-90}{$0$}};
            
            \draw[->] (00) edge[bend left=20] node[right] {$(1{-}p)^2$} (11);
            \draw[->] (11) edge[bend left=20] node[left] {$p^2$} (00);
            \draw[->] (00) edge[bend left=-10] node[above left,near start] {$p(1{-}p)$} (01);
            \draw[->] (11) edge[bend left=10] node[below left,near start] {$p(1{-}p)$} (01);
            \draw[->] (00) edge[bend left=10] node[above right,near start] {$(1{-}p)p$} (10);
            \draw[->] (11) edge[bend left=-10] node[below right,near start] {$(1{-}p)p$} (10);
            
            \draw[->] (01) edge[loop left] node[auto] {$1$} (01);
            \draw[->] (10) edge[loop right] node[auto] {$1$} (10);
            \draw[->] (00) edge[loop right] node[above] {$p^2$} (00);
            \draw[->] (11) edge[loop left] node[below] {$(1{-}p)^2$} (11);
        \end{tikzpicture}
    \end{minipage}
    \caption{%
        \textit{Left:} Probabilistic program modeling Von Neumann's algorithm~\cite{vonNeumann1951}, where $\pchoice{\prspace L\prspace}{p}{\prspace R\prspace}$ executes sub-programs $L$ and $R$ with probability $p$ and $1-p$, respectively.
        \textit{Right:} The program's underlying Markov chain (\Cref{def:mdpOfLoop}, \cpageref{def:mdpOfLoop}).
        The terminal states (doubly circled) have a self-loop, which is important for our theory to work.
    }
    \label{fig:vonNeumann}
\end{figure}

\begin{align*}
    \begin{matrix}
        \splitStateInline{0}{0} \\[-0.5pt]
        \splitStateInline{0}{1} \\[-0.5pt]
        \splitStateInline{1}{0} \\[-0.5pt]
        \splitStateInline{1}{1}
    \end{matrix}
    \quad
    &\begin{pmatrix}
        1 \\
        \red{0} \\
        \red{0} \\
        0
    \end{pmatrix}_{\!0}
    \!\rightsquigarrow\hspace{-3mm}
    &&\begin{pmatrix}
        0.16 \\
        \red{0.24} \\
        \red{0.24} \\
        0.36
    \end{pmatrix}_{\!1}
    \!\rightsquigarrow\hspace{-3mm}
    &&\begin{pmatrix}
        0.0832 \\
        \red{0.3648} \\
        \red{0.3648} \\
        0.1872
    \end{pmatrix}_{\!2}
    \!\rightsquigarrow
    \cdots
    \rightsquigarrow\hspace{-3mm}
    &&\begin{pmatrix}
        0.0004 \\
        \red{0.4993} \\
        \red{0.4993} \\
        0.0010
    \end{pmatrix}_{\!10}
    \!\!\rightsquigarrow
    \cdots\hspace{-3mm}
    &&\begin{pmatrix}
        0.0 \\
        \red{0.5} \\
        \red{0.5} \\
        0.0
    \end{pmatrix}_{\!\omega}
    \tag{$\dagger$}
    \label{eq:trajectory}
\end{align*}
As desired, the sequence converges (pointwise) to the uniform distribution over $\{ \splitStateInline{0}{1}, \splitStateInline{1}{0}\}$.
An elegant approach for proving this fact is to verify that the reachable distributions in \eqref{eq:trajectory} satisfy the following \emph{distributional invariant}:
\begin{align*}
    Pr [\splitStateInline{0}{1}]
    ~=~
    Pr [\splitStateInline{1}{0}]
    ~,
    \qquad
    \begin{minipage}{0.55\textwidth}
        i.e., \emph{the probability of being in \splitStateInline{0}{1} is \underline{alwa}y\underline{s} equal to the probability of being in \splitStateInline{1}{0}}.
    \end{minipage}
    \tag{$\star$}
    \label{eq:vonNeumannInv}
\end{align*}
Notice that the \red{2\textsuperscript{nd}} and \red{3\textsuperscript{rd}} entries of the vectors in \eqref{eq:trajectory} are consistent with \eqref{eq:vonNeumannInv}.
As we discuss in detail in \Cref{sec:verification}, this invariant property restricted to terminal states still holds \emph{in the limit}.
Together with the fact that the program terminates \emph{almost surely} (i.e., with probability $1$), correctness follows.%
\footnote{Termination proofs are not in our scope. See, e.g., \cite{10.1145/3158121,DBLP:journals/pacmpl/MajumdarS25}.}

\paragraph{Efficient Discrete Uniform Sampling: A \enquote{Fast Dice Roller}~{\normalfont\cite{DBLP:journals/corr/abs-1304-1916}.}}%
Consider the program in \Cref{fig:fdrIntro} (left) for efficient generation of the discrete uniform distribution from unbiased i.i.d.\ random bits.
The program depends on an integer input parameter $n > 0$ and thus poses all three challenges \ref{challenge1}--\ref{challenge3}.
Regarding \ref{challenge3}, we note that the FDR effectively describes an \emph{infinite family} of finite Markov chains.
Importantly, established \emph{probabilistic model checking} techniques (e.g.,~\cite{DBLP:books/daglib/0020348,DBLP:journals/sttt/HenselJKQV22,DBLP:journals/jcss/JungesK0W21}) are not applicable to infinite families---these techniques can only verify the algorithm for finitely many values of $n$~\cite{DBLP:journals/jar/MertensKQW25}.

Lumbroso~\cite{DBLP:journals/corr/abs-1304-1916} has already sketched a surprisingly simple distributional invariant for the FDR (\Cref{fig:fdrIntro}, right).
In \Cref{sec:fdr}, we provide the remaining formal details and outline a proof of inductivity, establishing functional correctness of the FDR relative to a formal denotational semantics.

\begin{figure}[t]
    \begin{minipage}{0.49\textwidth}
        \setlength{\jot}{2pt}
        \begin{align*}%
            &\progAnno{\hBrace{n > 0}} \\
            &\ass{v}{1} \SEMICLN \ass{c}{0} \SEMICLN \\
            &\WHILE{v < n} \\
            &\qquad \ass{v}{2v} \SEMICLN \\
            &\qquad \pchoice{\ass{c}{2c}}{0.5}{\ass{c}{2c+1}} \SEMICLN \\
            &\qquad \IF{c \geq n} \\
            &\qquad	\qquad \ass{v}{v-n} \SEMICLN \ass{c}{c-n} \\[-5pt]
            &\qquad \prbracecl \\[-5pt]
            &\prbracecl \\
            &\progAnno{\hBrace{c \sim \unif{}{0}{n-1}}}
        \end{align*}
    \end{minipage}
    \begin{minipage}{0.49\textwidth}
        \enquote{%
            Consider this statement, which is a loop invariant: $c$ is uniformly distributed over $\{0,\ldots,v-1\}$}\hfill\mbox{\cite[p.~4]{DBLP:journals/corr/abs-1304-1916}%
        }
        \vspace{5mm}\\
        \begin{tikzpicture}[scale=0.8,dot/.style = {circle, fill=black, minimum size=4pt, inner sep=0pt}, every loop/.style={min distance=10mm,out=225,in=315,looseness=5}, ->,scale=0.375]
            \node[dot] (-1) at (0,1) {};
            \node[dot] (1) at (0,0) {};
            \node[dot] (2) at (-2,-1) {};
            \node[dot] (3) at (-3,-2) {};
            \node[dot] (5) at (2,-1) {};
            \node[dot] (6) at (-1,-2) {};
            \node[dot] (8) at (1,-2) {};
            \node[] at (0,-5) {$n=3$};
            
            \draw (-1) -- (1);
            \draw (1) -- (2);
            \draw (1) to [bend right = 30] (5);
            \draw (2) -- (3);
            \draw (2) -- (6);
            \draw (5) -- (8);
            \draw (5) to [bend right = 30] (1);
            
            \path (3) edge[loop below] (3);
            \path (6) edge[loop below] (6);
            \path (8) edge[loop below] (8);
        \end{tikzpicture}
        \hfill
        \begin{tikzpicture}[scale=0.8, dot/.style = {circle, fill=black, minimum size=4pt, inner sep=0pt}, every loop/.style={min distance=10mm,out=225,in=315,looseness=5}, ->,scale=0.375]
            \node[dot] (-1) at (0,1) {};
            \node[dot] (1) at (0,0) {};
            \node[dot] (2) at (-2,-1) {};
            \node[dot] (3) at (-3,-2) {};
            \node[dot] (5) at (2,-1) {};
            \node[dot] (6) at (-1,-2) {};
            \node[dot] (8) at (1,-2) {};
            \node[dot] (9) at (2,-2) {};

            \node[dot] (10) at (-3.5,-3) {};
            \node[dot] (11) at (-2.5,-3) {};
            \node[dot] (12) at (-1.5,-3) {};
            \node[dot] (13) at (-0.5,-3) {};
            \node[dot] (14) at (0.5,-3) {};

            \node[dot] (15) at (4,-1.5) {};
            \node[dot] (16) at (6,-1.5) {};
            \node[dot] (17) at (8,-1.5) {};

            \node[dot] (20) at (3.5,-3) {};
            \node[dot] (21) at (4.5,-3) {};
            \node[dot] (22) at (5.5,-3) {};
            \node[dot] (23) at (6.5,-3) {};
            \node[dot] (24) at (7.5,-3) {};
            \node[dot] (25) at (8.5,-3) {};
            
            \node[] at (2,-6) {$n=5$};
            
            \draw (-1) -- (1);
            \draw (1) -- (2);
            \draw (1) to (5);
            \draw (2) -- (3);
            \draw (2) -- (6);
            \draw (5) -- (8);
            \draw (5) to (9);

            \draw (8) to [bend right = 30] (15);
            \draw (9) to [bend right = 30] (16);
            \draw (9) to [bend right = 30] (17);

            \draw (3) -- (10);
            \draw (3) -- (11);
            \draw (6) -- (12);
            \draw (6) -- (13);
            \draw (8) -- (14);

            \draw (15) -- (20);
            \draw (15) -- (21);
            \draw (16) -- (22);
            \draw (16) -- (23);
            \draw (17) -- (24);
            \draw (17) -- (25);
            
            \draw (25) to [bend right = 30] (1);
            
            \path (10) edge[loop below] (10);
            \path (11) edge[loop below] (11);
            \path (12) edge[loop below] (12);
            \path (13) edge[loop below] (13);
            \path (14) edge[loop below] (14);
            \path (20) edge[loop below] (20);
            \path (21) edge[loop below] (21);
            \path (22) edge[loop below] (22);
            \path (23) edge[loop below] (23);
            \path (24) edge[loop below] (24);
        \end{tikzpicture}
    \end{minipage}
    \caption{%
        \textit{Left:} The Fast Dice Roller (FDR) algorithm adapted from~\cite{DBLP:journals/corr/abs-1304-1916} for discrete uniform sampling.
        \textit{Right:} Lumbroso's distributional invariant~\cite{DBLP:journals/corr/abs-1304-1916} and two members of the FDR's infinite family of underlying Markov chains (\Cref{def:mdpOfLoop}, \cpageref{def:mdpOfLoop}).
    }
    \label{fig:fdrIntro}
\end{figure}

\paragraph{Contributions.}%
In summary, this paper makes the following contributions:
\begin{description}[leftmargin=1em]
    \item[Sampler verification:]
    We formally verify two non-trivial sampling algorithms by means of distributional loop invariants:
    The FDR~\cite{DBLP:journals/corr/abs-1304-1916} (\Cref{fig:fdrIntro}, left) and a more recent, space-efficient variant for sampling from \emph{arbitrary finite-support distributions}, the Fast \emph{Loaded} Dice Roller (FLDR)~\cite{DBLP:conf/aistats/SaadFRM20}.
    \item[Lifting the distributional paradigm from Markov chains to programs:]
    We show that program-level \emph{distributional loop invariants} coincide, in a precise sense, with a notion introduced in~\cite{DBLP:conf/cav/AkshayCMZ23} for finite Markov chains (\Cref{thm:translation}).
    \item[Extending the verification framework:]
    We complement existing invariant-based proof rules for loops~\cite{DBLP:journals/ijfcs/HartogV02,DBLP:conf/esop/BartheEGGHS18} and provide novel verification techniques for reasoning about injective state updates (\Cref{sec:verification}).
\end{description}
\ifappendixversion
\else
    An extended version of this paper containing all omitted proofs is available~\cite{arxiv}.
\fi

\section{Preliminaries}%
\label{sec:prelims}%
$\mathbb{Z}$ is the set of integers, $[0,1] \subseteq \mathbb{R}$ are the \emph{real} probabilities, and $\bools = \{\true,\false\}$.
A \emph{predicate} on a set $S$ is a function $\bool \colon S \to \bools$.
We write $s \models \bool$ iff $\bool(s) = \true$, otherwise $s \not\models \bool$.
We occasionally use $\lambda$-notation to define (unnamed) functions.

\paragraph{Probabilistic Programs.}%
Throughout the paper we fix a finite set $\vars$ of (integer) \emph{program variables} ranged over by $\progvar x$, $\progvar y$, etc.
For $X \subseteq \vars$ we define the (countable) set of \emph{program states} over $X$ as $\Sigma_{X} \coloneq \{\sigma \mid \sigma\colon X \to \mathbb{Z}\}$.
We write $\Sigma$ instead of $\Sigma_\vars$.
For $\sigma \in \Sigma$, $\progvar{x} \in \vars$, and $c \in \mathbb{Z}$, we define the \emph{updated program state} $\sigma[\progvar{x}/c] \coloneqq \sigma'$, where $\sigma'(\progvar{x}) = c$ and $\sigma'(\progvar{y}) = \sigma(\progvar{y})$ for all $\progvar{y} \in \vars \setminus \{\progvar{x}\}$.
$\progvar{x}_1 \mapsto c_1,\ldots,\progvar{x}_n\mapsto c_n$ denotes a program state in which, for all $i = 1,\ldots,n$, variable $\progvar{x}_i$ has value $c_i \in \mathbb{Z}$, and all $\progvar{x} \notin \{\progvar{x}_1,\ldots,\progvar{x}_n\}$ have value $0$.
\begin{definition}[Syntax]%
    \label{probprograms}%
	A \emph{probabilistic program} $\prog$ adheres to the grammar
	\begin{align*}
		\prog \quad&\Coloneqq\quad \ski && ~\mid~ \ass{\progvar{x}}{\arith} \tag*{\textcolor{gray}{(no-op $\mid$ assignment)}} \\
        & ~\mid~ \pchoice{\prog}{\probability}{\prog} && ~\mid~ \comp{\prog}{\prog} \tag*{\textcolor{gray}{(coin flip $\mid$ sequential composition)}} \\
        & ~\mid~ \ifelse{\bool}{\prog}{\prog} && ~\mid~ \while{\bool}{\prog} \tag*{\textcolor{gray}{(conditional choice $\mid$ loop)}}
        ~,
	\end{align*}
	where $\progvar{x} \in \vars$, $\probability \in [0,1]$, $\arith \colon \Sigma \rightarrow \mathbb{Z}$, and $\bool \colon \Sigma \to \bools$ is a predicate.
\end{definition}
We do not fix a concrete syntax for the \emph{expressions} $\arith$ and for the \prfont{if}- and \prfont{while}-\emph{guards} $\bool$ in our programs.
We abbreviate $\ifelse{\bool}{\prog}{\ski}$ to $\ifelseskip{\bool}{\prog}$.
A program $\prog$ is called \emph{loop-free} if it does not contain $\prfont{while}$, and \emph{deterministic} if it does not contain any coin flip constructs.

\paragraph{Distributions.}%
Let $S$ be a countable set (often $\Sigma$ in the remainder).
A function $\mu \colon S \to [0,1]$ with \emph{mass} $\weight\mu \coloneq \sum_{a \in S} \mu(a) \leq 1$ is called a \emph{distribution}.
We write $\Deltaspec{\sim \tau} S \coloneq \{\mu \colon S \to [0,1] \mid \weight{\mu} \leq 1 \land \weight{\mu} \sim \tau\}$ for the set of distributions with mass $\sim \tau$, where $\sim$ is a comparison relation and $\tau\in [0,1]$.
For example, $\Deltaspec{< 1} S$ is the set of \emph{sub-distributions} on $S$.
We define the shortcut $\Delta S \coloneq \Deltaspec{=1} S$ for the set of \emph{proper} distributions on $S$.
For $a \in S$, the \emph{Dirac} (sometimes also called \emph{point}) distribution $\dirac a \in \Delta S$ satisfies $\dirac{a}(a) = 1$.
For a distribution $\mu \in \Deltasub S$ and a predicate $\bool\colon S \to \bools$, we define $\prosim{\mu}{\bool}{}{} \coloneqq \sum_{s\models \bool} \mu(s)$, the probability that $\bool$ is satisfied in $\mu$.
The \emph{support} of $\mu \in \Deltasub S$ is $\supp(\mu) \coloneq \{s \in S \mid \mu(s) > 0\}$.

We endow the set of distributions $\Deltasub S$ with the (pointwise) partial order, i.e., for $\mu_1,\mu_2 \in \Deltasub S$, we write $\mu_1 \leq \mu_2$ iff for all $s \in S$, $\mu_1(s) \leq \mu_2(s)$.
Note that every nondecreasing infinite sequence of distributions ($\omega$-chain) has a supremum (least upper bound) in this order.

Addition $\mu_1 + \mu_2$ and multiplication $\mu_1 \cdot \mu_2$ of distributions $\mu_1,\mu_2 \in \Deltasub S$ are both defined componentwise (the result is not necessarily a distribution for $+$).
For a predicate $\bool \colon S \to \bools$, the \emph{Iverson bracket} $[\bool] \colon S \to [0,1]$ casts truth values to numbers, i.e., for all $s \in S$, $[\bool](s) = 1$ if $s \models \bool$, and $[\bool](s) = 0$ if $s \not\models \bool$.
Note that for $\mu \in \Deltasub S$, the product $[\bool] \cdot \mu \in \Deltasub S$ is a \emph{filtered} distribution which coincides with $\mu$ on all $s \models \bool$, and is zero elsewhere.
We also lift this notation to sets of distributions $M \subseteq \Deltasub S$:
$[\bool] \cdot M \coloneqq \{[\bool] \cdot \mu \mid \mu \in M\}$.
Furthermore, we define $\onlybdists \coloneqq \{\mu \in \Deltasub S \mid \forall s \in \supp(\mu)\colon s \models \bool\}$, the set of distributions that \emph{only} assign probability mass to elements $s$ with $s \models \bool$.

The \emph{marginal distribution} $\marg{X}{\mu} \in \Deltasub \Sigma_X$ of $\mu \in \subdists$ in $X \subseteq \vars$ is
\[
    \marg{X}{\mu} ~\coloneqq~ \lambda \sigma. \sum_{\substack{\tilde\sigma|_X = \sigma|_X}} \mu(\tilde\sigma)
    ~.
    \label{eq:margDef}
\]
We also write $\marg{x}{\mu}$ if $X = \{x\}$.
Marginalization is mass-preserving.
\ifappendixversion
    See \Cref{app:details} for additional preliminary notions.
\fi

\section{Strongest Postconditions and Reachable Distributions}%
\label{sec:post}%
In this section, we first review a standard denotational semantics $\denn$ for our programming language.
We consider $\denn$ as a \emph{ground truth} relative to which we verify our sampling programs.
Afterwards, we define \emph{distributional strongest postconditions} ($\postn$) and the notion of \emph{reachable distributions} during execution of a loop.

\begin{figure}[t]
    \renewcommand{\arraystretch}{1.2} 
    \begin{tabular}{ll}
        \toprule
        Program $\prog$ & Denotational semantics $\den{\prog}(\mu)$ \hfill\gray{($\mu \in \subdists$)} \\
        \midrule
        $\ski$ & $\mu$\\
        $\ass{\progvar{x}}{\arith}$ & $\lambda \sigma.\sum_{\sigma' \colon \sigma'[\progvar{x}/\arith(\sigma')] = \sigma} \mu(\sigma')$\\
        $\pchoice{\progOne}{\probability}{\progTwo}$ & $\den{\progOne} (\probability \cdot \mu) + \den{\progTwo} ((1-\probability) \cdot \mu)$\\
        $\comp{\progOne}{\progTwo}$ & $\den{\progTwo} (\den{\progOne} (\mu))$\\
        $\ifelse{\bool}{\progOne}{\progTwo}$ & $\den{\progOne} ([\bool] \cdot \mu) + \den{\progTwo} ([\neg \bool] \cdot \mu)$\\
        $\while{\bool}{\loopbody}$ & $(\lfp \Phi)(\mu), \textnormal{ where } \Phi(f) = \lambda \eta.[\neg \bool] \cdot \eta + f\big(\den{\loopbody}([\bool] \cdot \eta)\big)$ \\
        \bottomrule
    \end{tabular}
    \caption{%
        Denotational semantics.
        $\den{\prog}(\mu)$ is the output distribution of $\prog$ on input $\mu$.
    }
    \label{fig:den}
\end{figure}

\paragraph{Denotational Distribution Transformer Semantics.}%
The denotational semantics we consider views probabilistic programs as transformers of input distributions to output distributions~\cite{DBLP:conf/focs/Kozen79}.
The definition given here closely follows~\cite{DBLP:journals/ijfcs/HartogV02}.
\begin{definition}[Denotational Semantics]%
    \label{defn:den}%
	Let $\prog$ be a probabilistic program.
    We define the \emph{denotational semantics} ${\den{\prog}\colon \subdists \to \subdists}$ by induction on the structure of $C$ as in \Cref{fig:den}.
\end{definition}
The semantics $\denn$ satisfies the following~(see, e.g.,~\cite{DBLP:conf/lopstr/KlinkenbergBKKM20}):
(1) For every program $\prog$ and input distribution $\mu$ we have $\weight{\mu} \geq \weight{\den{\prog}(\mu)}$, i.e., programs do not generate probability mass.
Further, $\weight{\mu} > \weight{\den{\prog}(\mu)}$ is only possible if $\prog$ contains a loop; the missing probability mass accounts for the probability of non-termination.
(2) $\den\prog$ is \emph{linear}, i.e., $\den{\prog}(\mu_1 + \mu_2) = \den{\prog}(\mu_1) + \den{\prog}(\mu_2)$, and $\den{\prog}(\probability \cdot \mu) = \probability \cdot\den{\prog}( \mu)$ for $\probability \in [0,1]$.
(3) The least fixed point ($\mathsf{lfp}$) occurring the $\prfont{while}$-case can be determined via fixed point iteration which, intuitively, corresponds to unfolding the loop to ever-increasing depth.

\begin{figure}[t]
    \begin{align*}
        & \progAnno{\dirac{x\mapsto0,c\mapsto1}} \tag{initial distribution $\mu_0$} \\
        & \WHILE{c=1} \\
        & \qquad \pchoice{\ass{x}{x+1}}{\nicefrac 1 2}{\ass{c}{0}} \\
        & \prbracecl ~ \progAnno{\textstyle \sum_{j = 0}^\infty \frac{1}{2^{j+1}} \dirac{x\mapsto j, c\mapsto 0} \tag{denotational semantics $\mu_{\text{geo}} = \den{\prog_{\text{geo}}}(\mu_0)$}}
    \end{align*}
    \caption{%
        Program $\prog_{\text{geo}}$ generating a geometric distribution.
    }
    \label{counterex1}
\end{figure}

\paragraph{Distributional Strongest Postconditions.}%
When verifying algorithms, we often have to consider an entire \emph{set} of possible initial distributions and characterize the resulting set of output distributions.
\begin{definition}[Distributional $\postn$]%
    \label{defn:post}%
    The \emph{distributional strongest postcondition} transformer $\post{\prog}$ of probabilistic program $\prog$ is defined as
    \[
        \post{\prog} \colon 2^{\subdists} \to 2^{\subdists}~, \quad
        \post{\prog}(\initSet) ~\coloneq~ \big\{\den{\prog}(\mu) \mid \mu \in \initSet\big\}
        ~.
    \]
\end{definition}

\paragraph{Reachable Distributions.}%
Distributional invariants for probabilistic loops (to be defined in \Cref{sec:distis}) overapproximate the set of distributions \emph{reachable} during loop execution:
\begin{definition}[Reachable Distributions]%
    \label{defn:reach}%
	Let $\prog = \while{\bool}{\loopbody}$ be a loop and $\initSet \subseteq \subdists$.
    The distributions \emph{reachable} in $\prog$ from $\initSet$ are defined as
    \[
        \reach{\prog}(\initSet)
        ~\coloneqq~
        \lfp \LFPset.~\, \initSet ~\cup~ \post{\ifelseskip{\bool}{\loopbody}}(\LFPset)
        ~.
    \]
\end{definition}
Intuitively, $\reachn$ collects all distributions encountered along the way, i.e., those reached from some initial distribution in $\initSet$ after any \emph{finite} number of loop iterations.
Unrolling the fixed point yields an equivalent characterization:
\begin{corollary}%
    \label{cor:reach}%
	For every loop $\prog = \while{\bool}{\loopbody}$ and $\initSet \subseteq \subdists$ it holds that
    \[
    	\reach{\prog}(\initSet)
        ~=~
        \bigcup_{n \in \nats} \post{\ifelseskip{\bool}{\loopbody}}^n(\initSet)
        ~.
    \]
\end{corollary}
Since $\post{\prog}$ is monotonic with respect to $\subseteq$, so is $\reach{\prog}$, as an immediate consequence of \Cref{cor:reach}.
Perhaps counter-intuitively, the output distribution $\den{\prog}(\mu_0)$ of a loop $\prog$ on input $\mu_0$ is \emph{not} necessarily reachable in the sense of \Cref{defn:reach}:
\begin{example}[Reachable Distributions]%
    \label{ex:geodist}%
    Consider the program $\prog_{\text{geo}}$ (\Cref{counterex1}) generating the geometric distribution $\mu_{\text{geo}}$ with parameter $\nicefrac{1}{2}$.
    We have $\mu_{\text{geo}} \notin \reach{\prog_{\text{geo}}}(\{\mu_0\})$ because, intuitively, every reachable distribution assigns a strict positive probability to some program state $\sigma$ with $\sigma(\progvar{c}) = 1$.
    More precisely, $\reach{\prog}(\{\mu_0\}) = \{\mu_i \mid i \in \mathbb{N}\}$, where $\mu_i(\progvar{x} \mapsto j, \progvar{c} \mapsto 0) = \frac{1}{2^{j+1}}$ for all $0 \leq j < i$ and, crucially, $\mu_i(\progvar{x} \mapsto i, \progvar{c} \mapsto 1) = \frac{1}{2^i}$.
\end{example}

Despite \Cref{ex:geodist}, the exact output distribution can be recovered from the reachable distributions by means of a supremum:
\begin{restatable}[From $\reachn$ to $\denn$]{theorem}{postIsSup}%
    \label{thm:postIsSup}%
	Let $\prog = \while{\bool}{\loopbody}$ and $\initDist \in \subdists$.
    The distributions in the set $[\neg \bool] \cdot \reach{\prog}(\{\initDist\})$ form a nondecreasing sequence and
    \[
    	\den{\prog}(\initDist)
        ~=~
        \sup \, [\neg \bool] \cdot \reach{\prog}(\{\initDist\})
        ~.
    \]
\end{restatable}%
\ifappendixversion%
    \begin{proof}
        See \Cref{app:subsec:postIsSupProof}.
    \end{proof}
\fi

\section{Distributional Invariants}%
\label{sec:distis}%
While prior work~\cite{DBLP:journals/ijfcs/HartogV02,DBLP:conf/esop/BartheEGGHS18} studied distributional invariants for probabilistic loops, we were inspired by distributional reasoning in probabilistic finite-state models such as (explicitly represented) Markov chains~\cite{DBLP:conf/cav/AkshayCMZ23}.
In the following, we first review the relevant notions on Markov chains, then define distributional invariants on program level, and then relate both notions.%
\begin{definition}[Markov Chain]%
    \label{def:mc}%
    A \emph{Markov chain} is a tuple $\markc = \markctuple$ with countable \emph{state space} $\states$, and \emph{transition probability\footnote{Notice that $\prmdp(s)$ is allowed to be a \emph{sub}-probability distribution. This is necessary for appropriate handling of loops that terminate with probability $<1$; see \Cref{def:mdpOfLoop}.} function} $\prmdp\colon \states \rightarrow \Deltasub \states$.
\end{definition}
Overloading notation, we view a Markov chain $\markc$ as a distribution transformer similarly to the semantics $\denn$ from \Cref{sec:post}; that is, for every distribution $\mu \in \Deltasub \states$ we define its \emph{successor distribution} as $\markc(\mu) \coloneqq \sum_{s \in \states} \mu(s) \cdot \prmdp(s)$.
\begin{definition}[Distributional Markov Chain Invariant~\cite{DBLP:conf/cav/AkshayCMZ23}\footnote{The definition in \cite{DBLP:conf/cav/AkshayCMZ23} is phrased for the \emph{induced Markov chain} of a Markov decision process and excludes sub-distributions, i.e., applies only to proper distributions.}]%
    \label{defn:distInvMDP}%
    A set $\inv \subseteq \subdistsStates$ is a \emph{distributional invariant} w.r.t.\ $\markc = \markctuple$ and $\mu_0 \in \Deltasub \states$ if $\markc^n(\mu_0) \in \inv$ for all $n \in \nats$.
    We call $\inv$ \emph{inductive} (w.r.t.\ $\markc$) if $\forall \mu \in \inv \colon \markc(\mu) \in \inv$.
\end{definition}
Notice that every inductive set $\inv$ is a distributional invariant w.r.t.\ every initial distribution $\mu_0 \in \inv$.
Property \eqref{eq:vonNeumannInv} in \Cref{sec:intro} is an informal example of a distributional Markov chain invariant. 

\begin{figure}[t]
    \begin{minipage}{0.66\textwidth}
        \centering
        \begin{align*}
            & \progAnno{\tfrac 1 3 \dirac{x\mapsto 0, y \mapsto 0} + \tfrac 2 3\dirac{x\mapsto 1, y \mapsto 0}} \tag{initial distribution $\mu_0$} \\
            & \WHILE{y = 0} \\
            & \quad \IF{x=0} \\
            & \quad\quad \pchoice{\ass{x}{1}}{\nicefrac 1 2}{\,\pchoice{\ass{y}{1}}{\nicefrac 1 4}{\comp{\ass{x}{1}}{\ass{y}{1}}}\,} \\
            & \quad \ELSE \\
            & \quad\quad \pchoice{\ass{x}{0}}{\nicefrac 1 8}{\,\pchoice{\ass{y}{1}}{\nicefrac 3 7}{\comp{\ass{x}{0}}{\ass{y}{1}}}\,} \\
            & \quad \prbracecl \prspace \prbracecl
        \end{align*}
    \end{minipage}
    \hfill
    \begin{minipage}{0.33\textwidth}
        \centering
        \begin{tikzpicture}[scale=0.5, main/.style = {draw, circle, minimum size=0.2cm, inner sep=1pt}, ->, >=stealth]
            \node[main] (st1) at (0,0)     {$0\mid0$};
            \node[main] (st2) at (0,-4)    {$1\mid0$};
            \node[main,accepting] (st3) at (4,0)     {$0\mid1$};
            \node[main,accepting] (st4) at (4,-4)    {$1\mid1$};
            \draw (0,1.5)   -> (st1) node[midway, left]     {$\nicefrac{1}{3}$};
            \draw (0,-5.5)  -> (st2) node[midway, left]     {$\nicefrac{2}{3}$};
            \draw (st1) edge[bend left=0] node[right=-2pt]        {$\nicefrac{1}{2}$} (st2);
            \draw (st1) edge              node[above]        {$\nicefrac{1}{8}$} (st3);
            \draw (st1) edge              node[near end,below]  {$\nicefrac{3}{8}$} (st4);
            \draw (st2) edge[bend left=30] node[left]       {$\nicefrac{1}{8}$} (st1);
            \draw (st2) edge              node[near end,above]  {$\nicefrac{1}{2}$} (st3);
            \draw (st2) edge              node[below]        {$\nicefrac{3}{8}$} (st4);
            \path (st3) edge[loop above] node[right] {$1$} (st3);
            \path (st4) edge[loop below] node[right] {$1$} (st4);
        \end{tikzpicture}
    \end{minipage}
    \caption{%
        \textit{Left:} Probabilistic loop with initial distribution.
        \textit{Right:} Reachable fragment of the loop's big-step operational Markov chain. Terminal states have a double boundary.
    }
    \label{fig:operationalMC}
\end{figure}

\paragraph{Distributional Invariants for Probabilistic Loops.}%
Based on the $\reachn$-transformer (\Cref{defn:reach}), we now define distributional invariants for probabilistic loops.
Our definition parallels---and is inspired by---the one for Markov chains (\Cref{defn:distInvMDP}).
\begin{definition}[Distributional Loop Invariant]%
    \label{defn:distInv}%
    Let $\prog = \while{\bool}{\loopbody}$ and $\initSet \subseteq \subdists$. A set $\inv \subseteq \subdists$ such that $\reach{\prog}(\initSet) \subseteq \inv$ is called \emph{distributional loop invariant} w.r.t.\ $\prog$ and $\initSet$.
    Further, $\inv$ is \emph{inductive} if
    \[
        \forall \mu \in \inv \colon \den{\ifelseskip{\bool}{\loopbody}}(\mu) \,\in\, \inv
        \tag*{\textcolor{gray}{(i.e., $\inv$ is preserved by one guarded iteration)}.}
    \]
\end{definition}

\paragraph{From Probabilistic Loops to Markov Chains.}
We now establish a relationship between distributional invariants in Markov chains (\Cref{defn:distInvMDP}) and distributional loop invariants (\Cref{defn:distInv}).
In a nutshell, loop invariants are invariants in a certain Markov chain defined in the spirit of a big-step operational semantics:
\begin{definition}[Operational Markov Chain]%
    \label{def:mdpOfLoop}%
    The \emph{operational Markov chain} of a loop $\prog = \while{\bool}{\loopbody}$ is $\markcTr = (\Sigma, \prmdp)$, where $\prmdp(\sigma) \coloneqq \den{\ifelseskip{\bool}{\loopbody}}(\dirac \sigma)$.
\end{definition}
Intuitively, the operational Markov chain has a transition from every program state $\sigma$ satisfying the loop guard to every state reachable from $\sigma$ in \emph{one guarded loop iteration}. Let us consider a more formal example next.
\begin{example}[Invariants in Markov Chains]%
    \label{ex:tobiInvariant}%
    \Cref{fig:operationalMC} (right) depicts the fragment of the operational Markov chain of the program in \Cref{fig:operationalMC} (left) that is reachable from the initial distribution $\mu_0$.
    It can be verified that the Markov chain admits the following inductive invariant w.r.t.\ $\mu_0$:
    \begin{align*}
        \inv ~\coloneq~ \big\{ \mu \in \dists ~{}\mid{}~ & \pr{\mu}{\progvar{x} = 1 \land \progvar y = 0} ~=~ 2 \cdot \pr{\mu}{\progvar{x} = 0 \land \progvar y = 0} \\ 
        ~{}\land{}~ & \pr{\mu}{\progvar{x} = 1 \land \progvar y = 1} ~=~ \phantom{2{}\cdot{}}\pr{\mu}{\progvar{x} = 0 \land \progvar y = 1} ~\big\}
    \end{align*}
    By \Cref{thm:translation} (stated below), $\inv$ is thus also an inductive invariant for the program $\prog$ in \Cref{fig:operationalMC} (left).
\end{example}

Formally, the correspondence between distributional invariants in loops and Markov chains is as follows:
\begin{restatable}{theorem}{translation}%
    \label{thm:translation}%
	Let $\prog = \while{\bool}{\loopbody}$ with operational Markov chain $\mdpTr$.
    Further, let $\initDist \in \subdists$ and $\inv \subseteq \subdists$.
    Then:
    \begin{align*}
        & \inv \text{ is a distributional invariant for } \prog \text{ and } \{\initDist\} \text{ (Def.~\ref{defn:distInv})} \\
        \text{if and only if} \qquad & \inv \text{ is a distributional invariant for } \mdpTr \text{ and } \initDist \text{ (Def.~\ref{defn:distInvMDP})}
        ~.
    \end{align*}
    Furthermore, $\inv$ is inductive w.r.t.\ $\prog$ iff $\inv$ is inductive w.r.t.\ $\mdpTr$.
\end{restatable}%
\begin{proof}[sketch]%
    The key observation is that the stream of distributions $\markcTr^n(\mu_0)$, $n \in \nats$, resulting from iterating $\markcTr$ on $\mu_0$ coincides with $\reach{\prog}(\{\mu_0\})$.
    The proof also exploits linearity of $\denn$\ifappendixversion~(see \Cref{subsec:prooftranslation} for details)\fi.
\end{proof}

\section{Invariant Verification Framework}%
\label{sec:verification}%
\paragraph{Proof Rules for Loops.}%
Our framework uses proof rules for reasoning about loops. It is convenient to formulate these rules in Hoare triple notation.
\begin{definition}[Hoare Triples~\cite{DBLP:journals/ijfcs/HartogV02}]%
    \label{def:hoareTriples}%
    Let $\hSetOne, \hSetTwo \subseteq \subdists$ and $\prog$ be a program.
    A \emph{Hoare triple} $\tripleGen$ is \emph{valid} (denoted: $\models\tripleGen$) if $\post{\prog}(\hSetOne) \subseteq \hSetTwo$.
\end{definition}

We provide one proof rule each for \emph{partial} and \emph{total} correctness in \Cref{fig:whileRulesMainSec}.
Note that a program $\prog$ is \emph{almost surely terminating (AST)} w.r.t.\  $\initSet \subseteq \subdists$ if $\forall \mu \in \initSet\colon \weight{\mu} = \weight{\den{\prog}(\mu)}$. Furthermore, the \emph{$\omega$-closure} of $\initSet$ is defined as
\[
	\closOf\initSet
    ~\coloneqq~
    \big\{ \text{sup}_{n \in \nats}\, \mu_n \mid (\mu_n)_{n \in \nats} \text{ is an $\omega$-chain in }\initSet\big\}
    ~.
\]
\ifappendixversion%
    The rules in \Cref{fig:whileRulesMainSec} are proved sound in \Cref{app:hoareRules}.
\fi

Our proof rules are inspired by and complement existing Hoare logics~\cite{DBLP:journals/ijfcs/HartogV02,DBLP:conf/esop/BartheEGGHS18} (see \Cref{sec:relwork} for further discussion).

\begin{figure}[t]
    \begin{gather*}
        \frac{\triple{\inv}{\ifelseskip{\bool}{\loopbody}}{\inv}}
        {\triple{\inv}{\while{\bool}{\loopbody}}{\closOf{[\neg \bool] \inv}}}
        \tag{\textsc{While-P}}\label{WhileP} \\[1em]
        \frac{%
            \triple{\inv}{\ifelseskip{\bool}{\loopbody}}{\inv}
            \qquad
            \prog \text{ AST w.r.t.\ } \inv
            \qquad
            \tau = \inf_{\mu \in \inv} \weight{\mu}
        }%
        {\triple{\inv}{\while{\bool}{\loopbody}}{\closOf{[\neg \bool] \inv} \cap \Deltaspec{\geq \tau}\Sigma}}
        \tag{\textsc{While-T}}\label{WhileT}
    \end{gather*}
    \caption{Proof rules for reasoning about loops.}
    \label{fig:whileRulesMainSec}
\end{figure}

\paragraph{Injective Assignments.}%
We now present a simplification technique that can be applied to assignments \emph{injective in \progvar{$x$}}, i.e., assignments $\ass{\progvar{x}}{\arith}$ such that for any two program states $\sigma \neq \sigma'$, it holds that $\sigma[\progvar{x}/\arith(\sigma)] \ne \sigma'[\progvar{x}/\arith(\sigma')]$.
Such assignments induce a \emph{partial inverse} $\arith^{-1}_\progvar{x}$%
\ifappendixversion%
    ~(see  \Cref{def:partialInverse} in \Cref{app:hoareRules} for details)%
\else%
    (see \cite{arxiv} for details)%
\fi.
\begin{definition}[Substitution Operator for Distributions]%
    \label{defn:substForDists}%
    Let $\arith\colon \Sigma \dashrightarrow \mathbb{Z}$ be a partial function, $\progvar{x} \in \vars$, and $\mu \in \subdists$.
    We define the \emph{substitution}
    \[
    	\mu[\progvar{x}/\arith]
        ~\coloneqq~
        \lambda \sigma.
    	\begin{cases}
    		\mu(\sigma[\progvar{x}/\arith(\sigma)]) & \textnormal{if } \arith(\sigma) \textnormal{ is defined} \\
    		0 & \textnormal{otherwise.}
    	\end{cases}
    \]
\end{definition}
For example, we have $\dirac{\progvar x \mapsto 1}[x/x+1] = \dirac{\progvar x \mapsto 0}$.
Substitution is not mass preserving; e.g., $\dirac{\progvar x \mapsto 1}[x/2 x] = 0$ because there does not exist $x \in \mathbb Z$ such that $2x = 1$.
\begin{restatable}[Injective Assignments]{lemma}{injassignments}%
    \label{lem:injassignments}%
    Let $\ass{\progvar x}{\arith}$ be an assignment where $\arith$ is injective in $\progvar x$ with partial inverse $\arith^{-1}_\progvar{x}$. 
    Then it holds for all $\mu \in \subdists$ that
    \[
        \den{\ass{\progvar{x}}{\arith}}(\mu)
        ~=~
        \mu[\progvar x / \arith^{-1}_\progvar{x}]
        ~.
    \]
\end{restatable}

\paragraph{Unmodified Variables.}%
The following is a useful observation:
\begin{restatable}[Unmodified Variables]{lemma}{unmodvars}%
    \label{thm:unmodvars}%
    Let $\prog$ be a program with initial distribution $\mu_0 \in \subdists$ and let $\mu = \den{\prog}(\mu_0)$.
    For all $\progvar{x} \in \unmod{\prog}$ and $z \in \mathbb{Z}$,
    \[
        \prosim{\mu}{\progvar{x}}{=}{z} ~\leq~ \prosim{\mu_0}{\progvar{x}}{=}{z} ~,
    \] 
    where $\unmod{\prog}$ is the set of \emph{unmodified variables} in $\prog$.
\end{restatable}
In \Cref{thm:unmodvars}, we obtain an inequality due to potential non-termination of $\prog$.
In the following two sections, we apply the techniques from this section to two non-trivial sampling algorithms.

\section{The Fast Dice Roller}%
\label{sec:fdr}%

\begin{figure}[t]
    \begin{minipage}{0.49\textwidth}
        \setlength{\jot}{2pt} 
        \begin{align*}
            &\progAnno{\Mfdr = \{\dirac{v \mapsto 1, c\mapsto 0, n \mapsto \numb{n}} \mid  \numb{n} > 0 \}} \\
            &\WHILE{v < n} \\
            &\qquad \ass{v}{2v} \SEMICLN \\
            &\qquad \pchoice{\ass{c}{2c}}{\nicefrac 1 2}{\ass{c}{2c+1}} \SEMICLN \\
            &\qquad \IF{c \geq n} \\
            &\qquad	\qquad \ass{v}{v-n} \SEMICLN \ass{c}{c-n} \\[-3pt]
            &\qquad \prbracecl \\[-3pt]
            &\prbracecl
        \end{align*}
    \end{minipage}
    \hfill
    \begin{minipage}{0.49\textwidth}\centering
        \begin{tikzpicture}[->,scale=0.9,every state/.style={inner sep=0pt}]
            \node[state,initial,initial where=above,initial text=](1) at (0,0)  {$1,0$};
            \node[state](2) at (-2,-1) {$2,0$};
            \node[state,accepting](3) at (-3,-2) {$4,0$};
            \node[state](5) at (2,-1) {$2,1$};
            \node[state,accepting](6) at(-1,-2) {$4,1$};
            \node[state,accepting](8) at(1,-2) {$4,2$};
            
            \draw(1)-- node[above] {$\nicefrac 1 2$} (2);
            \draw(1)edge[bend right=10] node[below] {$\nicefrac 1 2$}(5);
            \draw(2)--node[auto] {$\nicefrac 1 2$}(3);
            \draw(2)--node[auto] {$\nicefrac 1 2$}(6);
            \draw(5)--node[auto] {$\nicefrac 1 2$}(8);
            
            \draw(5) to [bend right=10] node[above] {$\nicefrac 1 2$} (1);
            
            \path  (3)   edge[loop below] node[auto]  {\scriptsize$1$} (3);
            \path  (6)   edge[loop below] node[auto]  {\scriptsize$1$} (6);
            \path  (8)   edge[loop below] node[auto]  {\scriptsize$1$} (8);
        \end{tikzpicture}
    \end{minipage}
    \caption{%
        \textit{Left:} Program $\progfdr$.
        \textit{Right:} Markov chain $\mathcal{M}_{\progfdr}$ (fragment; $\numb{n} = 3$). The state $(4,3)$ is not reached, as it satisfies $c < n$.
    }
    \label{fig:MarkovChainFDR3}
\end{figure}

The \emph{Fast Dice Roller} (FDR)~\cite{DBLP:journals/corr/abs-1304-1916} takes an integer $\textvar{n} >0$ as input and outputs any integer $\textvar c \in \{0,\dots,\textvar{n}-1\}$ with uniform probability $\nicefrac{1}{\textvar{n}}$.
We verify the FDR as stated in \Cref{fig:MarkovChainFDR3} (left), a slight adaptation of the original program~\cite{DBLP:journals/corr/abs-1304-1916} to suit our syntax.
The FDR solves the uniform sampling problem using an optimal (expected) number of \emph{fair} coin flips~\cite{DBLP:journals/corr/abs-1304-1916} (we do not verify optimality here).

The FDR operates on $\varsfdr = \{\progvar{v}, \progvar{c}, \progvar{n}\}$, where \emph{values} are denoted $\numb v, \numb c, \numb n \in \mathbb{Z}$.
The variables $\textvar{v}$ and $\textvar{c}$ are initialized with $1$ and $0$, respectively.
The read-only variable $\textvar{n} \in \unmod{\progfdr}$ is initialized with an arbitrary $\numb{n}> 0$.
We thus consider the \emph{infinite} set of initial (Dirac) distributions $\Mfdr = \{\dirac{v \mapsto 1, c\mapsto 0, n \mapsto \numb{n}} \mid  \numb{n} > 0 \}$.
As $n$ is not modified, the operational Markov chain $\markc_\progfdr$ consists of infinitely many disconnected fragments: one fragment for each value $\numb{n}$.
\ifappendixversion%
    The fragments reachable from state $n \mapsto \numb n, v\mapsto 1,c \mapsto 0$ for $\numb n = 3$, $\numb n = 5$, and $\numb n = 9$ are depicted in \Cref{fig:MarkovChainFDR3}, \Cref{fig:fdrIntro}, and \Cref{fig:MarkovChainFDR9} in \Cref{app:practical}, respectively.
\else%
    The fragments reachable from state $n \mapsto \numb n, v\mapsto 1,c \mapsto 0$ for $\numb n = 3$, $\numb n = 5$, and $\numb n = 9$ are depicted in \Cref{fig:MarkovChainFDR3}, \Cref{fig:fdrIntro}, and \cite[Appendix]{arxiv}, respectively.
\fi

For $x \in \vars$ and integer $m \geq 0$, let the (discrete) \emph{uniform distribution} on $\{0,\ldots,m\}$ be defined as $\unif{x}{0}{m} \coloneq \frac{1}{m+1} \cdot [0 \leq x \leq m]$.
Formally, our goal is to prove validity of the Hoare triple $\triple{\Mfdr}{\progfdr}{\hSetTwo_\textsc{fdr}}$, where
\[
    \hSetTwo_\textsc{FDR}
    ~\coloneq~
    \{\mu \in \dists ~\mid~ \exists \numb n > 0 \colon \pr{\mu}{n=\numb n} = 1 \land \marg{c}{\mu} = \unif{c}{0}{\numb n-1}\}
    ~.
\]
Notice that the marginalization operator $\mathbb{M}_{c}$ operator (defined on page \pageref{eq:margDef}) discards the variables which are irrelevant for stating the correctness of the FDR.

\begin{figure}[t]
    \setlength{\jot}{1pt} 
    \begin{align*}
        & \progAnno{[v<n] \cdot \mu} \\
        & \ass{v}{2v} \SEMICLN \\
        & \progAnno{[\tfrac v 2 < n] \cdot \mu[v / \tfrac v 2]} \\ 
        & \pchoice{\ass{c}{2c}}{\nicefrac 1 2}{\ass{c}{2c+1}} \SEMICLN \\
        & \progAnno{\tfrac 1 2 \cdot [\tfrac v 2 < n] \cdot \left(\mu[v/\tfrac v 2, c / \tfrac c 2] + \mu[v/\tfrac v 2, c / \tfrac{c-1}{2}]\right) } \\ 
        & \IF{c \geq n} \\
        & \qquad \progAnno{\tfrac 1 2 \cdot [\tfrac v 2 < n \land c \geq n] \cdot \left(\mu[v/\tfrac v 2, c / \tfrac c 2] + \mu[v/\tfrac v 2, c / \tfrac{c-1}{2}]\right) } \\ 
        & \qquad \ass{v}{v-n} \SEMICLN \\
        & \qquad \progAnno{\tfrac 1 2 \cdot [\tfrac{v+n}{2} < n \land c \geq n] \cdot \left(\mu[v/\tfrac{v+n}{2}, c / \tfrac c 2] + \mu[v/\tfrac{v+n}{2}, c / \tfrac{c-1}{2}]\right) } \\ 
        & \qquad \ass{c}{c-n} \\
        & \qquad \progAnno{\tfrac 1 2 \cdot [\tfrac{v+n}{2} < n \land c \geq 0] \cdot \left(\mu[v/\tfrac{v+n}{2}, c / \tfrac{c+n}{2}] + \mu[v/\tfrac{v+n}{2}, c / \tfrac{c+n-1}{2}]\right) } \\ 
        & \prbracecl \\
        & \progAnno{\tfrac 1 2 \cdot [\tfrac{v+n}{2} < n \land c \geq 0] \cdot \left(\mu[v/\tfrac{v+n}{2}, c / \tfrac{c+n}{2}] + \mu[v/\tfrac{v+n}{2}, c / \tfrac{c+n-1}{2}]\right)} \\[-2pt]
        & \qquad \progAnnoNoSlashes{ + \tfrac 1 2 \cdot [\tfrac v 2 < n \land c < n] \cdot \left(\mu[v/\tfrac v 2, c / \tfrac c 2] + \mu[v/\tfrac v 2, c / \tfrac{c-1}{2}]\right)}
    \end{align*}
    \caption{Derivation of {$\den{\loopbody}([\textvar{v} < \textvar{n}] \mu)$}, where $\loopbody$ is the loop body of the FDR.}
    \label{nextfdr}
\end{figure}

We propose the following distributional loop invariant for $\progfdr$ and $\Mfdr$:
\begin{align*}
    \Ifdr & \eeq \Big\{  \mu \in \dists  ~\mid~  \exists \numb{n} > 0 \colon \prosim{\mu}{\progvar{n}}{=}{\numb{n}} ~=~ 1 \tag{a}\label{fdrA} \\
    &\lland \pr{\mu}{0 < \progvar{v} < 2n-1} ~=~ 1 \tag{b}\label{fdrB} \\
    &\lland \pr{\mu}{0 \leq \progvar{c} < \minim{\progvar{v}}{n}} ~=~ 1 \tag{c}\label{fdrC} \\
    &\lland \forall \sigma, \sigma' \models 0 \,{\leq}\, c \,{<}\, \minim{\progvar v}{n} \land n \,{=}\, \numb n \colon \sigma(\progvar{v}) {=} \sigma'(\progvar{v}) {\implies} \mu(\sigma) {=} \mu(\sigma')
    \Big\}
    \tag{d}\label{fdrD}
\end{align*}
The set $\Ifdr$ thus contains all \emph{proper} distributions satisfying conditions \eqref{fdrA} through \eqref{fdrD}; see \Cref{fig:fdrreach}.
Intuitively, \eqref{fdrA} ensures that the value of $\progvar n$ is fixed in every $\mu \in \Ifdr$.
Conditions \eqref{fdrB} and \eqref{fdrC} ensure bounds on $\textvar{v}$ and $\textvar{c}$, respectively.
In fact, \eqref{fdrA} through \eqref{fdrC} altogether ensure that each $\mu \in \Ifdr$ has \emph{finite support} (\Cref{fig:fdrreach}).
Finally, condition \eqref{fdrD} formalizes Lumbroso's original correctness argument:
\enquote{$\textvar{c}$ is uniformly distributed over $\{0,\dots,\textvar{v}-1\}$}~\cite{DBLP:journals/corr/abs-1304-1916}.
In our formulation of the FDR, variable $\textvar{c}$ is at all times uniformly distributed between $0$ and the \emph{minimum} of $\textvar{v} - 1$ and $\textvar{n} - 1$. The uniform distribution is established by forcing the probability masses assigned to all the states $\sigma$ with $\sigma \models 0 \leq c < \minim{\progvar v}{n}$ to be equal.
The four conditions are illustrated in \Cref{fig:fdrreach}.
\begin{restatable}[FDR Invariant]{lemma}{fdrinv}%
    \label{lem:fdrinv}%
    $\Ifdr$ is an inductive distributional invariant for $\progfdr$ and the set of initial distributions $\Mfdr = \{\dirac{v \mapsto 1, c\mapsto 0, n \mapsto \numb{n}} \mid  \numb{n} > 0 \}$.
\end{restatable}
\begin{proof}[sketch]
    The proof is split into two parts:
    \begin{enumerate}
        \item \textit{Containment of the initial distributions:}
        We argue that $\Mfdr \subseteq \Ifdr$. Let $\mu \in \Mfdr$. Since $\mu$ is a Dirac distribution, condition \eqref{fdrA} is satisfied.
        Condition \eqref{fdrD} is satisfied because, since $\mu$ is Dirac, there do not exist $\sigma \neq \sigma' \in \Sigma$ with $\mu(\sigma) > 0$ \emph{and} $\mu(\sigma') > 0$.
        Conditions \eqref{fdrB} and \eqref{fdrC} are trivially satisfied.
        \item \textit{Inductivity:}
        We show that for all $\mu \in \Ifdr$ it holds that $\den{\ifelseskip{\bool}{\loopbody}}(\mu) \in \Ifdr$ where $\bool = [\textvar{v} < \textvar{n}]$ and $\loopbody$ are the guard and the loop body of the FDR, respectively.
        We derive the sub-expression $\den{\loopbody}([\bool] \mu)$ in \Cref{nextfdr} by applying \Cref{lem:injassignments} repeatedly.
        \ifappendixversion%
            See \Cref{sec:fdrinv} for details.
        \fi%
    \end{enumerate}
\end{proof}

\begin{figure}[t]
    \begin{minipage}{0.55\textwidth}
        \begin{tikzpicture}[
            scale=0.85,
            every node/.style={font=\footnotesize},
            main/.style = {draw, circle split, minimum size=0.7cm, inner sep=1pt, align=center},
            axis/.style={thick,->},
            grid/.style={gray!30, thin},
            ellip/.style={red, thick}
        ]
            \fill[black!7] (0.5,-0.5) -- (6.7,-0.5) -- (6.7,3.5) -- (3.5,3.5) -- (3.5,2.5) -- (2.5,2.5) -- (2.5,1.5) -- (1.5,1.5) -- (1.5,0.5) -- (0.5,0.5) -- cycle;

            \foreach \x in {0,...,6} {
                \draw[grid] (\x,-0.5) -- (\x,3.5);
            }
            \foreach \y in {0,...,3} {
                \draw[grid] (-0.5,\y) -- (6.7,\y);
            }

            \draw[axis] (-0.3,-0.7) -- (6.9,-0.7);
            \draw[axis] (-0.5,-0.3) -- (-0.5,3.7) node[above=2pt] {$c$};

            \node[main] (00) at (0,0) {$0$\nodepart{lower}$0$};
            \node[main] (10) at (1,0) {$1$\nodepart{lower}$0$};
            \node[main] (01) at (0,1) {$0$\nodepart{lower}$1$};
            \node[main] (11) at (1,1) {$1$\nodepart{lower}$1$};
            \node[main] (20) at (2,0) {$2$\nodepart{lower}$0$};
            \node[main] (21) at (2,1) {$2$\nodepart{lower}$1$};

            \node[main] (n0) at (4,0) {$\numb{n}$\nodepart{lower}$0$};
            \node[main] (nn1) at (4,2.8) {$\numb{n}$\nodepart{lower}$\numb{n}{-}1$};
            \node[main] (2n0) at (6,0) {$2\numb{n}{-}2$\nodepart{lower}$0$};
            \node[main] (2nn1) at (6,2.8) {$2\numb{n}{-}2$\nodepart{lower}$\numb{n}{-}1$};

            \node at (3,0) {$\cdots$};
            \node at (5,0) {$\cdots$};
            \node at (4,1.5) {$\vdots$};

            \draw[ellip] (2,0.5) ellipse (0.5cm and 1.4cm);
            \draw[ellip] (4,1.5) ellipse (0.75cm and 2.4cm);

            \node[red] at (6.0, 1.1) {$\mu(\sigma) = \mu(\sigma')$};
            \node[red] at (3,0.8) {$\cdots$};
            \node[red] at (5.2,1.5) {$\cdots$};

            \node at (3,-1.05) {$v$};
        \end{tikzpicture}
    \end{minipage}
    \begin{minipage}{0.44\textwidth}
        In $\Ifdr$, condition \eqref{fdrA} fixes $\numb{n} \in \nats$, allowing us to focus on $(\textvar{v}, \textvar{c}) \in \mathbb{Z}^2$.
        Conditions \eqref{fdrB}--\eqref{fdrC} ensure that only finitely many states are reachable.
        By condition \eqref{fdrD}, the invariant assigns uniform probabilities per column for $\textvar{c} < \minim{\progvar{v}}{n}$ (red ellipses).
    \end{minipage}
    \caption{%
        Illustration of the invariant $\Ifdr$.
    }
    \label{fig:fdrreach}
\end{figure}

Using \Cref{lem:fdrinv} and the proof rules \ref{WhileP} and \ref{WhileT} from \Cref{sec:verification}, we can now derive the following partial and total correctness properties:

\begin{restatable}[FDR Correctness]{theorem}{fdrtotal}%
	\label{thm:fdrpartial}%
    \label{thm:fdrtotal}%
    Program $\progfdr$ is partially correct, i.e., for all $\mu_0 \in \Mfdr$ and $\numb{n} > 0$ such that $\prosim{\mu_0}{\progvar{n}}{=}{\numb{n}} = 1$, the final distribution $\mu = \den{\progfdr}(\mu_0)$ satisfies:
    \[
        \exists r \in [0,1] \colon \quad \marg{\progvar{c}}{\mu} \eeq r \cdot \unif{\progvar{c}}{0}{\numb{n}-1}
    ~.
    \]
    Moreover, if $\progfdr$ is AST w.r.t.\ $\Ifdr$, then program $\progfdr$ is totally correct, i.e.,
    \[
        \marg{\progvar{c}}{\mu} \eeq \unif{\progvar{c}}{0}{\numb{n}-1} ~.
    \]
\end{restatable}
\noindent%
Notice that the partial correctness statement in \Cref{thm:fdrpartial} ensures equal probabilities for each final value of $0 \leq \progvar c < \textvar{n}$, but does not imply anything about the termination probability.
We can prove that $\progfdr$ is AST using existing techniques~\cite{DBLP:journals/pacmpl/MajumdarS25,10.1145/3158121}%
\ifappendixversion%
    \ (see \Cref{app:fdrast}).
\else
    \ (see \cite[Appendix]{arxiv}).
\fi
This establishes total correctness of the FDR.

\section{The Fast Loaded Dice Roller}
\label{sec:fldr}

\begin{figure}[t]
    \begin{minipage}[t]{0.5\textwidth}
        \begin{align*}
            &\progAnno{\textnormal{Inputs}: a_1,\dots,a_n}\\
            &\ass{m}{\textstyle\sum_{i = 1}^n a_i} \SEMICLN \ass{k}{\lceil \log(m) \rceil} \SEMICLN \\
            & \ass{a_{n + 1}}{2^k - m} \SEMICLN\\
            & \texttt{for } j = 0,\dots,k ~\texttt{do} ~\prbraceop\\
            & \quad d := 2 \cdot \bound{h}{j-1} - 3\\
            & \quad \texttt{for } i = 1,\dots,n+1 ~\texttt{do} ~\prbraceop\\
            & \qquad  \texttt{bool} ~\ass{w}{(a_i >> k-j)} \SEMICLN \\
            & \qquad  \ass{h[j]}{h[j] + w} \SEMICLN\\
            & \qquad  \IF{w} \ass{H[d,j]}{i} \SEMICLN \ass{d}{d-1} \prbracecl\\
            & \quad \prbracecl~  \prbracecl
        \end{align*}
    \end{minipage}%
    \hfill
    \begin{minipage}[t]{0.45\textwidth}
        \begin{align*}
                &\progAnno{\Mfldr = \{\dirac{d \mapsto 0, c \mapsto 0, n \mapsto \numb{n}} \mid \numb{n} > 0\}}\\ 
                &\WHILE{\,d < \bound{h}{c}\,}\\
                &\quad \ass{c}{c+1} \SEMICLN\\
                &\quad \pchoice{\,\ass{d}{2d}\,}{\nicefrac 1 2}{\,\ass{d}{2d+1}\,} \SEMICLN \\
                &\quad \IF{\,H[d,c] = n + 1\,}\\
                &\qquad \progAnno{\text{reject and start over}} \\
                &\qquad	 \ass{c}{0} \SEMICLN \\
                &\qquad \ass{d}{0} \\
                &\quad \prbracecl\\
                &\prbracecl
        \end{align*}
    \end{minipage}
    \caption{\textit{Left:} Preprocessing for the FLDR. \textit{Right:} The FLDR program $\progfldr$. The output is given by the value of $H[d,c]$ upon program termination.}
    \label{translatedfLOADEDdrPre}
\end{figure}

The \emph{Fast Loaded Dice Roller} (FLDR)~\cite{DBLP:conf/aistats/SaadFRM20} samples a random integer from a given discrete probability distribution with finite support.
Specifically, given positive integers $a_1,\dots,a_n$ as input, the algorithm outputs $i \in \{1,\dots,n\}$ with probability $\frac{a_i}{m}$, where $m \coloneqq \sum_{i = 1}^n a_i$.
We use the tuple notation $(\frac{a_1}{m},\dots,\frac{a_n}{m})$ for the resulting output distribution.

Like the FDR from \Cref{sec:fdr}, the FLDR relies on fair coin flips only.
The high-level idea is to perform rejection sampling from the proposal distribution $(\frac{a_1}{2^k}, \dots, \frac{a_n}{2^k}, 1-\frac{m}{2^k})$, where $k = \lceil \log_2 m \rceil$.
Because the denominators are powers of $2$, this distribution can be sampled efficiently using only unbiased random bits.
While the FLDR does not achieve optimal runtime in terms of the number of coin flips, it can be implemented in a memory-efficient manner~\cite{DBLP:conf/aistats/SaadFRM20}.

First, a preprocessing (\Cref{translatedfLOADEDdrPre}, left) prepares suitable data structures: an array $h$ and a matrix $H$ (illustrated in \Cref{fig:bintree}). Then, the main program $\progfldr$ (\Cref{translatedfLOADEDdrPre}, right) is executed.
This program is a slight variant of the original FLDR from \cite{DBLP:conf/aistats/SaadFRM20}:
It is written in our syntax and contains less branches in the loop body.
The short-hand notation $\bound{h}{c} \coloneqq 2^c - \sum_{j=0}^{c} h[j] \cdot 2^{c-j}$ denotes the first position in row $c$ of matrix $H$ that contains a number other than $0$.
Next, we provide an intuition for the data structures $h$ and $H$, which encode a certain \emph{binary tree}.

\paragraph{The Binary Tree Underlying a Distribution.}%
Given a distribution $(p_1, \dots, p_n)$ with $n \geq 1$, its \emph{underlying binary tree} is constructed as follows:
For each $i$, a leaf labelled $i$ appears in row $c \geq 0$ if and only if the binary expansion of $p_i$ has a $1$ at position $c$ (e.g., the binary representation of $\frac{3}{8}$ is $0.011$, which has $1$s at positions $2$ and $3$).
Within each row, the leaves are ordered by their labels in descending order, which uniquely determines the structure of the binary tree.

The FLDR preprocessing effectively determines the binary tree underlying its proposal distribution $(\frac{a_1}{2^k}, \dots, \frac{a_n}{2^k}, 1-\frac{m}{2^k})$.
The tree is stored in the data structures $h$ and $H$ in the following manner:
$h[c]$ is the number of leaves in row $c$ of the tree, and we have $\H{d}{c} = j$ if in the $c$-th row, the $d$-th node is labelled by $j$; otherwise, $\H{d}{c} = 0$.
\begin{example}
    \label{ex:bintree}
    Consider the input $(a_1, a_2, a_3) = (3, 2, 1)$.
    The target output distribution is thus $(\frac{3}{6}, \frac{2}{6}, \frac{1}{6})$, and we have $m = 3 + 2 + 1 = 6$ and $k = \lceil \log_2 m \rceil = 3$.
    Since $2^k = 8$, the proposal distribution is $(\frac{3}{8}, \frac{2}{8}, \frac{1}{8}, \frac{2}{8})$ with binary encoding $(0.011, 0.01, 0.001, 0.01)$.
    \Cref{fig:bintree} shows the resulting binary tree for the proposal distribution as well as the data structures $h$ and $H$.
    
    \begin{figure}[t]
        \begin{center}
            \begin{tikzpicture}[node distance={20mm}, main/.style = {draw, shape = circle, minimum size = 0.1cm, inner sep=3pt},->,scale=0.7]
                \node[main, fill = black](1) at (0,0) {};
                
                \node[main, fill= black](2) at (-2,-0.7) {};
                \node[main, fill= black](3) at (2,-0.7) {};
                
                \node[main, fill= black](4) at (-3,-1.4) {};
                \node[main](5) at (-1,-1.4) {$4$};
                \node[main](6) at (1.3,-1.4) {$2$};
                \node[main](7) at (2.8,-1.4) {$1$};
                
                \node[main](8) at (-3.8,-2.1) {$3$};
                \node[main](9) at (-2.2,-2.1) {$1$};

                \draw(1) to (2);
                \draw(1) to (3);
                \draw(2) to (4);
                \draw(2) to (5);
                \draw(3) to (6);
                \draw(3) to (7);
                \draw(4) to (8);
                \draw(4) to (9);
                
                \node(t1) at (5,-1.2) {
                    \begin{tabular}{@{\hspace{6pt}} r @{\hspace{6pt}} r @{\hspace{6pt}}}
                    \toprule
                    $c$ & $h[c]$ \\
                    \midrule
                    0 & 0 \\
                    1 & 0 \\
                    2 & 3 \\
                    3 & 2 \\
                    \bottomrule
                \end{tabular}
                };
                \node(t2) at (10,-1.2) {
                    \begin{tabular}{@{\hspace{6pt}} r @{\hspace{8pt}} r @{\hspace{6pt}} r @{\hspace{6pt}} r @{\hspace{6pt}} r @{\hspace{6pt}}}
                    \toprule
                    $\H{d}{c}$ & $d = 0$ & 1 & 2 & 3 \\
                    \midrule
                    $c = 0$ & 0 &   &   &   \\
                    1       & 0 & 0 &   &   \\
                    2       & 0 & 4 & 2 & 1 \\
                    3       & 3 & 1 &   &   \\
                    \bottomrule
                \end{tabular}
                };
            \end{tikzpicture}
            \caption{
                \textit{Left:} Binary tree underlying the proposal distribution $(\frac{3}{8}, \frac{2}{8}, \frac{1}{8}, \frac{2}{8})$, which is $(0.011, 0.01, 0.001, 0.01)$ in binary.
                \textit{Right:} Data structures $h$ and $H$ storing the binary tree.
                $h$ determines the number of leaves in each row and $H$ stores the leaf labels.
            }
            \label{fig:bintree}
        \end{center}
    \end{figure}
\end{example}

Since $h,H$, $k$, and $m$ remain fixed (unmodified) after the preprocessing, we formally consider the program $\progfldr$ to operate on $\vars_{\textsc{fldr}} = \{c,d,n\}$.
The set of initial distributions is $\Mfldr = \{\dirac{d \mapsto 0, c \mapsto 0, n \mapsto \numb{n}} \mid \numb{n} > 0\}$.

\paragraph{A Distributional Invariant for the FLDR.}%
We first define the following auxiliary notation:
$
    \isThere{n, h, H}{c} \coloneqq [\H{\bound{h}{c}}{c} = n + 1]
$
is a Boolean ($0$ or $1$) indicating if row $c$ contains a leaf labeled $n+1$.
The probability to draw $i \in \{1,\ldots,n\}$ with zero rejections is
\[
    \prob{H}{i} ~\coloneqq~ \sum_{d,c\,:\,\H{d}{c} = i} 2^{-c}
\]
and with any number of rejections, the probability is
\[
    \probt{n, H}{i} ~\coloneqq~ \frac{\prob{H}{i} }{1-\prob{H}{n+1}}~.
\]
Similarly, the probability to draw $i$ without rejections starting at level $j$ of the tree is
\[
    \prob{H}{i, j} ~\coloneqq~ \sum_{d,c\,:\,\H{d}{c} = i,\,c > j} 2^{-(c-j)}
\]
and
$
    \probt{n, H}{i, j} \coloneqq \prob{H}{i, j} + \prob{H}{n+1, j} \cdot \probt{n, H}{i}
$
with rejections.
Also recall that $k$ has been defined above as $k = \lceil \log_2 m \rceil$.
We propose the following distributional invariant $\Ifldr$ for the FLDR:
{\allowdisplaybreaks\begin{align*}
    &\Ifldr ~\coloneqq~ \Big\{\mu \in \dists ~|~ \exists \numb{n} > 0\colon \big(\prosim{\mu}{\progvar{n}}{=}{\numb{n}} ~=~ 1 \tag{a}\label{fldrA}
    \\[3pt]
    &~\land~ \prosim{\mu}{0}{\leq}{\progvar{d} \leq \numb{n}} ~=~ 1\tag{b}\label{fldrB}
    \\[5pt]
    &~\land~ \prosim{\mu}{0}{\leq}{\progvar{c} \leq k} ~=~ 1\tag{c}\label{fldrC}
    \\[5pt]
    &\begin{aligned}
        &~\land~ \forall  i \in \{1,\ldots,\numb{n}\} \colon \\
        &\sum_{\sigma\colon \!\H{\sigma(d)}{\sigma(c)} = i} \!\!\!\!\!\!\!\!\Big( \mu(\sigma) + \textstyle\sum_{j=0}^{k-1} \mu(\sigma[d/0, c/j]) \cdot \probt{\numb{n}, H}{i,j} \Big) ~\leq~ \probt{\numb{n}, H}{i} 
    \end{aligned}\tag{d}\label{fldrD}
    \\[5pt]
    &\begin{aligned}
        &~\land~ \forall \sigma,\sigma' ~\models~ \big(\progvar{n} = \numb{n} ~\land~ 0 \leq \progvar{d} < \bound{h}{\progvar{c}}\big)\colon\\
        &\qquad\qquad \sigma(\progvar{c}) = \sigma'(\progvar{c}) \implies \mu(\sigma) = \mu(\sigma')
    \end{aligned}\tag{e}\label{fldrE}
    \\[5pt]
    &\begin{aligned}
        &~\land~ \forall \sigma,\sigma' ~\models~ \big(\progvar{n} = \numb{n} ~\land~ \bound{h}{\progvar{c}} + \isThere{\numb{n},h,H}{\progvar{c}} \leq \progvar{d} < \bound{h}{\progvar{c}} + h[\progvar{c}]\big) \colon
        \\[-4pt]
        & \qquad\qquad \sigma(\progvar{c}) = \sigma'(\progvar{c}) \implies \mu(\sigma) = \mu(\sigma') \qquad\qquad\qquad\qquad\qquad\qquad \big)\Big\}
    \end{aligned} \tag{f}\label{fldrF}
\end{align*}}%
To provide some intuition for the invariant $\Ifldr$, let us briefly explain the role of each condition.
As in the FDR invariant from \Cref{sec:fdr}, conditions \eqref{fldrA}--\eqref{fldrC} ensure that $\numb n$ is fixed, and that the variables $c$ and $d$ remain within the valid bounds determined by the number of rows and the width of the binary tree, respectively.
Conditions \eqref{fldrE} and \eqref{fldrF} enforce uniformity of probabilities among program states that correspond to nodes in the same row of the binary tree: specifically, \eqref{fldrE} ensures uniform probabilities for all internal nodes within a row, while \eqref{fldrF} does the same for output nodes, except for those labelled with the rejection value $n+1$ (which always have probability $0$ since the algorithm restarts upon reaching them).
Condition \eqref{fldrD} is unique to the FLDR invariant and is essential for inductiveness: it guarantees that, for each possible output $1 \leq i \leq n$, the total probability mass assigned to output nodes labelled $i$—both now and in all future loop iterations—never exceeds the desired output probability.
\ifappendixversion%
    See \Cref{subsec:InvariantFLDR} for further details.
\else
    See~\cite[Appendix]{arxiv} for further details.
\fi

Using our invariant, we obtain the following partial and total correctness results for the FLDR:
\begin{restatable}[FLDR Correctness]{theorem}{fldrtotal}%
    \label{thm:fldrtotal}%
    \label{thm:fldrpartial}%
    Program $\progfldr$ is partially correct, i.e., for all $\mu_0 \in \Mfldr$ and $\numb{n} \in \nats_{> 0}$ with $\prosim{\mu_0}{\progvar{n}}{=}{\numb{n}} = 1$, the output distribution $\mu = \den{\progfldr}(\mu_0)$ satisfies:
    \[
        \forall 1 \leq i \leq \numb{n}\colon \sum_{\sigma\colon \H{\sigma(d)}{\sigma(c)} = i} \mu(\sigma) ~\leq~ \frac{a_i}{m} ~.
    \]
        If $\progfldr$ is AST w.r.t.\ $\Ifldr$, then program $\progfldr$ is totally correct, i.e.,
    \[
        \forall 1 \leq i \leq \numb{n}\colon \sum_{\sigma\colon \H{\sigma(d)}{\sigma(c)} = i} \mu(\sigma) ~=~ \frac{a_i}{m} ~.
    \]
\end{restatable}
\noindent As for the FDR, we have a proof that our FLDR encoding $\progfldr$ is in fact AST, using  methods from \cite{DBLP:journals/pacmpl/MajumdarS25,10.1145/3158121}
\ifappendixversion%
    (see \Cref{app:fldrast})%
\else
    (see~\cite[Appendix]{arxiv})%
\fi.

\section{Related Work}%
\label{sec:relwork}%
This paper was inspired by recent work viewing explicit finite-state Markov models as distribution transformers~\cite{DBLP:conf/lics/AkshayGV18,DBLP:conf/cav/AkshayCMZ23,DBLP:conf/ijcai/0001CMZ24,DBLP:conf/birthday/AghamovBKNOPV25}.
The main difference to these works is that the programs we study naturally induce \emph{infinite} Markov chains or \emph{infinite families}.
We also focus on \emph{limiting} distributions rather than safety properties of the reachable distributions.

Hoare rules for loops based on distributional invariants have been studied in~\cite{DBLP:journals/ijfcs/HartogV02,DBLP:conf/esop/BartheEGGHS18}.
A technical difference to our rules (\Cref{fig:whileRulesMainSec}) is that we apply a closure operation to the \emph{postcondition} rather than requiring the \emph{invariant} to be already closed.
We thus impose fewer constraints on the invariant.
While this does not play a major role in our case studies, it could be beneficial in other examples.
Similar proof rules with an extension to nondeterminism are provided by Zilberstein et al. \cite{10.1145/3743131,Zilberstein2025}.
Their (demonic) \emph{outcome logic} enables reasoning about \emph{sets} of inputs and outcomes of a program similar to our $\postn$-transformer.

Several other theoretical frameworks exist that are---at least in principle---applicable to random sampler verification (i.e., they address challenges \ref{challenge1}--\ref{challenge3} from \cpageref{challenge1}).  
A prominent example is the \emph{weakest pre-expectation} (\wpRelWork) calculus~\cite{DBLP:series/mcs/McIverM05,DBLP:phd/dnb/Kaminski19,DBLP:conf/stoc/Kozen83}, a generalization of Dijkstra's classic weakest pre\emph{condition} calculus~\cite{DBLP:journals/cacm/Dijkstra75,DBLP:books/ph/Dijkstra76}.
Using \wpRelWork, the correctness of e.g., a uniform sampling program $\prog$ is expressed as
$\wpRelWork\llbracket\prog\rrbracket([\mathit{res} = i]) = \tfrac{1}{n} \cdot [0 \leq i < n]$, where $\mathit{res}$ stores the result and $i$ does not occur in $\prog$. 
However, proving this equality requires guessing and verifying an \emph{exact} least fixed point, which is notoriously difficult.
To our knowledge, no published \wpRelWork-based proof exists for either of our two case studies.
A program analysis using \emph{probability generating functions} was proposed in \cite{DBLP:conf/lopstr/KlinkenbergBKKM20}, but \emph{comparisons} between program variables (e.g., \( v < n \)) were identified as a fundamental obstacle~\cite{DBLP:conf/cav/ChenKKW22}.
Recently, Haase et al.~\cite{haase2026generatingfunctionsmeetoccupation} provided a novel correctness proof for the FDR program from \Cref{sec:fdr}.
Their proof uses a so-called \emph{occupation invariant}---an approach orthogonal to our distributional invariants---which requires reasoning about the expected number of visits to each program state.
The resulting invariant is arguably more complex and less intuitive than the distributional invariant we employ.
Bagnall et al.~\cite{DBLP:journals/pacmpl/Bagnall0023} describe \toolzar, a tool that automatically generates formally verified executable samplers from high-level probabilistic program specifications.
A \emph{type system} for verifying sampling algorithms built from pre-defined deterministic samplers based on infinite streams was proposed in \cite{DBLP:conf/lics/Dahlqvist0S23}. It appears promising to investigate whether the FDR could be verified assuming a pre-defined sampler for fair coin flips, though this would likely require extending their type system with unbounded iteration.

The related problem of inferring the \emph{posterior distribution} of probabilistic programs with \emph{conditioning} has also been studied extensively.
The tool \toolpsi~\cite{DBLP:conf/cav/GehrMV16} computes exact expressions for density functions but does not support unbounded while-loops~\cite[Sec.~3.6]{DBLP:conf/cav/GehrMV16}, which are required by the F(L)DR algorithms.
Other approaches \cite{wang2024static,DBLP:conf/pldi/BeutnerOZ22} support such loops, but provide only \emph{approximations}.

Additional related work includes \emph{testing} random samplers with statistical guarantees~\cite{DBLP:conf/aaai/ChakrabortyM19} and \emph{equivalence refutation} for probabilistic programs, which can prove \emph{incorrectness}, i.e., the presence of a bug in a given sampling algorithm~\cite{DBLP:journals/pacmpl/ChatterjeeGNZ24}.
Meel et al.~\cite{DBLP:conf/cav/SarkarCM25} assess the quality of binomial samplers and verify the program code directly using a statistical distance framework.

\section{Conclusion and Future Work}%
\label{sec:conclusion}%
This paper provided formal correctness proofs for two efficient, non-trivial sampler algorithms: the FDR and the FLDR, which are key to a broad range of applications. We presented a framework to verify quantitative properties of probabilistic programs by viewing programs as distribution transformers.

For future work, we plan to investigate assertion languages for distributions, which would make reasoning with distributional invariants more syntactic and thus more amenable to automation.
A particularly promising direction is to leverage the assertion language provided by Barthe et al.\ \cite{DBLP:conf/esop/BartheEGGHS18}, who also provide a weakest precondition-style approach for reasoning about reachable distributions.
Another opportunity for future work is to extend our verification framework toward probabilistic programs with non-determinism~\cite{DBLP:series/mcs/McIverM05}---this appears feasible since the current approach already deals with sets of distributions.

%
%
\ifappendixversion
    \newpage
    \appendix
    \section{Additional Preliminaries}
\label[appendix]{app:details}

\paragraph{Fixed Point Theory.}
A sequence of distributions $(\mu_n)_{n \in \nats}$ is \emph{ascending} if $\mu_n < \mu_{n+1}$ for every $n \in \nats$ and a sequence is called an $\omega$\emph{-chain} if $\mu_n \leq \mu_{n+1}$ for every $n \in \nats$.  The relation $\leq$ is to be understood componentwise, i.e. $\mu_1 \leq \mu_2$ holds iff $\mu_1(\sigma) \leq \mu_2(\sigma)$ for all $\sigma \in \Sigma$. Similarly, the \emph{supremum} of an $\omega$-chain is to be understood componentwise, i.e., $(\sup (\mu_n)_{n \in \nats})(\sigma) = \sup \{\mu_n(\sigma) ~|~ n \in \nats\} = \lim_{n \rightarrow \infty} \mu_n(\sigma)$.

A \emph{fixed point} of a function $f\colon D \rightarrow D$ is an element $d \in D$ such that $f(d) = d$. A \emph{least} fixed point is a fixed point $d$ of $f$ with $d \sqsubseteq d'$ for any fixed point $d'$ of $f$, with respect to a partial order $\sqsubseteq$ on $D$. Recall that a partial order is a reflexive, transitive and antisymmetric binary relation. A subset $S \sqsubseteq D$ is called a \emph{chain} if, for all $s_1, s_2 \in S\colon s_1 \sqsubseteq s_2$ or $s_2 \sqsubseteq s_1$. An element $d \in D$ is called an upper bound of $S$ if $s \sqsubseteq d$ for all $s \in S$. The least upper bound (supremum) of $S$, denoted $\bigsqcup S$ or $\sup S$, is an upper bound $d$ of $S$ such that for every upper bound $d'$ of $S$, we have $d \sqsubseteq d'$. A partial order on $D$ is called \emph{chain complete} (CCPO) if each of its chains has a least upper bound in $D$.

A function $F\colon D \rightarrow D'$ between partial orders $D$ and $D'$ respectively is called monotonic if, for every $d_1, d_2 \in D$ with $d_1 \sqsubseteq d_2$ it holds that $F(d_1) \sqsubseteq' F(d_2)$. A function $F$ is called continuous \cite{10.7551/mitpress/3054.001.0001} if, for every non-empty chain $S \subseteq D$, it holds that $F(\bigsqcup S) = \bigsqcup F(S)$. The fixpoint theorem by Kleene gives us insights about the least fixed point of a function.
\begin{theorem}[Kleene's Fixed Point Theorem~\cite{10.7551/mitpress/3054.001.0001}]%
    \label{thm:kleene}%
	Let $(D, \sqsubseteq)$ be a CCPO and $F\colon D \rightarrow D$ continuous. Then the least fixed point of $F$ is
	\[
		\lfp F ~=~ \bigsqcup\{F^n(\bigsqcup \emptyset) ~|~ n \in \nats\} ~.
	\]
\end{theorem}

    \section{Appendix to \Cref{sec:post}}
\label{app:post}

\paragraph{Comment on the definition of $\postn$:}
Unlike for e.g., $\den{\prog}$ and classic $\postn$ in deterministic programs (see e.g.,~\cite[Ch.~12]{DBLP:books/daglib/0067387}), we do not give a \enquote{direct} inductive definition of $\post{\prog}$. Notice that the following intuitive expression for the probabilistic choice case yields only an overapproximation of $\post{\prog}$ and thus does not infer an inductive definition:
\[
	\probability \cdot \post{\progOne}(\initSet) + (1 - \probability) \cdot \post{\progTwo}(\initSet) \supseteq \post{\pchoice{\progOne}{\probability}{\progTwo}}(\initSet) ~.
\]
It remains an open question whether such kind of definition for $\post{\prog}$ exists.

\subsection{Denotational Semantics: Additional Example}%

\begin{example}[Denotational Semantics]
    In \Cref{fig:denExample} (left) we determine the denotational semantics of a loop-free program for a given (Dirac) input distribution.
    \Cref{fig:denExample} (right) shows a loopy program---again with a fixed initial distribution---and the associated fixed point iteration.
    In this case, the iteration terminates after finitely many steps.
    The loop \enquote{leaks} half of the initial probability mass as it terminates only with probability $\nicefrac 1 2$.
\end{example}

\begin{figure}[t]
    \begin{minipage}{0.55\textwidth}
        \setlength{\jot}{1.25pt} 
        \begin{align*}
            %
            & \progAnno{\dirac{x \mapsto 0, y \mapsto 0}} \tag{initial distribution} \\
            & \IF{x = 0} \\
            & \quad \progAnno{\dirac{x \mapsto 0, y \mapsto 0}} \\ 
            & \quad \pchoice{\ass{x}{1}}{\nicefrac 1 2}{\ass{y}{1}} \\
            & \quad \progAnno{\tfrac 1 2 \dirac{x \mapsto 1, y \mapsto 0} + \tfrac 1 2 \dirac{x \mapsto 0, y \mapsto 1}} \\
            & \prbracecl \SEMICLN \progAnno{\tfrac 1 2 \dirac{x \mapsto 1, y \mapsto 0} + \tfrac 1 2 \dirac{x \mapsto 0, y \mapsto 1}} \\
            & \IF{x = 1} \\
            & \quad \progAnno{\tfrac 1 2 \dirac{x \mapsto 1, y \mapsto 0}} \\
            & \quad \pchoice{\ass{x}{0}}{\nicefrac 1 3}{\ass{y}{1}} \\
            & \quad \progAnno{\tfrac 1 6 \dirac{x \mapsto 0, y \mapsto 0} + \tfrac 1 3 \dirac{x \mapsto 1, y \mapsto 1}} \\
            & \prbracecl \, \progAnno{\tfrac 1 6 \dirac{x \mapsto 0, y \mapsto 0} + \tfrac 1 3 \dirac{x \mapsto 1, y \mapsto 1} + \tfrac 1 2 \dirac{x \mapsto 0, y \mapsto 1}} 
            %
        \end{align*}
    \end{minipage}
    \hfill
    \begin{minipage}{0.35\linewidth}
        \setlength{\jot}{1.25pt} 
        \begin{align*}
            & \progAnno{\dirac{x \mapsto 0, y \mapsto 1}} \tag{${}={} \mu_0$} \\
            & \WHILE{x<1} \\
            & \quad \pchoice{\ass{y}{0}}{\nicefrac 1 2}{\ass{x}{x+y}} \\
            & \prbracecl \\
            & \progAnno{\tfrac 1 2 \dirac{x \mapsto 1, y \mapsto 1}}
        \end{align*}
        
        {\setlength{\tabcolsep}{6pt}
        \renewcommand{\arraystretch}{1.1}
        \begin{tabular}{r l}
            $i$ & $\Phi^i(\mathbf{0})(\mu_0)$ \\ 
            \midrule
            $0$ & $0$ \\
            $1$ & $\dirac{x \mapsto 0, y \mapsto 1} \textcolor{gray}{~{}={}~ \mu_0}$ \\
            $2$ & $\tfrac 1 2 \dirac{x \mapsto 1, y \mapsto 1}$ \\
            $\geq 3$ & $\tfrac 1 2 \dirac{x \mapsto 1, y \mapsto 1}$ 
        \end{tabular}}
    \end{minipage}
    \caption{%
        \textit{Left:} Loop-free example program. The denotational semantics is determined top-to-bottom.
        \textit{Right:} A loop and its associated fixed point iteration.
    }
    \label{fig:denExample}
\end{figure}

\subsection{Proof of \Cref{thm:postIsSup}}
\label{app:subsec:postIsSupProof}

\postIsSup*
\begin{proof}
    Let $\initSet = \{\mu_0\}$.
    We have
    \begin{align*}
        & \reach{\prog}(\{\mu_0\}) \\
        \eeq & \bigcup_{n \in \nats} \post{\ifelseskip{\bool}{\loopbody}}^n(\{\mu_0\}) \tag{\Cref{cor:reach}} \\
        \eeq & \bigcup_{n \in \nats} \left\{\den{\ifelseskip{\bool}{\loopbody}}^n(\mu_0)\right\} ~.\tag{Definition of $\post{\dots}$ for singletons}
    \end{align*}
    We first show that $[\neg \bool] \cdot \reach{\prog}(\{\mu_0\})$ is a chain. Let $n \in \nats$, and let
    \[
    	\mu = \den{\ifelseskip{\bool}{\loopbody}}^n(\mu_0)
    \]
    and
    \[
    	\mu' = \den{\ifelseskip{\bool}{\loopbody}}^{n+1}(\mu_0) ~.
    \]
    By definition of the denotational semantics, $\mu' = [\neg \bool] \cdot \mu + \den{\loopbody}([\bool]\cdot\mu)$ which yields that $[\neg \bool] \cdot \mu' \geq [\neg \bool] \cdot \mu$. Thus, we conclude that $[\neg \bool] \cdot \reach{\prog}(\{\mu_0\})$ is a chain.
    
    It remains to show the equation in the theorem.
    Let $\sigma \in \Sigma$ be a program state.
    If $\sigma \models \bool$, 
    \[
    	\den{\prog}(\mu_0)(\sigma) = 0 = \sup \big( [\neg \bool] \cdot \reach{\prog}(\{\mu_0\}) \big)(\sigma) ~.
    \]
    If $\sigma \not\models \bool$, we show by induction on $n \in \nats$, that for the characteristic function of the LFP for the denotational semantics, i.e. $\Phi(f) = \lambda \mu'.[\neg \bool] \cdot \mu' + f(\den{\loopbody}([\bool] \cdot \mu'))$, it holds for all $n \in \nats$ that
    \[
    	\Phi^{n}(0)(\mu_0)(\sigma) = \den{\ifelseskip{\bool}{\loopbody}}^n(\mu_0)(\sigma) ~.
    \]
    For $n = 0$, we have that
    \[
    	\Phi^{n}(0)(\mu_0)(\sigma) = \mu_0(\sigma) = \den{\ifelseskip{\bool}{\loopbody}}^n(\mu_0)(\sigma) ~.
    \]
    
    For $n > 0$, we have that
    \begin{align*}
    	\Phi^{n}(0)(\mu_0)(\sigma) & = \Phi(\Phi^{n-1}(0))(\mu_0)(\sigma) \tag{$\Phi^{n}(0) = \Phi(\Phi^{n-1}(0)$}\\
    	& = [\neg \bool] \mu_0(\sigma) + \Phi^{n-1}(0)(\den{\loopbody}([\bool]\mu_0))(\sigma) \tag{Definition $\Phi(\Phi^{n-1}(0))(\mu_0)$}\\
    	& = [\neg \bool] \mu_0(\sigma) + \den{\ifelseskip{\bool}{\loopbody}}^{n-1}(\den{\loopbody}([\bool]\mu_0))(\sigma)\tag{I.H.}\\
    	& = \den{\ifelseskip{\bool}{\loopbody}}^{n-1}([\neg \bool] \mu_0)(\sigma)\\
     & \phantom{MMM} + \den{\ifelseskip{\bool}{\loopbody}}^{n-1}(\den{\loopbody}([\bool]\mu_0))(\sigma) \tag{Since $[\bool][\neg\bool]\mu_0 = 0$}\\
    	& = \den{\ifelseskip{\bool}{\loopbody}}^{n-1}([\neg \bool] \mu_0 + \den{\loopbody}([\bool]\mu_0))(\sigma) \tag{Linearity of Den. Sem.}\\
    	& = \den{\ifelseskip{\bool}{\loopbody}}^{n-1}(\den{\ifelseskip{\bool}{\loopbody}}(\mu_0))(\sigma) \tag{Definition of $\den{\ifelseskip{\bool}{\loopbody}}(\mu_0)$}\\
    	& = \den{\ifelseskip{\bool}{\loopbody}}^n(\mu_0)(\sigma) \tag{$\den{\prog}^n(\dots) = \den{\prog}^{n-1}(\den{\prog}(\dots))$}
    \end{align*}
    which concludes the induction. Using the result of the induction, i.e.
    \[
        \forall n \in \nats\colon \Phi^{n}(0)(\mu_0)(\sigma) = \den{\ifelseskip{\bool}{\loopbody}}^n(\mu_0)(\sigma),
    \]
    we obtain the final result for the proof of the theorem:
    \begin{align*}
        \den{\prog}(\mu_0)(\sigma) & \eeq \sup \{\Phi^{n}(0) ~|~ n \in \nats\}(\mu_0)(\sigma) \tag{Definition of Den. Sem.}\\
        & \eeq \sup \{\Phi^{n}(0)(\mu_0)(\sigma) ~|~ n \in \nats\} \tag{$\sigma \models \neg \bool$}\\
        & \eeq \sup \{\den{\ifelseskip{\bool}{\loopbody}}^n(\mu_0)(\sigma) ~|~ n \in \nats\}\tag{Induction Result}\\
        & \eeq \sup \{[\neg \bool] \cdot \den{\ifelseskip{\bool}{\loopbody}}^n(\mu_0)(\sigma) ~|~ n \in \nats\}\tag{$\sigma \models \neg \bool$}\\
        & \eeq \sup \big( [\neg \bool] \cdot \reach{\prog}(\{\mu_0\}) \big)(\sigma)~.\tag{\Cref{cor:reach}}
    \end{align*}
\end{proof}

    \section{Appendix to \Cref{sec:distis}}
\label{app:distis}

\subsection{A Small-step MDP Translation}
\label{subsec:smallsteptr}
In \Cref{sec:distis}, we provide an operational markov chain of a probabilistic loop in the spirit of a big-step operational semantics. In the literature \cite{DBLP:journals/pacmpl/BatzBKW24}, a small-step semantics is the more common approach and can be defined (for our syntax of probabilistic programs) as follows:
For the translation from a probabilistic \emph{and even nondeterministic} program $\prog$ to its underlying MDP $\mdpTr$ via a small-step semantics, we extend our definition for probabilistic programs by nondeterministic choices $\prog = \nchoice{\progOne}{\progTwo}$. Then, we define the set of program configurations as $\conf \coloneqq (\pgcl \cup \{\Downarrow\}) \times \Sigma$. The small-step execution relation of the form $\rightarrow~ \subseteq \conf \times \{\tau,\alpha,\beta\} \times [0,1] \times \conf$ is given by the rules in \Cref{fig:smallstep} (adapted from \cite{DBLP:journals/pacmpl/BatzBKW24}, Figure 18). Based on the small-step execution relation, a small-step operational MDP is defined.

\begin{definition}
The small-step operational MDP $\mdpTr$ of a probabilistic and nondeterministic program is $\mdpTr = (\states, \act, \prmdp)$, where
\begin{itemize}
\item $\states = \{(\prog',\sigma') \in \conf ~|~ \exists \sigma \in \Sigma\colon (\prog',\sigma') \textnormal{ is reachable from } (\prog, \sigma)\}$
\item $\act = \{\tau,  \alpha, \beta\}$
\item $\prmdp\colon \states \times \act \times \states \rightarrow [0,1]$ is the transition probability function given by
\[
	\prmdp(t,l,t') =
	\begin{cases}
		q & \textnormal{ if } t \overset{l,q}{\rightarrow} t'\\
		0 & \textnormal{ otherwise.}
	\end{cases}
\]
\end{itemize}
\end{definition}

\begin{figure}[t]
\begin{center}
\begin{tabular}{ccc}
    \AxiomC{$$}
\UnaryInfC{$\ski,\sigma \smallstep{\tau}{1} \Downarrow,\sigma$}
\DisplayProof
&
    \AxiomC{$$}
\UnaryInfC{$\ass{x}{\arith}, \sigma \smallstep{\tau}{1} \Downarrow, \sigma[x/\arith(\sigma)]$}
\vspace{0.5cm}
\DisplayProof\\
\vspace{0.5cm}

    \AxiomC{$\progOne \ne \progTwo$}
\UnaryInfC{$\pchoice{\progOne}{\probability}{\progTwo},\sigma \smallstep{\tau}{\probability} \progOne,\sigma$}
\DisplayProof
&
    \AxiomC{$\progOne \ne \progTwo$}
\UnaryInfC{$\pchoice{\progOne}{\probability}{\progTwo},\sigma \smallstep{\tau}{1 - \probability} \progTwo,\sigma$}
\DisplayProof\\
\vspace{0.5cm}
    \AxiomC{$$}
\UnaryInfC{$\Downarrow, \sigma \smallstep{\tau}{1} \Downarrow, \sigma$}
\DisplayProof
&
    \AxiomC{$$}
\UnaryInfC{$\pchoice{\prog}{\probability}{\prog},\sigma \smallstep{\tau}{1} \prog,\sigma$}
\DisplayProof\\
\vspace{0.5cm}

    \AxiomC{$$}
\UnaryInfC{$\nchoice{\progOne}{\progTwo}, \sigma \smallstep{\alpha}{1} \progOne,\sigma$}
\DisplayProof
&
    \AxiomC{$$}
\UnaryInfC{$\nchoice{\progOne}{\progTwo}, \sigma \smallstep{\beta}{1} \progTwo,\sigma$}
\DisplayProof
&
\\
\vspace{0.5cm}

    \AxiomC{$\progOne, \sigma \smallstep{l}{q} \Downarrow, \sigma'$}
\UnaryInfC{$\comp{\progOne}{\progTwo}, \sigma \smallstep{l}{q} \progTwo,\sigma'$}
\DisplayProof
&
    \AxiomC{$\progOne, \sigma \smallstep{l}{q} \progOne', \sigma'$}
\UnaryInfC{$\comp{\progOne}{\progTwo}, \sigma \smallstep{l}{q} \comp{\progOne'}{\progTwo},\sigma'$}
\DisplayProof
&
\\
\vspace{0.5cm}

    \AxiomC{$\sigma \models \bool$}
\UnaryInfC{$\ifelse{\bool}{\progOne}{\progTwo}, \sigma \smallstep{\tau}{1} \progOne,\sigma$}
\DisplayProof
&
    \AxiomC{$\sigma \not\models \bool$}
\UnaryInfC{$\ifelse{\bool}{\progOne}{\progTwo}, \sigma \smallstep{\tau}{1} \progTwo,\sigma$}
\DisplayProof
&
\\
\vspace{0.5cm}

    \AxiomC{$\sigma \not\models \bool$}
\UnaryInfC{$\while{\bool}{\loopbody}, \sigma \smallstep{\tau}{1} \Downarrow,\sigma$}
\DisplayProof
&
    \AxiomC{$\sigma \models \bool$}
\UnaryInfC{$\while{\bool}{\loopbody}, \sigma \smallstep{\tau}{1} \comp{\loopbody}{\while{\bool}{\loopbody}},\sigma$}
\DisplayProof
&
\end{tabular}
\end{center}
\caption{The rules defining the small-step execution relation $\rightarrow~ \subseteq \conf ~\times~ \{\tau,\alpha,\beta\} ~\times~ [0,1] ~\times~ \conf$ for probabilistic and nondeterministic programs.}
\label{fig:smallstep}
\end{figure}

\subsection{Proof of \Cref{thm:translation}}
\label{subsec:prooftranslation}

\translation*

We first show the following lemma:
\begin{lemma}
    \label{thm:mcCorrespHelper}
    Let $\bool$, $\loopbody$, $\mu_0$ and $\prog$ as in \Cref{thm:translation}.
    For all $n \in \nats$ it holds that
    \[
    \den{\ifelseskip{\bool}{\loopbody}}^{n}(\mu_0)
    \eeq
    \mdpTr^n(\mu_0)
    ~.
    \]
\end{lemma}
\begin{proof}[of \Cref{thm:mcCorrespHelper}]
    By induction on $n$.
    The claim holds trivially for $n=0$.
    For $n\geq 0$ we argue as follows:
    \begin{align*}
        & \den{\ifelseskip{\bool}{\loopbody}}^{n+1}(\mu_0) \\
        \eeq & \den{\ifelseskip{\bool}{\loopbody}}\left( \mdpTr^n(\mu_0) \right) \tag{I.H.} \\
        \eeq & [\bool] \cdot \den{\loopbody}\left( \mdpTr^n(\mu_0) \right) + [\neg\bool] \cdot \mdpTr^n(\mu_0) \tag{Definition of $\denn$} \\
        \eeq & [\bool] \cdot \den{\loopbody}\left( \sum_{\sigma \in \Sigma} \mdpTr^n(\mu_0)(\sigma) \cdot \dirac{\sigma} \right) + [\neg\bool] \cdot \left( \sum_{\sigma \in \Sigma} \mdpTr^n(\mu_0)(\sigma) \cdot \dirac{\sigma} \right) \tag{Rewrite in terms of Dirac distributions} \\
        \eeq & \den{\loopbody}\left( \sum_{\sigma\models \bool} \mdpTr^n(\mu_0)(\sigma) \cdot \dirac{\sigma} \right) +   \sum_{\sigma \not\models \bool} \mdpTr^n(\mu_0)(\sigma) \cdot \dirac{\sigma}  \tag{Simplification} \\
        \eeq & \sum_{\sigma\models \bool} \mdpTr^n(\mu_0)(\sigma) \cdot  \den{\loopbody}\left( \dirac{\sigma} \right) + \sum_{\sigma \not\models \bool} \mdpTr^n(\mu_0)(\sigma) \cdot    \dirac{\sigma}  \tag{Linearity of $\den{\loopbody}$} \\
        \eeq & \sum_{\sigma\models \bool} \mdpTr^n(\mu_0)(\sigma) \cdot  \prmdp(\sigma) + \sum_{\sigma \not\models \bool} \mdpTr^n(\mu_0)(\sigma) \cdot   \prmdp(\sigma)  \tag{Definition of $\prmdp$} \\
        \eeq & \sum_{\sigma\in \Sigma} \mdpTr^n(\mu_0)(\sigma) \cdot  \prmdp(\sigma)  \tag{Simplification} \\
        \eeq & \mdpTr(\mdpTr^n(\mu_0)) \tag{Definition of $\mdpTr$} \\
        \eeq & \mdpTr^{n+1}(\mu_0) 
    \end{align*}
\end{proof}

We can now prove \Cref{thm:translation}:
\begin{proof}[of \Cref{thm:translation}]
    \begin{align*}
        & \inv \text{ is a distributional invariant w.r.t.\ } \prog \text{ and } \{\initDist\} \\
        \iiff & \reach{\prog}(\{\mu_0\}) \subseteq \inv \tag{\Cref{defn:distInv}} \\
        \iiff & \forall n \in \nats \colon \post{\ifelseskip{\bool}{\loopbody}}^n\left(\{\mu_0\}\right) \subseteq \inv \tag{\Cref{cor:reach}} \\
        \iiff & \forall n \in \nats \colon \den{\ifelseskip{\bool}{\loopbody}}^n\left(\mu_0\right) \in \inv \tag{Definition of $\postn$} \\
        \iiff & \forall n \in \nats \colon \mdpTr^n\left(\mu_0\right) \in \inv \tag{\Cref{thm:mcCorrespHelper}} \\
        \iiff & \inv \text{ is a distributional invariant w.r.t.\ } \mdpTr \text{ and } \initDist \tag{\Cref{defn:distInvMDP}}
    \end{align*}
    The argument for inductivity is similar:
    \begin{align*}
        & \inv \text{ is inductive w.r.t.\ } \prog  \\
        %
        %
        \iiff & \forall \mu \in \inv \colon \den{\ifelseskip{\bool}{\loopbody}}\left(\mu\right) \in \inv \tag{\Cref{defn:distInv}} \\
        \iiff & \forall \mu \in \inv  \colon \mdpTr\left(\mu\right) \in \inv \tag{\Cref{thm:mcCorrespHelper}} \\
        \iiff & \inv \text{ is inductive w.r.t.\ } \mdpTr \tag{\Cref{defn:distInvMDP}}
    \end{align*}
\end{proof}

    \section{Appendix to \Cref{sec:verification}}
\label{app:hoare}

\paragraph{Comments on Hoare Triples} Notice that $\models\tripleGen$ if and only if $\forall \mu \in \hSetOne\colon \den{\prog}(\mu) \in \hSetTwo$.
Furthermore, inductivity of $\inv \subseteq \subdists$ w.r.t.\ $\while{\bool}{\loopbody}$ can be phrased as $\models \triple{\inv}{\ifelseskip{\bool}{\loopbody}}{\inv}$.
We only consider inductive invariants in this section and focus on inference rules for probabilistic loops.
Rules for other program statements are provided in~\cite{DBLP:journals/ijfcs/HartogV02,DBLP:conf/esop/BartheEGGHS18} (also see \Cref{subsec:loopfreerules}).

\begin{remark}[Relation to Classic Hoare Triples]
    \Cref{def:hoareTriples} generalizes traditional Hoare triples for deterministic programs, both the ones for total as well as for partial correctness.
    Indeed, if $\prog$ is deterministic and $\hSetOne,\hSetTwo$ are sets of Dirac distributions, then $\tripleGen$ encodes \emph{total} correctness, i.e.\ if $\prog$ starts with a state from $\hSetOne$, then $\prog$ terminates in a state from $\hSetTwo$.
    \emph{Partial} correctness, on the other hand, can be encoded as $\triple{\hSetOne}{\prog}{\hSetTwo \cup \{0\}}$, where $0$ is the constant zero distribution allowing for non-termination.
\end{remark}

\subsection{Hoare Rules for Loop-free Programs}
\label{subsec:loopfreerules}
We give rules with which valid Hoare triples can be derived. The soundness of the following rules for loop-free programs is easy to see.

\begin{center}
\begin{tabular}{ll}
    \AxiomC{$$}
\LeftLabel{(Skip):\;} 
\UnaryInfC{$\triple{\hSetOne}{\ski}{\hSetOne}$}
\DisplayProof
&
    \AxiomC{$\triple{\hSetOne}{\progOne}{\hSetTwo}$}
    \AxiomC{$\triple{\hSetTwo}{\progTwo}{\hSetThree}$}
\LeftLabel{(Comp):\;} 
\BinaryInfC{$\triple{\hSetOne}{\comp{\progOne}{\progTwo}}{\hSetThree}$}
\vspace{1cm}
\DisplayProof\\

    \AxiomC{$$}
\LeftLabel{(Assign):\;} 
\UnaryInfC{$\triple{\hSetOne[\progvar{x}/\arith]}{\ass{\progvar{x}}{\arith}}{\hSetOne}$}
\DisplayProof
&
    \AxiomC{$\triple{[\bool]\hSetOne}{\progOne}{\hSetTwo}$}
    \AxiomC{$\triple{[\neg \bool]\hSetOne}{\progTwo}{\hSetThree}$}
\LeftLabel{(If):\;} 
\BinaryInfC{$\triple{\hSetOne}{\ifelse{\bool}{\progOne}{\progTwo}}{\hSetTwo + \hSetThree}$}
\vspace{1cm}
\DisplayProof

\end{tabular}
\begin{tabular}{c}

    \AxiomC{$\triple{\hSetOne}{\progOne}{\hSetTwo}$}
    \AxiomC{$\triple{\hSetOne}{\progTwo}{\hSetThree}$}
\LeftLabel{(Prob):\;} 
\BinaryInfC{$\triple{\hSetOne}{\pchoice{\progOne}{\probability}{\progTwo}}{\probability \cdot \hSetTwo + (1 - \probability) \cdot \hSetThree}$}
\vspace{1cm}
\DisplayProof\\

    \AxiomC{$\triple{\hSetOne}{\prog}{\hSetThree}$}
    \AxiomC{$\triple{\hSetTwo}{\prog}{\hSetThree}$}
\LeftLabel{(Or):\;} 
\BinaryInfC{$\triple{\hSetOne \cup \hSetTwo}{\prog}{\hSetThree}$}
\vspace{1cm}
\DisplayProof\\

    \AxiomC{$\hSetOne \subseteq \hSetTwo$}
    \AxiomC{$\triple{\hSetOne}{\prog}{\hSetThree}$}
    \AxiomC{$\hSetThree \subseteq \hSetFour$}
\LeftLabel{(Imp):\;} 
\TrinaryInfC{$\triple{\hSetOne}{\prog}{\hSetFour}$}
\DisplayProof

\end{tabular}
\end{center}

\subsection{Hoare rules for loops}
\label[appendix]{app:hoareRules}
Towards developing Hoare rules for probabilistic loops, it is instructive to consider the following rule first.
It is sound for deterministic programs, but generally \emph{un}sound for probabilistic ones:
\[
    \frac{\triple{\inv}{\ifelseskip{\bool}{\loopbody}}{\inv}}
    {\triple{\inv}{\while{\bool}{\loopbody}}{[\bool] \cdot \inv}}
    \tag{$\lightning$\,\textsc{While-Unsound}\,$\lightning$}
    \label{WhileUnsound}
\]
The root cause is that $\mu_0 \in \inv$ does \emph{not} necessarily mean that $\den{\prog}(\mu_0) \in \reach{\prog}(\{\mu_0\})$ (see \Cref{ex:geodist}).
We treat this issue by applying a suitable \emph{closure operator} to the post (\Cref{defn:clos} below).
This is inspired by~\cite{DBLP:journals/ijfcs/HartogV02,DBLP:conf/esop/BartheEGGHS18}.


\begin{definition}[$\omega$-Closure]%
    \label{defn:clos}%
	For a set $\initSet \subseteq \subdists$ we define the \emph{$\omega$-closure}
	\[
		\closOf\initSet
        ~\coloneqq~
        \left\{ \sup (\mu_n)_{n \in \nats} \mid (\mu_n)_{n \in \nats} \text{ is an $\omega$-chain in }\initSet\right\}
        ~.
	\]
\end{definition}
In other words, $\closOf{\initSet}$ is the smallest $\omega$-complete partial order containing $\initSet$.
We say that $\initSet$ is \emph{closed} if $\closOf{\initSet} = \initSet$.

\begin{definition}[Almost Sure Termination]%
    \label{def:ast}%
    Let $\prog = \while{\bool}{\loopbody}$ be a loop and let $\initSet \subseteq \subdists$ be a set of initial distributions. $\prog$ is \emph{almost surely terminating (AST)} w.r.t.\ $\initSet$ if $\forall \mu \in \initSet\colon \weight{\mu} = \weight{\den{\prog}(\mu)}$. Furthermore, $\prog$ is called \emph{universally} AST if $\prog$ is AST w.r.t.\ $\subdists$.
\end{definition}

\begin{example}[AST]%
    \label{ex:ast}%
    We state (without proof) that the program from \Cref{fig:operationalMC} is AST.
    The intuition is that in every loop iteration, there is a chance of $\geq \nicefrac 1 8$ of falsifying the loop guard. This is inevitable in the long run, i.e.\ with probability one, the loop guard becomes false eventually.
\end{example}

\begin{example}[Reasoning with the \textnormal{\textsc{While}} Rules]%
    We apply \eqref{WhileT} to the loop $\prog = \while{\bool}{\loopbody}$ from \Cref{fig:operationalMC}.
    We have already stated an inductive invariant $\inv$ in \Cref{ex:tobiInvariant}, i.e.\ we know $\models \triple{\inv}{\ifelseskip{\bool}{\loopbody}}{\inv}$.
    Further, $\prog$ is AST (see \Cref{ex:ast}) and we have $\inf_{\mu \in \inv} \weight{\mu} = 1$ by definition of $\inv$.
    Applying \eqref{WhileT} thus yields $\models \triple{\inv}{\while{\bool}{\loopbody}}{\closOf{[\neg \bool] \inv} \cap \dists}$.
    Finally, we notice that $[\neg\bool]\inv$ is the set of distributions assigning arbitrary, but equal probability mass to states $x\mapsto 0, y \mapsto 1$ and $x\mapsto 1, y \mapsto 1$.
    Thus, $\closOf{[\neg\bool]\inv} = [\neg\bool]\inv$, i.e.\ $[\neg\bool]\inv$ is closed, and we find that
    \[
        \closOf{[\neg\bool]\inv} \cap \dists
        ~=~
        \left\{\tfrac 1 2 \dirac{x\mapsto 0, y \mapsto 1} + \tfrac 1 2 \dirac{x\mapsto 1, y \mapsto 1} \right\}
        ~.
    \]
    Overall, we have shown that if $\prog$ starts with a proper distribution $\mu_0 \in \inv$ (meaning that $x\mapsto 1, y \mapsto 0$ is twice as likely as $x\mapsto 0, y \mapsto 0$, and $x\mapsto 0, y \mapsto 1$ is just as likely as $x\mapsto 1, y \mapsto 1$), then $\den{\prog}(\mu_0) = \tfrac 1 2 \dirac{x\mapsto 0, y \mapsto 1} + \tfrac 1 2 \dirac{x\mapsto 1, y \mapsto 1}$.
    We have thus determined the denotational semantics \emph{exactly} for all such initial $\mu_0$.
    
    The weaker \eqref{WhileP} rule can be applied similarly to $\prog$; the difference to \eqref{WhileT} is that we only obtain $\models \triple{\inv}{\while{\bool}{\loopbody}}{\closOf{[\neg \bool] \inv}}$, where
    \[
        \closOf{[\neg\bool]\inv}
        ~=~
        \left\{p \dirac{x\mapsto 0, y \mapsto 1} + p \dirac{x\mapsto 1, y \mapsto 1} \mid 0 \leq p \leq \tfrac 1 2 \right\}
        ~.
        \tag{$\clubsuit$}
        \label{eq:tobiExPartial}
    \]
    In particular, the Hoare triple $\triple{\inv}{\while{\bool}{\loopbody}}{\closOf{[\neg \bool] \inv}}$ does not reveal whether $\prog$ terminates with positive probability on any $\mu_0 \in \inv$.
    The triple thus indicates \emph{partial correctness} only.
\end{example}

\label{subsec:novelwhilesound}
\begin{restatable}{theorem}{novelWhileSound}%
    \label{thm:novelWhileSound}%
    Rules \eqref{WhileP} and \eqref{WhileT} are sound.
\end{restatable}
\begin{proof}
Let $\prog = \while{\bool}{\loopbody}$. For the soundness of \eqref{WhileP}, it is to prove that for all $\mu_0 \in \inv$ it holds that $\den{\prog}(\mu_0) \in \closOf{[\neg \bool] \inv}$ and for the soundness of \eqref{WhileT} respectively that for all $\mu_0 \in \inv$ it holds that $\den{\prog}(\mu_0) \in \closOf{[\neg \bool] \inv} \cap \Deltaspec{\geq \tau}\Sigma$ where $\tau = \inf_{\mu \in \inv} \weight{\mu}$.

Let $\mu_0 \in \inv$. According to \Cref{thm:postIsSup}, $[\neg \bool] \cdot \reach{\prog}(\{\mu_0\})$ is an $\omega$-chain and furthermore we show that the chain is contained in $[\neg \bool] \inv$:
\begin{align*}
    [\neg \bool] \cdot \reach{\prog}(\{\mu_0\}) & \overset{\{\mu_0\} \subseteq \inv}{\subseteq} [\neg \bool] \cdot \reach{\prog}(\inv) \overset{\inv~ \textnormal{invariant}}{\subseteq} [\neg \bool] \inv~.
\end{align*}
By definition of the $\omega$-closure it then follows that
\[
    \sup\big( [\neg \bool] \cdot \reach{\prog}(\{\mu_0\}) \big) \in \closOf{[\neg \bool] \inv}
\]
and finally by \Cref{thm:postIsSup}:
\begin{align*}
    \den{\prog}(\mu_0) & \overset{\Cref{thm:postIsSup}}{=} \sup\big( [\neg \bool] \cdot \reach{\prog}(\{\mu_0\}) \big) \in \closOf{[\neg \bool] \inv}
\end{align*}

which concludes soundness of \eqref{WhileP}. For the soundness of \eqref{WhileT}, we additionally need to show that under the assumption of $\prog$ being AST w.r.t $\inv$, it holds for all $\mu_0 \in \inv$ that
\[
    \den{\prog}(\mu_0) \in \Deltaspec{\geq \tau}\Sigma~.
\]
By the definition of $\prog$ being AST w.r.t $\inv$, we obtain for all $\mu_0 \in \inv$ that
\begin{align*}
    \weight{\den{\prog}(\mu_0)} = \weight{\mu_0} \overset{\mu_0 \in \inv}{\geq} \inf_{\mu \in \inv} \weight{\mu} = \tau
\end{align*}
which concludes soundness of \eqref{WhileT}.
\end{proof}

    \paragraph{Comments on the substitution operator for distributions.}%
It follows from the definition that substitution preserves pointwise $\cdot$ and $+$.
Notice that substitution is not only defined for distributions but also for functions of the form $[\bool]$ where $\bool$ is a predicate, in which case it holds that $[\bool][\progvar x/\arith] = [\bool[\progvar x/\arith]]$ where $\bool[\progvar x/\arith]$ results from $\bool$ by \emph{syntactically replacing}\footnote{Assuming a concrete syntactic form of $\bool$ (e.g., as a first-order formula) is given.} all free occurrences of $\progvar x$ in $\bool$ by $\arith$
If $\mu \in \subdists$, then $\mu[\progvar x / \arith]$ is a distribution assigning to state $\sigma$ the probability that $\mu$ assigns to the updated state $\sigma[\progvar x / \arith(\sigma)]$ (if $\arith(\sigma)$ is defined; otherwise $\mu[\progvar x / \arith]$ assigns $0$ to $\sigma$).
For example, for the Dirac distribution $\dirac{\progvar x\mapsto 1}$ we have $\dirac{\progvar x \mapsto 1}[x/x+1] = \dirac{\progvar x \mapsto 0}$.
Substitution is not mass preserving:
For instance, $\dirac{\progvar x \mapsto 1}[x/2 x] = 0$ because there is no $x \in \mathbb Z$ such that $2x = 1$.

Post distributions resulting from injective assignments can be expressed using the substitution operator for distributions (proof in \Cref{subsec:proofinjassignments}):

\begin{definition}[Partial Inverse]%
    \label{def:partialInverse}%
    Let $\arith \colon\Sigma \to \mathbb{Z}$ and $\progvar x \in \vars$.
    We say that $\arith$ is \emph{injective in $\progvar x$} if for all $\sigma \in \Sigma$, the function $\arith_{\sigma,\progvar x} \colon \mathbb Z \to \mathbb Z, c \mapsto \arith(\sigma[\progvar x / c])$ is injective.
    In this case, we define the \emph{partial inverse of $\arith$ in $\progvar x$} as
    \[
        \arith^{-1}_{\progvar x} \colon \Sigma \dashrightarrow \mathbb Z \,,\quad \sigma \mapsto \arith_{\sigma,\progvar x}^{-1}(\sigma(\progvar x))
        ~.
    \]
\end{definition}
For example, if $\arith = x - y$ then $\arith^{-1}_x = x + y$ and $\arith^{-1}_y = x - y$.
In practice, these expressions can be obtained by solving $x' = x-y$ and $y' = x-y$ for $x$ and $y$, respectively.
Notice that both $\arith^{-1}_x$ and $\arith^{-1}_y$ are total in this example.

\subsection{Proof of \Cref{lem:injassignments}}
\label[appendix]{subsec:proofinjassignments}

\injassignments*
\begin{proof}
    We first observe that the function $\lambda \sigma. \sigma[\progvar x / \arith(\sigma)]$ has partial inverse $\lambda \sigma. \sigma[\progvar x / \arith^{-1}_{\progvar x}(\sigma)]$.
    The claim can now be shown as follows:
    \begin{align*}
        & \den{\ass{\progvar x}{\arith}}(\mu) \\
        ~{}={}~ & \lambda \sigma. \sum_{\sigma' \colon \sigma'[\progvar x / \arith(\sigma')] = \sigma} \mu(\sigma') \tag{Definition of $\den{\ass{\progvar x}{\arith}}$} \\
        ~{}={}~ & \lambda \sigma.
        \begin{cases}
            \sum_{\sigma' \colon \sigma' = \sigma[\progvar x / \arith^{-1}_\progvar{x}(\sigma)]} \mu(\sigma')  & \text{if } \arith^{-1}_\progvar{x}(\sigma') \text{ is defined}\\
            0 & \text{else}
        \end{cases} \tag{Observation from above} \\
        ~{}={}~ & \lambda \sigma.
        \begin{cases}
            \mu(\sigma[\progvar x / \arith^{-1}_\progvar{x}(\sigma)])  & \text{if } \arith^{-1}_\progvar{x}(\sigma') \text{ is defined}\\
            0 & \text{else}
        \end{cases} \tag{Simplifying} \\
        ~{}={}~ & \mu[\progvar x / \arith^{-1}_\progvar{x}] \tag{Definition of substitution}
    \end{align*}
\end{proof}

\begin{figure}[t]
    \setlength{\jot}{2pt} 
    \begin{align*}
        & \progAnno{\mu_0} \tag*{\gray{arbitrary initial distribution}} \\
        & \IF{x > 0} \\
        & \qquad \progAnno{[x>0] \cdot \mu_0} \\
        & \qquad \ass{x}{x-y} \tag*{\gray{$(x-y)_x^{-1} = x+y$}} \\
        & \qquad \progAnno{[x+y>0] \cdot \mu_0[x/x+y]} \\
        & \ELSE \\
        & \qquad \progAnno{[x\leq0] \cdot \mu_0} \\
        & \qquad \ass{y}{2y}  \tag*{\gray{$(2y)_y^{-1} = \tfrac y 2$ (\emph{partial} division function $\mathbb Z \dashrightarrow \mathbb Z$)}} \\
        & \qquad \progAnno{[x \leq 0] \cdot \mu_0[y / \tfrac y 2]} \\
        & \prbracecl \SEMICLN \progAnno{[x+y>0] \cdot \mu_0[x/x+y] \,+\, [x \leq 0] \cdot \mu_0[y / \tfrac y 2]} \\
        & \ass{x}{x+1} \tag*{\gray{$(x+1)_x^{-1} = x - 1$}} \\
        & \progAnno{[x{-}1+y>0] \cdot \mu_0[x/x+y][x/x{-}1] + [x{-}1 \leq 0] \cdot \mu_0[y / \tfrac y 2][x / x{-}1]} \tag*{\gray{This is $\den{\prog}(\mu_0)$.}} \\
        & \progAnno{[x+y>1] \cdot \mu_0[x/x-1+y] \,+\, [x \leq 1] \cdot \mu_0[y / \tfrac y 2][x / x -1]} \tag*{\gray{(simplification)}}
    \end{align*}
    \caption{%
        Characterizing the denotational semantics $\den{\prog}$ of a program with injective assignments.
    }
    \label{fig:tobiSubstitutionEx}
\end{figure}

\begin{example}[Programs with Injective Assignments]%
    \label{ex:tobiSubstitutionEx}%
    Consider the program $\prog$ in \Cref{fig:tobiSubstitutionEx} whose assigments are all injective in the sense of \Cref{lem:injassignments}.
    The bottommost annotation helps determining the output distribution w.r.t.\ an arbitrary initial distribution $\mu_0$ (i.e., the denotational semantics) as shown by means of the following examples:
    \begin{center}
        \setlength{\tabcolsep}{4pt}
        \begin{tabular}{l | l l l l}
            \textit{Final} & \multicolumn{4}{c}{\textit{Probability of $\sigma$ in final distribution $\den{\prog}(\mu_0)$ in terms of $\mu_0$}} \\
            \textit{state $\sigma$} & $[x+y>1]$ & $\mu_0[x/x-1+y]$ & $[x \leq 1]$ & $\mu_0[y / \tfrac y 2][x / x -1]$ \\
            \midrule
            $\begin{matrix*}[l] x\mapsto 1 \\[-2pt] y\mapsto 2 \end{matrix*}$ & $[1+2>1] = 1$ & $\mu_0 \Big(\begin{matrix*}[l] x\mapsto 1 {-} 1 {+} 2 = 2 \\[-2pt] y\mapsto 2\end{matrix*}\Big)$ & $[1 \leq 1] = 1$ & $\mu_0\Big(\begin{matrix*}[l]x\mapsto 1{-}1=0 \\[-2pt] y\mapsto \tfrac 2 2 = 1\end{matrix*}\Big)$ \\[10pt]
            $\begin{matrix*}[l] x\mapsto 0 \\[-2pt] y\mapsto 1 \end{matrix*}$ & $0$ & $\gray{\mu_0 \Big(\begin{matrix*}[l] x\mapsto 0 {-} 1 {+} 1 = 0 \\[-2pt] y\mapsto 1\end{matrix*}\Big)}$ & $\gray{1}$ & $\begin{matrix}0 \\ \gray{\text{(due to partiality)}}\end{matrix}$
        \end{tabular}
    \end{center}
    In words, the probability that $\prog$ terminates in state $x\mapsto 1, y \mapsto 2$ equals the initial probability of
    $x\mapsto 2, y \mapsto 2$ plus the initial probability of $x\mapsto 0, y \mapsto 1$.
    On the other hand, the probability of terminating in $x\mapsto 0, y \mapsto 1$ is $0$ for all initial distributions.
\end{example}

\subsection{Proof of \Cref{thm:unmodvars}}
\label{subsec:unmodvars}
\unmodvars*
\begin{proof}
We only consider the assignment case $\prog = (\ass{\progvar{y}}{\arith})$ because this is the only program construct where variables change their values. We have that $\progvar{y} \ne \progvar{x}$ since $\progvar{x}$ is an unmodified variable:
\begin{align*}
    \prosim{\mu_0}{\progvar{x}}{=}{z} & = \sum_{\sigma\colon \sigma(\progvar{x}) = z} \mu_0(\sigma)\tag{Definition of $\prosim{\mu_0}{\dots}{}{}$}\\
    & = \sum_{\sigma\colon \sigma(\progvar{x}) = z} \sum_{\sigma = \sigma'[\progvar{y}/\arith(\sigma')]} \mu_0(\sigma') \tag{$\progvar{y} \ne \progvar{x}$}\\
    & = \sum_{\sigma\colon \sigma(\progvar{x}) = z} \den{\prog}(\mu_0)(\sigma)\tag{Definition of Den. Sem.}\\
    & = \prosim{\mu}{\progvar{x}}{=}{z}~. \tag{Definition of $\prosim{\mu}{\dots}{}{}$}
\end{align*}
\end{proof}

    \section{Appendix to \Cref{sec:fdr}}
\label{app:practical}

\begin{figure}[t]
    \begin{center}
        \begin{tikzpicture}[node distance={30mm}, thick, main/.style = {draw, circle, minimum size=0.8cm},->,x=1cm,y=1cm]
    
    \node(-1) at (0,1)  {};
    \node[main](1) at (0,0)  {$1,0$};
    \node[main](2) at (-3,-1) {$2,0$};
    \node[main](3) at (-5,-2) {$4,0$};
    \node[main](4) at (-6,-3){$8,0$};

    \node[main](5) at (3,-1) {$2,1$};
    \node[main](6) at(-2,-2) {$4,1$};
    \node[main](7) at(-4,-3){$8,1$};
    
    \node[main](8) at(2,-2){$4,2$};
    \node[main](9) at(-3,-3){$8,2$};
    
    \node[main](10) at(5,-2){$4,3$};
    
    \node[main](12) at (-1,-3){$8,3$};
    \node[main](13) at(1,-3){$8,4$};
    \node[main](14) at(3,-3){$8,5$};
    
    \draw(-1)--(1);
    \draw(1)--(2);
    \draw(1)--(5);
    \draw(2)--(3);
    \draw(2)--(6);
    \draw(5)--(8);
    \draw(5)--(10);
    \draw(3)--(4);
    \draw(3)--(7);
    \draw(6)--(9);
    \draw(6)--(12);
    \draw(8)--(13);
    \draw(8)--(14);
    
    \draw(10) to [bend right = 5] (2);
    \draw(10) to [bend right = 35] (5);
    \path  (4)   edge[loop below] node[above]  {} (3);
    \path  (7)   edge[loop below] node[above]  {} (3);
    \path  (9)   edge[loop below] node[above]  {} (3);
    \path  (12)   edge[loop below] node[above]  {} (3);
    \path  (13)   edge[loop below] node[above]  {} (3);
    \path  (14)   edge[loop below] node[above]  {} (3);
\end{tikzpicture}
        \caption{The $\numb{n}=6$ fragment of the big-step operational Markov chain modeling the FDR. A state label $(\numb{v},\numb{c})$ is to be interpreted as the program state $\sigma$ with $\sigma(\progvar{v}) = \numb{v}, \sigma(\progvar{c}) = \numb{c}$, and $\sigma(\progvar{n}) = \numb{n}$. The probabilities of the transitions are $0.5$ whenever there are two outgoing transitions, and otherwise $1$.}\label{fig:MarkovChainFDR6}
    \end{center}
\end{figure}

\begin{figure}[t]
    \begin{center}
        \begin{tikzpicture}[xscale=0.85,yscale=0.7, main/.style = {draw, shape = circle, fill = white, minimum size = 0.2cm, inner sep=0pt},->,x=0.3cm,y=0.5cm]
    \node(-1) at (0,1)  {};
    \node[main,accepting](161) at (-20-2,-12)  {};
    \node[main,accepting](162) at (-19-2,-12)  {};
    \node[main,accepting](163) at (-18-2,-12)  {};
    \node[main,accepting](164) at (-17-2,-12)  {};
    \node[main,accepting](165) at (-16-2,-12)  {};
    \node[main,accepting](166) at (-15-2,-12)  {};
    \node[main,accepting](167) at (-14-2,-12)  {};
    \node[main,accepting](168) at (-13-2,-12)  {};
    \node[main,accepting](169) at (-12-2,-12)  {};
    
    \node[main,accepting](109) at (18+2,-9)  {};
    \node[main,accepting](108) at (17+2,-9)  {};
    \node[main,accepting](107) at (16+2,-9)  {};
    \node[main,accepting](106) at (15+2,-9)  {};
    \node[main,accepting](105) at (14+2,-9)  {};
    \node[main,accepting](104) at (13+2,-9)  {};
    \node[main,accepting](103) at (12+2,-9)  {};
    \node[main,accepting](102) at (11+2,-9)  {};
    \node[main,accepting](101) at (10+2,-9)  {};
    
    \node[main,accepting](149) at (4,-11)  {};
    \node[main,accepting](148) at (3,-11)  {};
    \node[main,accepting](147) at (2,-11)  {};
    \node[main,accepting](146) at (1,-11)  {};
    \node[main,accepting](145) at (0,-11)  {};
    \node[main,accepting](144) at (-1,-11)  {};
    \node[main,accepting](143) at (-2,-11)  {};
    \node[main,accepting](142) at (-3,-11)  {};
    \node[main,accepting](141) at (-4,-11)  {};
    
    \node[main](81) at (-15,-7)  {};
    \node[main](82) at (-14,-7)  {};
    \node[main](83) at (-13,-7)  {};
    \node[main](84) at (-12,-7)  {};
    \node[main](85) at (-11,-7)  {};
    \node[main](86) at (-10,-7)  {};
    \node[main](87) at (-9,-7)  {};
    \node[main](88) at (-8,-7)  {};

    \node[main](71) at (-4,-6)  {};
    \node[main](72) at (-3,-6)  {};
    \node[main](73) at (-2,-6)  {};
    \node[main](74) at (-1,-6)  {};
    \node[main](75) at (0,-6)  {};
    \node[main](76) at (1,-6)  {};
    \node[main](77) at (2,-6)  {};
    
    \node[main](55) at (12,-4)  {};
    \node[main](54) at (11,-4)  {};
    \node[main](53) at (10,-4)  {};
    \node[main](52) at (9,-4)  {};
    \node[main](51) at (8,-4)  {};
    
    \node[main](44) at (-4,-3)  {};
    \node[main](43) at (-5,-3)  {};
    \node[main](42) at (-6,-3)  {};
    \node[main](41) at (-7,-3)  {};
    
    \node[main](22) at (-2,-1)  {};
    \node[main](21) at (-3,-1)  {};
    
    \node[main](11) at (0,0)  {};
    
    \node at (-25, 0) {$\numb v {=} 1$};
    \node at (-25, -1) {$\numb v {=} 2$};
    \node at (-25, -3) {$\numb v {=} 4$};
    \node at (-25, -7) {$\numb v {=} 8$};
    \node at (-25, -12) {$\numb v {=} 16$};
    \node at (-25, -6) {$\numb v {=} 7$};
    \node at (-25, -11) {$\numb v {=} 14$};
    \node at (-25, -4) {$\numb v {=} 5$};
    \node at (-25, -9) {$\numb v {=} 10$};
    
    \draw(-1)--(11);
    
    \draw(11)--(21);
    \draw(11)--(22);
    
    \draw(21)--(41);
    \draw(21)--(42);
    \draw(22)--(43);
    \draw(22)--(44);
    
    \draw(41)--(81);
    \draw(41)--(82);
    \draw(42)--(83);
    \draw(42)--(84);
    \draw(43)--(85);
    \draw(43)--(86);
    \draw(44)--(87);
    \draw(44)--(88);
    
    \draw(81)--(161);
    \draw(81)--(162);
    \draw(82)--(163);
    \draw(82)--(164);
    \draw(83)--(165);
    \draw(83)--(166);
    \draw(84)--(167);
    \draw(84)--(168);
    \draw(85)--(169);
    \draw(85)to[bend right=70](71);
    \draw(86)to[bend right=70](72);
    \draw(86)to[bend right=70](73);
    \draw(87)to[bend right=70](74);
    \draw(87)to[bend right=70](75);
    \draw(88)to[bend right=70](76);
    \draw(88)to[bend right=70](77);
    
    \draw(71)--(141);
    \draw(71)--(142);
    \draw(72)--(143);
    \draw(72)--(144);
    \draw(73)--(145);
    \draw(73)--(146);
    \draw(74)--(147);
    \draw(74)--(148);
    \draw(75)--(149);
    \draw(75)to[bend right=70](51);
    \draw(76)to[bend right=70](52);
    \draw(76)to[bend right=70](53);
    \draw(77)to[bend right=70](54);
    \draw(77)to[bend right=70](55);
    
    \draw(51)--(101);
    \draw(51)--(102);
    \draw(52)--(103);
    \draw(52)--(104);
    \draw(53)--(105);
    \draw(53)--(106);
    \draw(54)--(107);
    \draw(54)--(108);
    \draw(55)--(109);
    \draw(55)--(11);
    
    
    
    
    \path  (161)   edge[loop below] node[above]  {} (3);
    \path  (162)   edge[loop below] node[above]  {} (3);
    \path  (163)   edge[loop below] node[above]  {} (3);
    \path  (164)   edge[loop below] node[above]  {} (3);
    \path  (165)   edge[loop below] node[above]  {} (3);
    \path  (166)   edge[loop below] node[above]  {} (3);
    \path  (167)   edge[loop below] node[above]  {} (3);
    \path  (168)   edge[loop below] node[above]  {} (3);
    \path  (169)   edge[loop below] node[above]  {} (3);
    
    \path  (141)   edge[loop below] node[above]  {} (3);
    \path  (142)   edge[loop below] node[above]  {} (3);
    \path  (143)   edge[loop below] node[above]  {} (3);
    \path  (144)   edge[loop below] node[above]  {} (3);
    \path  (145)   edge[loop below] node[above]  {} (3);
    \path  (146)   edge[loop below] node[above]  {} (3);
    \path  (147)   edge[loop below] node[above]  {} (3);
    \path  (148)   edge[loop below] node[above]  {} (3);
    \path  (149)   edge[loop below] node[above]  {} (3);
    
    \path  (101)   edge[loop below] node[above]  {} (3);
    \path  (102)   edge[loop below] node[above]  {} (3);
    \path  (103)   edge[loop below] node[above]  {} (3);
    \path  (104)   edge[loop below] node[above]  {} (3);
    \path  (105)   edge[loop below] node[above]  {} (3);
    \path  (106)   edge[loop below] node[above]  {} (3);
    \path  (107)   edge[loop below] node[above]  {} (3);
    \path  (108)   edge[loop below] node[above]  {} (3);
    \path  (109)   edge[loop below] node[above]  {} (3);
\end{tikzpicture}
        \caption{%
            The FDR's Markov chain for $\numb{n}=9$ (transition probabilities omitted).
        }
        \label{fig:MarkovChainFDR9}
    \end{center}
\end{figure}

\subsection{Proof of \Cref{lem:fdrinv} (FDR Invariant)}
\label[appendix]{sec:fdrinv}
\fdrinv*
\begin{proof}
	The main part and the outline of the proof can be found in the main text in \Cref{sec:fdr}. As a first step, we apply the definition of the denotational semantics for $\prfont{if}$-statements:
    \[
        \den{\ifelseskip{\bool}{\loopbody}}(\mu) ~=~ [\neg \bool] \cdot \mu + \den{\loopbody}([\bool] \cdot \mu) ~.
    \]
    Then, we derive the sub-expression $\den{\loopbody}([\bool] \cdot \mu)$ in \Cref{nextfdr} by repeatedly applying \Cref{lem:injassignments}.
    It remains to argue that $\mu' \coloneqq [\neg \bool] \cdot \mu + \den{\loopbody}([\bool] \cdot \mu)$ is indeed contained in $\Ifdr$ again, i.e. that it satisfies conditions \eqref{fdrA} -- \eqref{fdrD}. We have that
\begin{align*}
    & [\neg \bool] \cdot \mu + \den{\loopbody}([\bool] \cdot \mu)\\
    & = [\progvar{v} \geq \progvar{n}] \mu \tag{$\mu_1$}\\
    & + \frac{1}{2} [\progvar{v} < \progvar{n}, \progvar{c} \geq 0] (\mu[\progvar{v}/\frac{\progvar{v}+\progvar{n}}{2}, \progvar{c}/\frac{\progvar{c}+\progvar{n}}{2}] + \mu[\progvar{v}/\frac{\progvar{v}+\progvar{n}}{2}, \progvar{c}/\frac{\progvar{c}+\progvar{n}-1}{2}]) \tag{$\mu_2$}\\
    & + \frac{1}{2} [\progvar{v} < 2\progvar{n}, \progvar{c} < \progvar{n}] (\mu[\progvar{v}/\frac{\progvar{v}}{2}, \progvar{c}/\frac{\progvar{c}}{2}] + \mu[\progvar{v}/\frac{\progvar{v}}{2}, \progvar{c}/\frac{\progvar{c}-1}{2}]). \tag{$\mu_3$}
\end{align*}
This expression for $[\neg \bool] \cdot \mu + \den{\loopbody}([\bool] \cdot \mu)$ is a sum of three sub-distributions $\mu_1$, $\mu_2$ and $\mu_3$. The three sub-distributions are not yet \enquote{disjoint}, meaning that there are certain program states $\sigma$ and choices for $\mu \in \Ifdr$, such that multiple of the sub-distributions $\mu_1$, $\mu_2$ and $\mu_3$ assign a positive probability to $\sigma$. This is caused by the fact that a program state can fulfill the Iverson brackets of either $\mu_2$ and $\mu_3$, or of $\mu_1$ and $\mu_3$ simultaneously. In the following, we find an expression for $[\neg \bool] \cdot \mu + \den{\loopbody}([\bool] \cdot \mu)$ that \emph{is} a sum of disjoint sub-distributions which then simplifies our argumentation that $[\neg \bool] \cdot \mu + \den{\loopbody}([\bool] \cdot \mu)$ is conatined in $\Ifdr$. First, we rewrite our expression for $[\neg \bool] \cdot \mu + \den{\loopbody}([\bool] \cdot \mu)$ to a sum of all with respect to interference (of the Iverson brackets) different parts. This results in a sum of five sub-distributions, $\mu_1$, $\mu_2$ and $\mu_3$ stand-alone (abbreviated by s.a.), $\mu_2$ and $\mu_3$ interfered, and $\mu_1$ and $\mu_3$ interfered:
\begin{align*}
    & [\neg \bool] \cdot \mu + \den{\loopbody}([\bool] \cdot \mu)\\
    & = [\progvar{v} \geq \progvar{n}, (\progvar{v} \geq 2\progvar{n} \lor \progvar{c} \geq \progvar{n})] \mu \tag{$\mu_1$ s.a.}\\
    & + [\progvar{v} \geq \progvar{n}, \progvar{v} < 2\progvar{n}, \progvar{c} < \progvar{n}] (\mu + \frac{1}{2}(\mu[\progvar{v}/\frac{\progvar{v}}{2}, \progvar{c}/\frac{\progvar{c}}{2}] + \mu[\progvar{v}/\frac{\progvar{v}}{2}, \progvar{c}/\frac{\progvar{c}-1}{2}])) \tag{$\mu_1 ~\&~ \mu_3$}\\
    & + [\progvar{v} < 2\progvar{n}, \progvar{c} < \progvar{n}, \progvar{v} < \progvar{n}, (\progvar{v} \geq \progvar{n} \lor \progvar{c} < 0)] \frac{1}{2} (\mu[\progvar{v}/\frac{\progvar{v}}{2}, \progvar{c}/\frac{\progvar{c}}{2}] + \mu[\progvar{v}/\frac{\progvar{v}}{2}, \progvar{c}/\frac{\progvar{c}-1}{2}]) \tag{$\mu_3$ s.a.}\\
    & + [\progvar{v} < 2\progvar{n}, \progvar{c} < \progvar{n}, \progvar{v} < \progvar{n}, \progvar{c} \geq 0] \frac{1}{2} (\mu[\progvar{v}/\frac{\progvar{v}}{2}, \progvar{c}/\frac{\progvar{c}}{2}] + \mu[\progvar{v}/\frac{\progvar{v}}{2}, \progvar{c}/\frac{\progvar{c}-1}{2}]\\
    & \phantom{+ [\progvar{v} < 2\progvar{n}, \progvar{c} < \progvar{n}, \progvar{v} < \progvar{n}, \progvar{c} \geq 0] \frac{1}{2} (} + \mu[\progvar{v}/\frac{\progvar{v}+\progvar{n}}{2}, \progvar{c}/\frac{\progvar{c}+\progvar{n}}{2}] + \mu[\progvar{v}/\frac{\progvar{v}+\progvar{n}}{2}, \progvar{c}/\frac{\progvar{c}+\progvar{n}-1}{2}]) \tag{$\mu_2 ~\&~ \mu_3$}\\
    & + [\progvar{v} < \progvar{n}, \progvar{c} \geq 0, (\progvar{v} \geq 2\progvar{n} \lor \progvar{c} \geq \progvar{n})] \frac{1}{2} (\mu[\progvar{v}/\frac{\progvar{v}+\progvar{n}}{2}, \progvar{c}/\frac{\progvar{c}+\progvar{n}}{2}] + \mu[\progvar{v}/\frac{\progvar{v}+\progvar{n}}{2}, \progvar{c}/\frac{\progvar{c}+\progvar{n}-1}{2}]).\tag{$\mu_2$ s.a.}
\end{align*}

This expression for $[\neg \bool] \cdot \mu + \den{\loopbody}([\bool] \cdot \mu)$ is sizewise even larger than before. However, recall that $\mu$ is not an arbitrary distribution but itself contained in $\Ifdr$. We argue that due to conditions $\eqref{fdrB}$ and $\eqref{fdrC}$ from the invariant, it holds that $(\mu_1~ s.a.) = (\mu_2~ s.a.) = (\mu_3~ s.a.) = 0$. Condition $\eqref{fdrB}$ declares that $\mu(\sigma) = 0$ for all program states $\sigma$ with $\sigma(\progvar{v}) \geq 2\cdot\sigma(\progvar{n})$, witnessing $(\mu_1~ s.a.) = 0$. Condition $\eqref{fdrC}$ declares that $\mu(\sigma) = 0$ for all program states $\sigma$ with $\sigma(\progvar{c}) < 0$, witnessing $(\mu_3~ s.a.) = 0$, or with $\sigma(\progvar{c}) \geq \sigma(\progvar{n})$, witnessing $(\mu_2~ s.a.) = 0$. This leaves us with the following final expression for $[\neg \bool] \cdot \mu + \den{\loopbody}([\bool] \cdot \mu)$:
\begin{align*}
    & [\neg \bool] \cdot \mu + \den{\loopbody}([\bool] \cdot \mu)\\
    & = [\progvar{v} \geq \progvar{n}, \progvar{v} < 2\progvar{n}, \progvar{c} < \progvar{n}] (\mu + \frac{1}{2}(\mu[\progvar{v}/\frac{\progvar{v}}{2}, \progvar{c}/\frac{\progvar{c}}{2}] + \mu[\progvar{v}/\frac{\progvar{v}}{2}, \progvar{c}/\frac{\progvar{c}-1}{2}])) \tag{$\mu_4$}\\
    & + [\progvar{v} < \progvar{n}, \progvar{c} < \progvar{n}, \progvar{c} \geq 0]\frac{1}{2} (\mu[\progvar{v}/\frac{\progvar{v}}{2}, \progvar{c}/\frac{\progvar{c}}{2}] + \mu[\progvar{v}/\frac{\progvar{v}}{2}, \progvar{c}/\frac{\progvar{c}-1}{2}]\\
    & \phantom{+ [\progvar{v} < 2\progvar{n}, \progvar{c} < \progvar{n}, \progvar{v} < \progvar{n}, \progvar{c} \geq 0]} + \mu[\progvar{v}/\frac{\progvar{v}+\progvar{n}}{2}, \progvar{c}/\frac{\progvar{c}+\progvar{n}}{2}] + \mu[\progvar{v}/\frac{\progvar{v}+\progvar{n}}{2}, \progvar{c}/\frac{\progvar{c}+\progvar{n}-1}{2}]). \tag{$\mu_5$}
\end{align*}

We simplified $[\neg \bool] \cdot \mu + \den{\loopbody}([\bool] \cdot \mu)$ to a sum of two sub-distributions $\mu_4$ and $\mu_5$ which are disjoint: Sub-distribution $\mu_4$ only assigns positive probabilities to program states $\sigma$ with $\sigma(\progvar{v}) \geq \sigma(\progvar{n})$, while sub-distribution $\mu_5$ only assigns positive probabilities to program states $\sigma$ with $\sigma(\progvar{v}) < \sigma(\progvar{n})$. Now, we argue that $[\neg \bool] \cdot \mu + \den{\loopbody}([\bool] \cdot \mu) \in \Ifdr$. To do so, we prove separately that the sum $\mu_4 + \mu_5$ satisfies each of the conditions $\eqref{fdrA}, \eqref{fdrB}, \eqref{fdrC}$ and $\eqref{fdrD}$ of $\Ifdr$.
\begin{itemize}
	\item $\eqref{fdrA}$: Since $\mu \in \Ifdr$, we have that $\mu$ satisfies condition $\eqref{fdrA}$ itself, meaning that there exists an $\numb{n} \in \nats_{> 0}$ such that $\prosim{\mu}{\progvar{n}}{=}{\numb{n}} = 1$ (and $\prosim{\mu}{\progvar{n}}{=}{\numb{n'}} = 0$ for all $\numb{n'} \ne \numb{n}$). We have that $\textvar{n}$ is an unmodified variable in the loopbody of the FDR, and thus \Cref{thm:unmodvars} is applicable and yields that $\prosim{\mu_4+\mu_5}{\progvar{n}}{=}{\numb{n}} \leq 1$ (and $\prosim{\mu_4+\mu_5}{\progvar{n}}{=}{\numb{n'}} = 0$ for all $\numb{n'} \ne \numb{n}$). Since however $\mu$, as a distribution from the invariant, has a mass of $1$, we get that
    \begin{align*}
    	1 & = \weight{\mu}\\
        & = \weight{[\neg \bool] \cdot \mu + \den{\loopbody}([\bool] \cdot \mu)}\tag{$\loopbody$ loop-free}\\
        & = \weight{\mu_4 + \mu_5}
    \end{align*}
    which lets us conclude that $\prosim{\mu_4+\mu_5}{\progvar{n}}{=}{\numb{n}} = 1$.

    \item $\eqref{fdrB}$: Condition $\eqref{fdrB}$ as stated in $\Ifdr$ is equivalent to $\prosim{\mu}{\progvar{v}}{\leq}{0} = 0 \land \prosim{\mu}{\progvar{v}}{\geq}{2\numb{n}-1} = 0$ and we prove each of the two conjuncts seperately. In the following, we often use three different ways to argue that a sub-distribution assigns $0$ to a program state. A): Iverson brackets, B): since $\mu$ as a distribution from $\Ifdr$ satisfies conditions $\eqref{fdrA}, \eqref{fdrB}, \eqref{fdrC}$ and $\eqref{fdrD}$ itself and these conditions declare various program states as unreachable, or C): we insert values $z \notin \mathbb{Z}$ into the substitution operator for distributions, which yields $0$ by definition (recall \Cref{defn:substForDists}). We have that
    \begin{align*}
        \prosim{\mu_4}{\progvar{v}}{\leq}{0} & = \sum_{\sigma\colon \sigma(\progvar{v}) \leq 0} \mu_4(\sigma) \tag{Definition of $\textnormal{Pr}_{\mu_4}(\dots)$}\\
        & = 0 \tag{Iverson Bracket of $\mu_4$}
    \end{align*}
    and
    \begin{align*}
        \prosim{\mu_5}{\progvar{v}}{\leq}{0} & = \sum_{\sigma\colon \sigma(\progvar{v}) \leq 0} \mu_5(\sigma) \tag{Definition of $\textnormal{Pr}_{\mu_5}(\dots)$}\\
        & = \sum_{\sigma\colon \sigma(\progvar{v}) \leq 0} [\progvar{v} < \progvar{n}, \progvar{c} < \progvar{n}, \progvar{c} \geq 0] \frac{1}{2} (\mu[\progvar{v}/\frac{\progvar{v}}{2}, \progvar{c}/\frac{\progvar{c}}{2}] + \mu[\progvar{v}/\frac{\progvar{v}}{2}, \progvar{c}/\frac{\progvar{c}-1}{2}]\\
        & \phantom{MMMMMMM} + \mu[\progvar{v}/\frac{\progvar{v}+\progvar{n}}{2}, \progvar{c}/\frac{\progvar{c}+\progvar{n}}{2}]\\
        & \phantom{MMMMMMM} + \mu[\progvar{v}/\frac{\progvar{v}+\progvar{n}}{2}, \progvar{c}/\frac{\progvar{c}+\progvar{n}-1}{2}])(\sigma) \tag{Definition of $\mu_5$}\\
        & = \sum_{\sigma\colon \sigma(\progvar{v}) \leq 0, \sigma(\progvar{c}) \geq 0, \sigma(\progvar{c}) < \sigma(\progvar{n})} \frac{1}{2} (\mu[\progvar{v}/\frac{\progvar{v}}{2}, \progvar{c}/\frac{\progvar{c}}{2}] + \mu[\progvar{v}/\frac{\progvar{v}}{2}, \progvar{c}/\frac{\progvar{c}-1}{2}]\\
        & \phantom{MMMMMMM} + \mu[\progvar{v}/\frac{\progvar{v}+\progvar{n}}{2}, \progvar{c}/\frac{\progvar{c}+\progvar{n}}{2}]\\
        & \phantom{MMMMMMM} + \mu[\progvar{v}/\frac{\progvar{v}+\progvar{n}}{2}, \progvar{c}/\frac{\progvar{c}+\progvar{n}-1}{2}])(\sigma) \tag{Iverson Bracket}\\
        & = \sum_{\sigma\colon \sigma(\progvar{v}) \leq 0, \sigma(\progvar{c}) \geq 0, \sigma(\progvar{c}) < \sigma(\progvar{n})} \frac{1}{2} (\mu[\progvar{v}/\frac{\progvar{v}+\progvar{n}}{2}, \progvar{c}/\frac{\progvar{c}+\progvar{n}}{2}] + \mu[\progvar{v}/\frac{\progvar{v}+\progvar{n}}{2}, \progvar{c}/\frac{\progvar{c}+\progvar{n}-1}{2}])(\sigma) \tag{$\sigma(\progvar{v}) < 0$ iff $\frac{\sigma(\progvar{v})}{2} < 0$ and $\mu$ satisfies \eqref{fdrB}}\\
        & = \sum_{\sigma\colon \sigma(\progvar{v}) \leq 0, \sigma(\progvar{c}) \geq 0, \sigma(\progvar{c}) < \sigma(\progvar{n})} \frac{1}{2} (\mu[\progvar{v}/\frac{\progvar{v}+\progvar{n}}{2}, \progvar{c}/\frac{\progvar{c}+\progvar{n}-1}{2}])(\sigma) \tag{$\frac{\progvar{c}+\progvar{n}}{2} \geq \frac{\progvar{v}+\progvar{n}}{2}$ and $\mu$ satisfies \eqref{fdrC}}\\
        & = 0 \tag{either $\frac{\progvar{c}+\progvar{n}-1}{2} \geq \frac{\progvar{v}+\progvar{n}}{2}$ and $\mu$ satisfies \eqref{fdrC} or $\frac{\progvar{c}+\progvar{n}-1}{2} \notin \mathbb{Z}$ iff $\frac{\progvar{v}+\progvar{n}}{2} \in \mathbb{Z}$}
    \end{align*}
    as well as 
    \begin{align*}
        \prosim{\mu_4}{\progvar{v}}{\geq}{2\numb{n}-1} & = \sum_{\sigma\colon \sigma(\progvar{v}) \geq 2\numb{n}-1} \mu_4(\sigma) \tag{Definition of $\textnormal{Pr}_{\mu_4}(\dots)$}\\
        & = \sum_{\sigma\colon \sigma(\progvar{v}) = 2\numb{n}-1} \mu_4(\sigma) \tag{Iverson Bracket of $\mu_4$}\\
        & = \sum_{\sigma\colon \sigma(\progvar{v}) = 2\numb{n}-1} \mu(\sigma) + \frac{1}{2}(\mu[\progvar{v}/\frac{\progvar{v}}{2}, \progvar{c}/\frac{\progvar{c}}{2}] + \mu[\progvar{v}/\frac{\progvar{v}}{2}, \progvar{c}/\frac{\progvar{c}-1}{2}])(\sigma) \tag{Definition of $\mu_4$}\\
        & = \sum_{\sigma\colon \sigma(\progvar{v}) = 2\numb{n}-1} \frac{1}{2}(\mu[\progvar{v}/\frac{\progvar{v}}{2}, \progvar{c}/\frac{\progvar{c}}{2}] + \mu[\progvar{v}/\frac{\progvar{v}}{2}, \progvar{c}/\frac{\progvar{c}-1}{2}])(\sigma) \tag{$\mu$ satisfies \eqref{fdrB}}\\
        & = 0 \tag{$\frac{2\numb{n}-1}{2} \notin \mathbb{Z}$}
    \end{align*}
    and
    \begin{align*}
        \prosim{\mu_5}{\progvar{v}}{\geq}{2\numb{n}-1} & = \sum_{\sigma\colon \sigma(\progvar{v}) \leq 0} \mu_5(\sigma) \tag{Definition of $\textnormal{Pr}_{\mu_5}(\dots)$}\\
        & = 0. \tag{Iverson Bracket of $\mu_5$}
    \end{align*}

    \item $\eqref{fdrC}$: Condition $\eqref{fdrC}$ as stated in $\Ifdr$ is equivalent to $\prosim{\mu}{\progvar{c}}{<}{0} = 0 \land \prosim{\mu}{\progvar{c}}{\geq}{\numb{n}} = 0 \land \prosim{\mu}{\progvar{c}}{\geq}{\progvar{v}} = 0$ and we prove each of the three conjuncts seperately. We have that
    \begin{align*}
        \prosim{\mu_4}{\progvar{c}}{<}{0} & = \sum_{\sigma\colon \sigma(\progvar{c}) < 0} \mu_4(\sigma) \tag{Definition of $\textnormal{Pr}_{\mu_4}(\dots)$}\\
        & = \sum_{\sigma\colon \sigma(\progvar{c}) < 0} [\progvar{v} \geq \progvar{n}, \progvar{v} < 2\progvar{n}, \progvar{c} < \progvar{n}] (\mu + \frac{1}{2}(\mu[\progvar{v}/\frac{\progvar{v}}{2}, \progvar{c}/\frac{\progvar{c}}{2}]\\
        & \phantom{MMMMMMMMM} + \mu[\progvar{v}/\frac{\progvar{v}}{2}, \progvar{c}/\frac{\progvar{c}-1}{2}]))(\sigma) \tag{Definition of $\mu_4$}\\
        & = 0 \tag{$\frac{\sigma(\progvar{c})}{2} < 0$ and $\frac{\sigma(\progvar{c})-1}{2} < 0$ and $\mu$ satisfies \eqref{fdrC}}
    \end{align*}
    and
    \begin{align*}
        \prosim{\mu_5}{\progvar{c}}{<}{0} & = \sum_{\sigma\colon \sigma(\progvar{c}) < 0} \mu_5(\sigma) \tag{Definition of $\textnormal{Pr}_{\mu_5}(\dots)$}\\
        & = 0 \tag{Iverson Bracket of $\mu_5$}
    \end{align*}
    as well as
    \begin{align*}
        \prosim{\mu_4}{\progvar{c}}{\geq}{\numb{n}} & = \sum_{\sigma\colon \sigma(\progvar{c}) \geq \numb{n}} \mu_4(\sigma) \tag{Definition of $\textnormal{Pr}_{\mu_4}(\dots)$}\\
        & = 0 \tag{Iverson Bracket of $\mu_4$}
    \end{align*}
    and
    \begin{align*}
        \prosim{\mu_5}{\progvar{c}}{\geq}{\numb{n}} & = \sum_{\sigma\colon \sigma(\progvar{c}) \geq \numb{n}} \mu_5(\sigma) \tag{Definition of $\textnormal{Pr}_{\mu_5}(\dots)$}\\
        & = 0 \tag{Iverson Bracket of $\mu_5$}
    \end{align*}
    as well as
    \begin{align*}
        \prosim{\mu_4}{\progvar{c}}{\geq}{\progvar{v}} & = \sum_{\sigma\colon \sigma(\progvar{c}) \geq \sigma(\progvar{v})} \mu_4(\sigma) \tag{Definition of $\textnormal{Pr}_{\mu_4}(\dots)$}\\
        & = 0 \tag{Iverson Bracket of $\mu_4$}
    \end{align*}
    and
    \begin{align*}
        \prosim{\mu_5}{\progvar{c}}{\geq}{\progvar{v}} & = \sum_{\sigma\colon \sigma(\progvar{c}) \geq \sigma(\progvar{v})} \mu_5(\sigma) \tag{Definition of $\textnormal{Pr}_{\mu_5}(\dots)$}\\
        & = \sum_{\sigma\colon \sigma(\progvar{v}) \leq 0} [\progvar{v} < \progvar{n}, \progvar{c} < \progvar{n}, \progvar{c} \geq 0] \frac{1}{2} (\mu[\progvar{v}/\frac{\progvar{v}}{2}, \progvar{c}/\frac{\progvar{c}}{2}] + \mu[\progvar{v}/\frac{\progvar{v}}{2}, \progvar{c}/\frac{\progvar{c}-1}{2}]\\
        & \phantom{MMMMMMM} + \mu[\progvar{v}/\frac{\progvar{v}+\progvar{n}}{2}, \progvar{c}/\frac{\progvar{c}+\progvar{n}}{2}]\\
        & \phantom{MMMMMMM} + \mu[\progvar{v}/\frac{\progvar{v}+\progvar{n}}{2}, \progvar{c}/\frac{\progvar{c}+\progvar{n}-1}{2}])(\sigma) \tag{Definition of $\mu_5$}\\
        & = \sum_{\sigma\colon \sigma(\progvar{v}) \leq 0, \sigma(\progvar{c}) \geq 0, \sigma(\progvar{c}) < \sigma(\progvar{n})} \frac{1}{2} (\mu[\progvar{v}/\frac{\progvar{v}}{2}, \progvar{c}/\frac{\progvar{c}}{2}] + \mu[\progvar{v}/\frac{\progvar{v}}{2}, \progvar{c}/\frac{\progvar{c}-1}{2}]\\
        & \phantom{MMMMMMM} + \mu[\progvar{v}/\frac{\progvar{v}+\progvar{n}}{2}, \progvar{c}/\frac{\progvar{c}+\progvar{n}}{2}]\\
        & \phantom{MMMMMMM} + \mu[\progvar{v}/\frac{\progvar{v}+\progvar{n}}{2}, \progvar{c}/\frac{\progvar{c}+\progvar{n}-1}{2}])(\sigma) \tag{Iverson Bracket}\\
        & = \sum_{\sigma\colon \sigma(\progvar{v}) \leq 0, \sigma(\progvar{c}) \geq 0, \sigma(\progvar{c}) < \sigma(\progvar{n})} \frac{1}{2} (\mu[\progvar{v}/\frac{\progvar{v}}{2}, \progvar{c}/\frac{\progvar{c}-1}{2}] + \mu[\progvar{v}/\frac{\progvar{v}+\progvar{n}}{2}, \progvar{c}/\frac{\progvar{c}+\progvar{n}}{2}]\\
        & \phantom{MMMMMMM} + \mu[\progvar{v}/\frac{\progvar{v}+\progvar{n}}{2}, \progvar{c}/\frac{\progvar{c}+\progvar{n}-1}{2}])(\sigma) \tag{$\frac{\sigma(\progvar{c})}{2} \geq \frac{\sigma(\progvar{v})}{2}$ and $\mu$ satisfies \eqref{fdrC}}\\
        & = \sum_{\sigma\colon \sigma(\progvar{v}) \leq 0, \sigma(\progvar{c}) \geq 0, \sigma(\progvar{c}) < \sigma(\progvar{n})} \frac{1}{2} (\mu[\progvar{v}/\frac{\progvar{v}}{2}, \progvar{c}/\frac{\progvar{c}-1}{2}] + \mu[\progvar{v}/\frac{\progvar{v}+\progvar{n}}{2}, \progvar{c}/\frac{\progvar{c}+\progvar{n}-1}{2}])(\sigma) \tag{$\frac{\sigma(\progvar{c}) + \numb{n}}{2} \geq \frac{\sigma(\progvar{v}) + \numb{n}}{2}$ and $\mu$ satisfies \eqref{fdrC}}\\
        & = \sum_{\sigma\colon \sigma(\progvar{v}) \leq 0, \sigma(\progvar{c}) \geq 0, \sigma(\progvar{c}) < \sigma(\progvar{n})} \frac{1}{2} (\mu[\progvar{v}/\frac{\progvar{v}+\progvar{n}}{2}, \progvar{c}/\frac{\progvar{c}+\progvar{n}-1}{2}])(\sigma) \tag{either $\frac{\sigma(\progvar{c}) - 1}{2} \geq \frac{\sigma(\progvar{v})}{2}$ and $\mu$ satisfies \eqref{fdrC} or $\frac{\sigma(\progvar{c}) - 1}{2} \notin \mathbb{Z}$ iff $\frac{\sigma(\progvar{v})}{2} \in \mathbb{Z}$}\\
        & = 0. \tag{either $\frac{\sigma(\progvar{c}) + \numb{n} - 1}{2} \geq \frac{\sigma(\progvar{v}) + \numb{n}}{2}$ and $\mu$ satisfies \eqref{fdrC} or $\frac{\sigma(\progvar{c}) + \numb{n} - 1}{2} \notin \mathbb{Z}$ iff $\frac{\sigma(\progvar{v}) + \numb{n}}{2} \in \mathbb{Z}$}
    \end{align*}

    \item $\eqref{fdrD}$: Let $\sigma, \sigma' \in \Sigma$ be two program states. If $\sigma(\progvar{v}) = \sigma'(\progvar{v}), \sigma(\progvar{n}) = \numb{n} = \sigma'(\progvar{n})$ and $0 \leq \sigma(\progvar{c}), \sigma'(\progvar{c}) < \minim{\sigma(\progvar{v})}{\numb{n}}$, then we conclude for sub-distribution
	\begin{itemize}
	\item $\mu_4$: We have that $\frac{\sigma(\progvar{c})}{2} \in \mathbb{Z}$ iff $\frac{\sigma(\progvar{c})-1}{2} \notin \mathbb{Z}$. Furthermore,
	\[
		0 \leq \frac{\sigma(\progvar{c})}{2}, \frac{\sigma'(\progvar{c})}{2} < \minim{\frac{\sigma(\progvar{v})}{2}}{\numb{n}}
	\]
	as well as
	\[
		0 \leq \frac{\sigma(\progvar{c})-1}{2}, \frac{\sigma'(\progvar{c})-1}{2} < \minim{\frac{\sigma(\progvar{v})}{2}}{\numb{n}}.
	\]
	Thus, we insert into $\mu$ only pairs of states that satisfy themselves the left-hand side of the implication of condition $\eqref{fdrD}$. Since $\mu$ itself satisfies condition $\eqref{fdrD}$, we also have that $\mu_4$ satisfies $\eqref{fdrD}$.
	\item $\mu_5$: We have that $\frac{\sigma(\progvar{c})+\numb{n}}{2} \in \mathbb{Z}$ iff $\frac{\sigma(\progvar{c})+\numb{n}-1}{2} \notin \mathbb{Z}$. Furthermore,
	\[
		0 \leq \frac{\sigma(\progvar{c})+\numb{n}}{2}, \frac{\sigma'(\progvar{c})+\numb{n}}{2} < \minim{\frac{\sigma(\progvar{v})+\numb{n}}{2}}{\numb{n}}
	\]
	as well as
	\[
		0 \leq \frac{\sigma(\progvar{c})+\sigma(\progvar{n})-1}{2}, \frac{\sigma'(\progvar{c})+\sigma(\progvar{n})-1}{2} < \min\{\frac{\sigma(\progvar{v})+\sigma(\progvar{n})}{2},\sigma(\progvar{n})\}
	\]
	and
	\[
		0 \leq \frac{\sigma(\progvar{c})}{2}, \frac{\sigma'(\progvar{c})}{2} < \minim{\frac{\sigma(\progvar{v})}{2}}{\numb{n}}
	\]
	as well as
	\[
		0 \leq \frac{\sigma(\progvar{c})-1}{2}, \frac{\sigma'(\progvar{c})-1}{2} < \minim{\frac{\sigma(\progvar{v})}{2}}{\numb{n}}.
	\]
	Thus, we insert into $\mu$ only pairs of states that satisfy themselves the left-hand side of the implication of condition $\eqref{fdrD}$. Since $\mu$ itself satisfies condition $\eqref{fdrD}$, we also have that $\mu_5$ satisfies $\eqref{fdrD}$.
	\end{itemize}
\end{itemize}
\end{proof}

\subsection{Proof of \Cref{thm:fdrpartial} (FDR Correctness)}
\label{sec:fdrpartial}
\fdrtotal*
\begin{proof}
\Cref{lem:fdrinv} proved $\Ifdr$ to be an inductive, distributional invariant $\Ifdr$ for the FDR and its initial distributions $\Mfdr$. Thus, we can apply the \eqref{WhileP} Hoare rule which yields validity of the Hoare triple
\[
    \triple{\Ifdr}{\prog}{\closOf{[\neg \bool] \Ifdr}}.
\]
Since $\Mfdr \subseteq \Ifdr$, we have for all $\mu_0 \in \Mfdr$ that $\mu_0 \in \Ifdr$ and by definition of valid Hoare triples that $\mu \in \closOf{[\neg \bool] \Ifdr}$. See \Cref{fig:fdrreach} (right) for an illustrative graphic of how distribution from $\closOf{[\neg \bool] \Ifdr}$ are shaped. Condition $\eqref{fdrA}$ of $\Ifdr$ implies that there exists an $\numb{n'} \in \nats_{> 0}$ such that $\prosim{\mu}{\progvar{n}}{=}{\numb{n'}} = 1$. Towards a contradiction, assume $\numb{n'} \ne \numb{n}$. Then, since $\progvar{n} \in \unmod{\textnormal{FDR}}$, we obtain by \Cref{thm:unmodvars} that $\prosim{\mu}{\progvar{n}}{=}{\numb{n'}} = 1 \leq \prosim{\mu_0}{\progvar{n}}{=}{\numb{n'}} = 0$, a contradiction. So, $\numb{n'} = \numb{n}$ indeed holds.

Now, let $z,z' \in \mathbb{Z}$.
\begin{itemize}
    \item If $z < 0$, then due to condition $\eqref{fdrC}$ of $\Ifdr$, we have that $\forall \sigma \in \Sigma\colon \sigma(\progvar{c}) < 0 \Rightarrow \mu(\sigma) = 0$ and thus
    \begin{align*}
        \marg{\{\progvar{c}\}}{\mu}(\progvar{c} \mapsto z) = \sum_{\sigma\colon \sigma(\progvar{c}) = z} \mu(\sigma) \overset{\eqref{fdrC}}{=} 0~.
    \end{align*}
    \item If $z \geq \numb{n}$, then due to condition $\eqref{fdrC}$ of $\Ifdr$, we have that $\forall \sigma \in \Sigma\colon \sigma(\progvar{c}) \geq \numb{n} \Rightarrow \mu(\sigma) = 0$ and thus
    \begin{align*}
        \marg{\{\progvar{c}\}}{\mu}(\progvar{c} \mapsto z) = \sum_{\sigma\colon \sigma(\progvar{c}) = z} \mu(\sigma) \overset{\eqref{fdrC}}{=} 0~.
    \end{align*}
    \item If $0 \leq z,z' \leq \numb{n}-1$, then, using conditions $\eqref{fdrB}$ and $\eqref{fdrD}$ of the invariant, we get that
    \begin{align*}
        \marg{\{\progvar{c}\}}{\mu}(\progvar{c} \mapsto z) = \sum_{\sigma\colon \sigma(\progvar{c}) = z} \mu(\sigma) & = \sum_{i \in \mathbb{Z}} \sum_{\sigma\colon \sigma(\progvar{c}) = z \land \sigma(\progvar{v}) = i} \mu(\sigma)\\
        & = \sum_{i \in \mathbb{Z}, \numb{n} \leq i} \sum_{\sigma\colon \sigma(\progvar{c}) = z \land \sigma(\progvar{v}) = i} \mu(\sigma)\tag{$\mu \in \langle \neg \bool \rangle$}\\
        & = \sum_{i \in \mathbb{Z}, \numb{n} \leq i \leq 2\numb{n}-2} \sum_{\sigma\colon \sigma(\progvar{c}) = z \land \sigma(\progvar{v}) = i} \mu(\sigma) \tag{Condition $\eqref{fdrB}$}\\
        & = \sum_{i \in \mathbb{Z}, \numb{n} \leq i \leq 2\numb{n}-2} \sum_{\sigma\colon \sigma(\progvar{c}) = z' \land \sigma(\progvar{v}) = i} \mu(\sigma) \tag{Condition $\eqref{fdrD}$}\\
        & = \marg{\{\progvar{c}\}}{\mu}(\progvar{c} \mapsto z')~.
    \end{align*}

    We have that $\weight{\mu} \leq 1$ (since the invariant contains no state values with a mass > 1) and thus for all $0 \leq z \leq \numb{n}-1$, it holds that $\marg{\{\progvar{c}\}}{\mu}(\progvar{c} \mapsto z) \leq \frac{1}{\numb{n}}$, such that we finally obtain for $r = \marg{\{\progvar{c}\}}{\mu}(\progvar{c} \mapsto z) \cdot \numb{n} \in [0,1]$ that
    \[
        \marg{\{\progvar{c}\}}{\mu} = r \cdot \unif{\progvar{c}}{0}{\numb{n}-1}~.
    \]
\end{itemize} 

    Now, we additionally assume the FDR to terminate almost surely with respect to $\Ifdr$. This allows us to apply the \eqref{WhileT} Hoare rule which yields validity of the Hoare triple
    \[
        \triple{\Ifdr}{\prog}{\closOf{[\neg \bool] \Ifdr} \cap \Deltaspec{\geq \tau}\Sigma},
    \]
    where $\tau = \inf_{\mu \in \Ifdr} \weight{\mu}$. Thus, it holds for the post that $\mu \in \Deltaspec{\geq \tau}\Sigma$. We have that $\tau = 1$ since our invariant only contains distributions with a mass of $1$ by construction, and therefore, $\weight{\mu} = 1$. Towards a contradiction, assume $r < 1$. Then,
    \begin{align*}
        1 = \weight{\mu} = \weight{\marg{\{\progvar{c}\}}{\mu}} = \weight{r \cdot \unif{\progvar{c}}{0}{\numb{n}-1}} = r \cdot 1 < 1
    \end{align*}
    which is a contradiction. Thus, we obtain that $r = 1$ and by that our final result:
    \[
        \marg{\{\progvar{c}\}}{\mu} = \unif{\progvar{c}}{0}{\numb{n}-1}~.
    \]
\end{proof}

\subsection{AST of the FDR}
\label[appendix]{app:fdrast}

We can prove our FDR encoding to be AST on $\Ifdr$ by using the sound and relative complete proof rule 3.2 by Majumdar et al. \cite{DBLP:journals/pacmpl/MajumdarS25}. In order to apply the rule, we consider the \emph{probabilistic control flow graph} underlying our FDR program and a \emph{state}-based invariant that we derive from the \emph{distributional} invariant $\Ifdr$ presented in this paper. At the loop location, the state-based invariant contains precisely the states from the support of $\Ifdr$:
\[
    \textnormal{Inv} \coloneqq \underset{\mu \in \Ifdr}{\bigcup} \supp(\mu)~.
\]
For these states, we employ the supermartingale function
    \begin{align*}
        v(\sigma) \coloneqq \begin{cases}
            0 & \textit{if } \sigma \not \models [v < n]\\
            \sigma(n) & \text{otherwise}
        \end{cases}
    \end{align*}
    which is a suitable choice, since for every $\numb{r} \in \mathbb{R}$, only finitely many states $\sigma \in \textnormal{Inv}$ satisfy $\sigma(n) \leq \numb{r}$.
And we employ the variant function
    \begin{align}
        u(\sigma) \coloneqq \begin{cases}
            0 & \textit{if } \sigma \not \models [v < n]\\
            \left\lceil \log_2 (\frac{\sigma(n)}{\sigma(v)-\sigma(c)}) \right\rceil & \text{otherwise}
        \end{cases}
    \end{align}
    which, intuitively, is suitable since it assigns to each program state an overapproximation of its minimal distance to termination (illustrated in \Cref{fig:astfdr}) and thus decreases with probability at least $0.5$ in each loop iteration.

    \begin{figure}[t]
    \begin{center}
        \begin{tikzpicture}[node distance={30mm}, thick, main/.style = {draw, circle, minimum size=0.8cm},->,x=1cm,y=1cm]
    
    \node(-1) at (0,1)  {};
    \node[main, label=above right:3](1) at (0,0)  {$1,0$};
    \node[main, label=above right:2](2) at (-3,-1) {$2,0$};
    \node[main, label=above right:1](3) at (-5,-2) {$4,0$};
    \node[main, label=above left:0](4) at (-6,-3){$8,0$};

    \node[main, label=above right:2](5) at (3,-1) {$2,1$};
    \node[main, label=above right:1](6) at(-2,-2) {$4,1$};
    \node[main, label=above:0](7) at(-4,-3){$8,1$};
    
    \node[main, label=left:1](8) at(2,-2){$4,2$};
    \node[main, label=above:0](9) at(-3,-3){$8,2$};
    
    \node[main, label=above right:3](10) at(5,-2){$4,3$};
    
    \node[main, label=above right:0](12) at (-1,-3){$8,3$};
    \node[main, label=above left:0](13) at(1,-3){$8,4$};
    \node[main, label=above right:0](14) at(3,-3){$8,5$};
    
    \draw(-1)--(1);
    \draw(1)--(2);
    \draw(1)--(5);
    \draw(2)--(3);
    \draw(2)--(6);
    \draw(5)--(8);
    \draw(5)--(10);
    \draw(3)--(4);
    \draw(3)--(7);
    \draw(6)--(9);
    \draw(6)--(12);
    \draw(8)--(13);
    \draw(8)--(14);
    
    \draw(10) to [bend right = 5] (2);
    \draw(10) to [bend right = 35] (5);
    \path  (4)   edge[loop below] node[above]  {} (3);
    \path  (7)   edge[loop below] node[above]  {} (3);
    \path  (9)   edge[loop below] node[above]  {} (3);
    \path  (12)   edge[loop below] node[above]  {} (3);
    \path  (13)   edge[loop below] node[above]  {} (3);
    \path  (14)   edge[loop below] node[above]  {} (3);
\end{tikzpicture}
        \caption{\Cref{fig:MarkovChainFDR6} revisited. The numbers above the nodes illustrate the minimal distance to termination for each program state.}\label{fig:astfdr}
    \end{center}
\end{figure}

    \section{Appendix to \Cref{sec:fldr}}
\subsection{Predicate Definitions for the FLDR}
\label[appendix]{subsec:predsfldr}
Towards defining an invariant for the FLDR, we define some helpful predicates. These are simply short-hand notations for certain expressions and make the following invariant more intuitive. We define
\[
\bound{h}{c} \coloneqq 2^c - \sum_{j=0}^{c} h[j] \cdot 2^{c-j}
\]
which represents the first position in row $c$ of the binary tree that contains a number and \emph{not} an internal node. Then, we define
\[
\isThere{n, h, H}{c} \coloneqq [\H{\bound{h}{c}}{c} = n + 1]
\]
as a boolean ($0$ or $1$) indicating whether row $c$ contains the newly created die side $n+1$. We further define
\[
\prob{H}{i} \coloneqq \sum_{d,c\colon  \H{d}{c} = i} 2^{-c}
\]
which represents the probability to draw $i$ \emph{without} taking any re-rolls into account when rolling the new die side $n+1$. 
We define
\[
\probt{n, H}{i} \coloneqq \prob{H}{i} \cdot (1 + \sum_{j=1}^{\infty}(\prob{H}{n+1})^j)
\]
which represents the probability to draw $i$ \emph{with} taking re-rolls into account when rolling the new die side $n+1$. 
We define
\[
\prob{H}{i, j} \coloneqq \sum_{d,c\colon \H{d}{c} = i, c > j} 2^{-(c-j)}
\]
which represents the probability to draw $i$ \emph{when starting at level $j$ instead of at the root}, and without taking the re-roll into account when rolling $n+1$. 
Finally, we define
\[
\probt{n, H}{i, j} \coloneqq \prob{H}{i, j} + \prob{H}{n+1, j} \cdot \probt{n, H}{i}
\]
which represents the probability to draw $i$ \emph{when starting at level $j$ instead of at the root}, and with taking the re-roll into account when rolling $n+1$. 

\begin{example}
    \label{ex:predicates}
    Continuing \Cref{ex:bintree}, for our predicates we have the following. Note that $\bound{h}{0}$ and $\bound{h}{1}$ are not really relevant values, because rows $0$ and $1$ of our binary tree do not contain internal nodes.
    \begin{center}
        \begin{tabular}{cccccc}
            \begin{tabular}{c|c}
                $c$ & $\bound{h}{c}$\\
                \hline
                0 & $1$ \\
                1 & $2$ \\
                2 & $1$\\
                3 & $0$
            \end{tabular}
            &
            \begin{tabular}{c|c}
                $c$ & $\isThere{n, h, H}{c}$\\
                \hline
                0 & $0$\\
                1 & $0$\\
                2 & $1$\\
                3 & $0$
            \end{tabular}
            &
            \begin{tabular}{c|c}
                $i$ & $\prob{H}{i}$\\
                \hline
                1 & $\nicefrac{3}{8}$\\
                2 & $\nicefrac{2}{8}$\\
                3 & $\nicefrac{1}{8}$\\
                4 & $\nicefrac{2}{8}$
            \end{tabular}
            &
            \begin{tabular}{c|c}
                $i$ & $\probt{n, H}{i}$\\
                \hline
                1 & $\nicefrac{3}{6}$\\
                2 & $\nicefrac{2}{6}$\\
                3 & $\nicefrac{1}{6}$
            \end{tabular}
            \\
            \begin{tabular}{c|c}
                $i$ & $\prob{H}{i,1}$\\
                \hline
                1 & $\nicefrac{3}{4}$\\
                2 & $\nicefrac{2}{4}$\\
                3 & $\nicefrac{1}{4}$\\
                4 & $\nicefrac{2}{4}$
            \end{tabular}
            &
            \begin{tabular}{c|c}
                $i$ & $\prob{H}{i,2}$\\
                \hline
                1 & $\nicefrac{1}{2}$\\
                2 & $0$\\
                3 & $\nicefrac{1}{2}$\\
                4 & $0$
            \end{tabular}
            &
            \begin{tabular}{c|c}
                $i$ & $\probt{n, H}{i,1}$\\
                \hline
                1 & $1$\\
                2 & $\nicefrac{4}{6}$\\
                3 & $\nicefrac{2}{6}$ 
            \end{tabular}
            &
            \begin{tabular}{c|c}
                $i$ & $\probt{n, H}{i,2}$\\
                \hline
                1 & $\nicefrac{1}{2}$\\
                2 & $0$\\
                3 & $\nicefrac{1}{2}$\\
            \end{tabular}
            
        \end{tabular}
    \end{center}
\end{example}

\subsection{Invariant Definition of the FLDR}
\label[appendix]{subsec:InvariantFLDR}
We use short-hand notations that are defined in \Cref{subsec:predsfldr}. We construct a set of distributions $\Ifldr$ that we afterwards prove to be an inductive distributional invariant for the FLDR program:
\begin{align*}
    \Ifldr &= \{\mu \in \dists ~|~ \exists \numb{n} > 0\colon (\prosim{\mu}{\progvar{n}}{=}{\numb{n}} = 1 \tag{a}\\
    &\land \prosim{\mu}{0}{\leq}{\progvar{d} \leq \numb{n}} = 1\tag{b}\\
    &\land \prosim{\mu}{0}{\leq}{\progvar{c} \leq k} = 1\tag{c}\\
    & \land \forall i \in \nats\colon (1 \leq i \leq \numb{n} \Rightarrow\\
    & \sum_{\sigma\colon \H{\sigma(d)}{\sigma(c)} = i} \big(\mu(\sigma) + \sum_{j=0}^{k-1} \mu(\sigma[d/0, c/j]) \cdot \probt{\numb{n}, H}{i,j}\big) \leq \probt{\numb{n}, H}{i})\tag{d}\\
    &\land \forall \sigma,\sigma' \models \progvar{n} = \numb{n} \land 0 \leq \progvar{d} \leq \bound{h}{\progvar{c}}-1\colon\\
    & \phantom{MMMMMMMM} \sigma(\progvar{c}) = \sigma'(\progvar{c}) \Rightarrow \mu(\sigma) = \mu(\sigma')\tag{e}\\
    & \land \forall \sigma,\sigma' \models \progvar{n} = \numb{n} \land \bound{h}{\progvar{c}} + \isThere{\numb{n},h,H}{\progvar{c}} \leq \progvar{d} < \bound{h}{\progvar{c}} + h[\progvar{c}]\colon\\
    & \phantom{MMMMMMMM} \sigma(\progvar{c}) = \sigma'(\progvar{c}) \Rightarrow \mu(\sigma) = \mu(\sigma'))\}\tag{f}
\end{align*}

The invariant contains all distributions that satisfy 6 conditions $(a)$, $(b)$, $(c)$, $(d)$, $(e)$ and $(f)$. Intuitively, conditions $(a)$, $(b)$ and $(c)$ are very similar to the ones from the invariant of the FDR: Condition $(a)$ ensures that for each $\mu \in \Ifldr$ there exists a fixed value $\numb{n} \in \nats_{> 0}$ such that program states $\sigma$ with $\sigma(\progvar{n}) \ne \numb{n}$ are never reached. Conditions $(b)$ and $(c)$ then ensure that program states assigning too low or too great of values to the variables $d$ and $c$ are never reached (these are program states that are outside of the binary tree, e.g. in \Cref{ex:bintree}).
Conditions $(e)$ and $(f)$ are similar to the condition $(d)$ of the FDR invariant: they ensure uniform probabilities for program states \enquote{in the same row}. Condition $(e)$ ensures uniform probabilities for all the inner nodes of the binary tree that lie inside the same row while condition $(f)$ ensures uniform probabilities for all the output nodes that lie inside the same row (\emph{except} nodes labelled with the new die side $n+1$ --- these are always assigned probability $0$ since we restart upon hitting them). Condition $(d)$ has no analogue in the FDR invariant but is needed to make the invariant inductive: Condition $(d)$ ensures for all possible outputs $1 \leq i \leq n$, that, now \emph{and} in the future, not too much probability is assigned to output nodes labelled with $i$.

\begin{figure}[t]
    \begin{center}
        \begin{tikzpicture}[scale=0.65, node distance={30mm}, main/.style = {draw, shape = circle, minimum size = 0.2cm, inner sep=5pt},->]
            \node[main, fill = black](1) at (0,0) {};
            
            \node[main, fill= black](2) at (-2,-1) {};
            \node[main, fill= black](3) at (2,-1) {};
            
            \node[main, fill= black](4) at (-3,-2) {};
            \node[main](5) at (-1,-2) {$4$};
            \node[main](6) at (1.3,-2) {$2$};
            \node[main](7) at (2.8,-2) {$1$};
            
            \node[main](8) at (-3.8,-3) {$3$};
            \node[main](9) at (-2.2,-3) {$1$};
            \draw[red] (0,-1) ellipse (2.5cm and 0.6cm);
            \draw[red] (0,0) ellipse (0.5cm and 0.5cm);
            \draw[red] (-3,-2) ellipse (0.5cm and 0.5cm);
            \draw[orange] (2,-2) ellipse (1.4cm and 0.6cm);
            \draw[orange] (-3,-3) ellipse (1.4cm and 0.6cm);
            \node[red, align=left] at (-2.3, 0.3) {$\mu(\sigma) = \mu(\sigma')$};
            \node[orange, align=left] at (1, -3.3) {$\mu(\sigma) = \mu(\sigma')$};
            
            \draw(1) to (2);
            \draw(1) to (3);
            \draw(2) to (4);
            \draw(2) to (5);
            \draw(3) to (6);
            \draw(3) to (7);
            \draw(4) to (8);
            \draw(4) to (9);

            \node[main, fill = red](1) at (0+8,0) {};
            
            \node[main, fill= orange](2) at (-2+8,-1) {};
            \node[main, fill= black](3) at (2+8,-1) {};
            
            \node[main, fill= green](4) at (-3+8,-2) {};
            \node[main](5) at (-1+8,-2) {$4$};
            \node[main](6) at (1.3+8,-2) {$2$};
            \node[main, blue](7) at (2.8+8,-2) {$1$};
            
            \node[main](8) at (-3.8+8,-3) {$3$};
            \node[main, blue](9) at (-2.2+8,-3) {$1$};
            \node[align=left] at (2.4+8, -3.3) {For $i = 1$:\\$\frac{1}{2} \mu(\color{red}\bullet\color{black}) + \mu(\color{orange}\bullet\color{black}) + \frac{1}{2}\mu(\color{green}\bullet\color{black}) + \mu(\color{blue}\bullet\color{black}) \leq \frac{3}{6}$};
            
            \draw(1) to (2);
            \draw(1) to (3);
            \draw(2) to (4);
            \draw(2) to (5);
            \draw(3) to (6);
            \draw(3) to (7);
            \draw(4) to (8);
            \draw(4) to (9);
        \end{tikzpicture}
        \caption{Illustration of the conditions from the invariant for program $\progfldr$.}
        \label{fig:bintreeTwo}
    \end{center}
\end{figure}

\subsection{Proof of \Cref{lem:fldrinv} (FLDR Invariant)}
\label[appendix]{sec:fldrinv}
\begin{restatable}[FLDR Invariant]{lemma}{fldrinv}%
    \label{lem:fldrinv}%
    $\Ifldr$ is an inductive distributional invariant for $\progfldr$ and the above given set of initial distributions $\Mfldr$.
\end{restatable}
\begin{proof}
The proof is split into two parts: Proving containment of the initial distributions, and proving that $\Ifldr$ is inductive.
    
    \emph{Containment of the initial distributions.}
    We argue that $\Mfldr \subseteq \Ifldr$. Let $\mu \in \Mfldr$. Since $\mu$ is     a Dirac distribution, condition $(a)$ is satisfied by choosing the unique $\numb{n} \in \nats_{> 0}$ such that $\mu(d \mapsto 0, c \mapsto 0, n \mapsto \numb{n}) = 1$. Conditions $(e)$ and $(f)$ are satisfied because $\mu$ is a Dirac distribution and thus there do not even exist two different program states $\sigma, \sigma' \in \Sigma$ with $\mu(\sigma) > 0$ \emph{and} $\mu(\sigma') > 0$. Condition $(d)$ is satisfied since for all $1 \leq i \leq \numb{n}$, since $\mu$ is a Dirac distribution and therefore:
    \begin{align*}
        \sum_{\sigma\colon \H{\sigma(d)}{\sigma(c)} = i} \mu(\sigma) + \sum_{j=0}^{k-1} \mu(\sigma[d/0, c/j]) \cdot \probt{\numb{n}, H}{i,j} & = \probt{\numb{n}, H}{i,0}\\
 & = \probt{\numb{n}, H}{i}~.
    \end{align*}
    Condition $(c)$ is satisfied because for all $\numb{c} \in \mathbb{Z}\setminus\{0\}$, it holds for all $\sigma \in \Sigma$ with $\sigma(\progvar{c}) = \numb{c}$ that $\mu(\sigma) = 0$. Similarly, condition $(b)$ is satisfied because for all $\numb{d} \in \mathbb{Z}\setminus\{0\}$, it holds for all $\sigma \in \Sigma$ with $\sigma(\progvar{d}) = \numb{d}$ that $\mu(\sigma) = 0$.\\
    
    \emph{Inductiveness.}
    To show that $\Ifldr$ is inductive, we show that for an arbitrary $\mu \in \Ifldr$, it holds that $\den{\ifelseskip{\bool}{\loopbody}}(\mu) \in \Ifldr$ where $\loopbody$ is the loop body of the FLDR. As a first step, we apply the definition of the denotational semantics for the conditional case:
    \[
    \den{\ifelseskip{\bool}{\loopbody}}(\mu) = [\neg \bool] \cdot \mu + \den{\loopbody}([\bool] \cdot \mu)~.
    \]
    We derive the sub-expression $\den{\loopbody}([\bool] \cdot \mu)$ in \Cref{nextfldr}, whereby we repeatedly apply \Cref{lem:injassignments} to obtain simplified expressions for the post after injective assignments. We provide the final argumentation for inductiveness (in a similar manner as for the FDR), i.e. that $[\neg \bool] \cdot \mu + \den{\loopbody}([\bool] \cdot \mu) \in \Ifldr$ for an arbitrary $\mu \in \Ifldr$. We derived an expression for $\den{\loopbody}([\bool] \cdot \mu)$ graphically in \Cref{nextfldr}. Summarized, we obtain
    \begin{align*}
        & [\neg \bool] \cdot \mu + \den{\loopbody}([\bool] \cdot \mu)\\
        & = [d \geq \bound{h}{c}] \cdot \mu\\
        & + \mu \textnormal{ where } \mu(0,0) = |[\H{d}{c}=n+1] \psi|\\
        & + [\H{d}{c} \ne n+1] \psi
    \end{align*}
    
    for $\psi$ as defined in \Cref{nextfldr}. Similar as for the Fast Dice Roller, we further transform and simplify this expression by using that $\mu$ is not an arbitrary distribution, but $\mu \in \Ifldr$. Conditions $(b)$ and $(c)$ of the invariant ensure us that for all program states $\sigma$ with $\sigma(\progvar{c}) = \sigma(\progvar{d}) = 0$, we have that $\psi(\sigma) = 0$. Furthermore, we distinguish between program states satisfying $d < \bound{h}{c}$ and $d < \bound{h}{c}$. By that, we rewrite $[\neg \bool] \cdot \mu + \den{\loopbody}([\bool] \cdot \mu)$ into the following sum of three disjoint (sub)-distributions:
    
    \begin{align*}
        & [\neg \bool] \mu + \den{\loopbody}([\bool] \mu)\\
        & = \mu \textnormal{ where } \mu(0,0) = |[\H{d}{c}=n+1] \psi| \tag{$\mu_1$}\\ 
        & + [c \ne 0 \lor d \ne 0, d < \bound{h}{c}] \cdot \psi \tag{$\mu_2$}\\ 
        & + [c \ne 0 \lor d \ne 0, d \geq \bound{h}{c}, \H{d}{c} \ne n+1] \mu + \psi \tag{$\mu_3$}
    \end{align*}
    
    Now, we argue that conditions $(a), (b), (c), (d), (e)$ and $(f)$ of $\Ifldr$ are satisfied. In the following, let $\sigma, \sigma' \in \Sigma$ be program states and let $\mu' = \mu_1 + \mu_2 + \mu_3$ denote our derived sum.
    \begin{itemize}
        \item $(a)$: The argumentation is the same as in the proof of the FDR invariant.
        \item $(b)$ and $(c)$: The argumentation is similar as in the proof of the FDR invariant: For program states $\sigma$ and $\sigma'$, we have that $\frac{\sigma(\progvar{c})}{2} \in \mathbb{Z}$ iff $\frac{\sigma(\progvar{c})-1}{2} \notin \mathbb{Z}$. Furthermore, we have that $\frac{\sigma(\progvar{d})}{2} < \bound{h}{\sigma(\progvar{c})-1}$ iff $\frac{\sigma'(\progvar{d})}{2} < \bound{h}{\sigma'(\progvar{c})-1}$. Since $\mu$ itself satisfies conditions $(b)$ and $(c)$, we get that $\psi$ satisfies conditions $(b)$ and $(c)$ and furthermore also that the whole sum $\mu'$ satisfies conditions $(b)$ and $(c)$.
        
        \item $(d)$: For all $1 \leq i \leq \numb{n}$, we have that
        \begin{align*}
            & \sum_{\sigma\colon \H{\sigma(d)}{\sigma(c)} = i} \mu'(\sigma)\\
            = & \sum_{\sigma\colon \H{\sigma(d)}{\sigma(c)} = i} \mu(\sigma) + \sum_{\sigma | \H{\sigma(d)}{\sigma(c)} = i} \mu(\sigma[d/0, c/c-1]) \tag{1}
        \end{align*}
        and that
        \begin{align*}
            & \sum_{j=0}^{k-1} \mu'(\sigma[d/0, c/j]) \cdot \probt{\numb{n}, H}{i,j}\\
            = & \sum_{j=0}^{k-1} \mu(\sigma[d/0, c/j]) \cdot \probt{\numb{n}, H}{i,j} - \sum_{\sigma | \H{\sigma(d)}{\sigma(c)} = i} \mu(\sigma[d/0, c/c-1])\tag{2}
        \end{align*}
        yielding
        \begin{align*}
            & \sum_{\sigma\colon \H{\sigma(d)}{\sigma(c)} = i} \mu'(\sigma) + \sum_{j=0}^{k-1} \mu'(\sigma[d/0, c/j]) \cdot \probt{\numb{n}, H}{i,j}\\
            = & \sum_{\sigma\colon \H{\sigma(d)}{\sigma(c)} = i} \mu(\sigma) + \sum_{j=0}^{k-1} \mu(\sigma[d/0, c/j]) \cdot \probt{\numb{n}, H}{i,j} \tag{Inserting $\mu_1$ and $\mu_2$}\\
            \leq & ~\probt{H}{i}~. \tag{$\mu$ itself satisfies condition $(b)$}
        \end{align*}

        \item $(e)$ and $(f)$: Similar argumentation as for condition $(d)$ in the proof of the FDR invariant. Since $\mu$ itself satisfies conditions $(e)$ and $(f)$ and $\psi$ satisfies conditions $(e)$ and $(f)$ (by construction of the matrix $h$), we conclude that the whole sum $\mu'$ satisfies conditions $(e)$ and $(f)$.
    \end{itemize}
\end{proof}

\begin{figure}[t]
    \begin{align*}
        &\progAnno{[d < \bound{h}{c}] \cdot \mu}\\
        &\ass{c}{c+1} \SEMICLN \\
        &\progAnno{[d < \bound{h}{c-1}] \cdot \mu[c/c-1]}\\
        &\pchoice{\ass{d}{2d}}{\nicefrac 1 2}{\ass{d}{2d+1}} \SEMICLN \\
        &\progAnno{\psi \coloneqq \tfrac{1}{2} [\tfrac{d}{2} < \bound{h}{c{-}1}] \cdot \mu[d/\tfrac{d}{2}, c/c{-}1] + \tfrac{1}{2} [\tfrac{d-1}{2} < \bound{h}{c{-}1}] \cdot \mu[d/\tfrac{d-1}{2}, c/c{-}1]}\\
        &\IF{H[d,c] = n + 1}\\
        &\qquad \progAnno{[H[d,c]=n+1] \cdot \psi}\\
        &\qquad \ass{c}{0} \SEMICLN \ass{d}{0} \\
        &\qquad \progAnno{\mu \textnormal{ where } \mu(0,0) = |[H[d,c]=n+1] \psi|}\\
        &\prbracecl\\
        &\progAnno{[H[d,c] \ne n+1] \psi + \mu \textnormal{ where } \mu(0,0) = |[H[d,c]=n+1] \psi|}
    \end{align*}
    \caption{Derivation of {$\den{\loopbody}([\bool] \cdot \mu)$}  where $\loopbody$ is the loop body of the FLDR.}
    \label{nextfldr}
\end{figure}

\subsection{Proof of \Cref{thm:fldrpartial} (FLDR Correctness)}
\label{sec:fldrpartial}
\fldrtotal*
\begin{proof}
    \Cref{lem:fldrinv} proved $\Ifldr$ to be an invariant for the FLDR and our initial set of distributions $\Mfldr$. Thus, we can apply the \eqref{WhileP} Hoare rule which yields validity of the Hoare triple
    \[
        \triple{\Ifldr}{\prog}{\closOf{[\neg \bool] \Ifldr}}~.
    \]
    Since $\Mfldr \subseteq \Ifldr$, we have for all $\mu_0 \in \Mfldr$ that $\mu_0 \in \Ifldr$ and by definition of valid Hoare triples that $\mu \in \closOf{[\neg \bool] \Ifldr}$. By that $\mu \in \langle \neg \bool \rangle$ which means that $\mu(\sigma) = 0$ for all $\sigma$ with $\sigma \models \bool$, that means with $\sigma(\progvar{d}) < \bound{h}{\progvar{c}}$. Thus, condition $(d)$ of $\Ifldr$ shows that for all $1 \leq i \leq \numb{n}\colon$
    \begin{align*}
        & \sum_{\sigma\colon \H{\sigma(d)}{\sigma(c)} = i} \mu(\sigma) + \sum_{j=0}^{k-1} \mu(\sigma[d/0, c/j]) \cdot \probt{\numb{n}, H}{i,j} \leq \probt{\numb{n}, H}{i})\\
        = & \sum_{\sigma\colon \H{\sigma(d)}{\sigma(c)} = i} \mu(\sigma) \tag{$\mu \in \clos([\neg \bool] \cdot \Ifldr)$}\\
        \leq & ~\probt{\numb{n}, H}{i} \tag{Condition $(d)$ of $\Ifldr$}\\
        = & ~\frac{a_i}{m}~. \tag{Definition of $\probt{\numb{n}, H}{i}$}
    \end{align*}

    Now, we additionally assume the FLDR to terminate almost surely with respect to $\Ifldr$. This allows us to apply the \eqref{WhileT} Hoare rule which yields validity of the Hoare triple
    \[
        \triple{\Ifldr}{\prog}{\closOf{[\neg \bool] \Ifldr} \cap \Deltaspec{\geq \tau}\Sigma},
    \]
where $\tau = \inf_{\mu \in \Ifldr} \weight{\mu}$.
    Thus, it holds for the post that $\mu \in \Deltaspec{\geq \tau}\Sigma$. We have that $\tau = 1$ since our invariant only contains distributions with a mass of $1$ by construction, and therefore, $\weight{\mu} = 1$. Towards a contradiction, assume there is an $1 \leq i \leq \numb{n}$ such that $\sum_{\sigma\colon \H{\sigma(d)}{\sigma(c)} = i} \mu(\sigma) < \frac{a_i}{m}$. Then,
    \begin{align*}
        1 = \weight{\mu} = \sum_{1 \leq i \leq \numb{n}} \sum_{\sigma\colon \H{\sigma(d)}{\sigma(c)} = i} \mu(\sigma) < \sum_{1 \leq i \leq \numb{n}} \frac{a_i}{m} = 1
    \end{align*}
    which is a contradiction. Thus, we obtain our final result:
    \[
        \forall 1 \leq i \leq \numb{n}\colon \sum_{\sigma\colon \H{\sigma(d)}{\sigma(c)} = i} \mu(\sigma) = \frac{a_i}{m}~.
    \]
\end{proof}

\subsection{Almost Sure Termination of the FLDR}
\label[appendix]{app:fldrast}

We can prove our FLDR encoding to be AST on $\Ifldr$ in a similar way as for the FDR by using the sound and complete proof rule 3.2 by Majumdar et al. \cite{DBLP:journals/pacmpl/MajumdarS25}. In order to apply the rule, we consider the \emph{probabilistic control flow graph} underlying our FLDR program and a \emph{state}-based invariant that we derive from the \emph{distributional} invariant $\Ifldr$ presented in this paper. At the loop location, the state-based invariant contains precisely the states from the support of $\Ifldr$:
\[
    \textnormal{Inv} \coloneqq \underset{\mu \in \Ifldr}{\bigcup} \supp(\mu)~.
\]
For these states, we employ the supermartingale function
\begin{align*}
    v(\sigma) \coloneqq \begin{cases}
        0 & \text{if } \sigma \not \models \bool\\
        m & \text{otherwise} ~,
    \end{cases}
\end{align*}
where $\bool$ is the loop-guard of our FLDR program. This is a suitable choice, since for every $\numb{r} \in \mathbb{R}$, only finitely many inputs $a_1,\dots,a_n$ satisfy $m \leq \numb{r}$.
Finally, we employ the variant function
\begin{align*}
    u(\sigma) \coloneqq \begin{cases}
        0 & \text{if } \sigma \not \models \bool\\
        \left\lceil \log_2 (\frac{k}{\sigma(c)-\sigma(d)}) \right\rceil & \text{otherwise}~,
    \end{cases}
\end{align*}
which, intuitively, is suitable since it assigns to each program state an overapproximation of its minimal distance to termination (illustrated in \Cref{fig:astfldr}) and thus decreases with probability at least $0.5$ in each loop iteration.

\begin{figure}[t]
    \begin{center}
        \begin{tikzpicture}[node distance={20mm}, main/.style = {draw, shape = circle, minimum size = 0.1cm, inner sep=3pt},->,scale=0.85]
                \node[main, fill = black, label=above right:2](1) at (0,0) {};
                
                \node[main, fill= black, label=above right:2](2) at (-2,-0.7) {};
                \node[main, fill= black, label=above right:1](3) at (2,-0.7) {};
                
                \node[main, fill= black, label=above left:1](4) at (-3,-1.4) {};
                \node[main, label=above left:0](6) at (1.3,-1.4) {$2$};
                \node[main, label=above right:0](7) at (2.8,-1.4) {$1$};
                
                \node[main, label=above left:0](8) at (-3.8,-2.1) {$3$};
                \node[main, label=above right:0](9) at (-2.2,-2.1) {$1$};

                \draw(1) to (2);
                \draw(1) to (3);
                \draw(2) to (4);
                \draw[bend right](2) to (1);
                \draw(3) to (6);
                \draw(3) to (7);
                \draw(4) to (8);
                \draw(4) to (9);
            \end{tikzpicture}
        \caption{The binary tree from \Cref{fig:bintree} revisited. Node $4$ is omitted since upon reaching it, the FLDR program returns to its initial state \emph{in the same loop iteration}. The numbers above the nodes illustrate the minimal distance to termination for each program state.}\label{fig:astfldr}
    \end{center}
\end{figure}

\fi
%
%
\begin{credits}
\subsubsection{\ackname}
The authors acknowledge support from the DFG RTG 2236 (UnRAVeL) and the European Union's Horizon 2020 research and innovation program under the Marie Skłodowska-Curie grant agreement No.\ 101008233 (MISSION).
\subsubsection{\discintname}
The authors have no competing interests to declare that are relevant to the content of this article.
\end{credits}
%
%
%
\bibliographystyle{splncs04}
\bibliography{references}
\end{document}